\documentclass[sigconf, nonacm]{acmart}
\pdfoutput=1

\newcommand\vldbdoi{10.14778/3476249.3476268}

\newcommand\vldbavailabilityurl{}
 
\usepackage{ifthen}
\newboolean{fullversion}
\setboolean{fullversion}{true}

\usepackage{balance}

\usepackage{amsmath}
\usepackage{amsthm}
\usepackage{geometry}
\usepackage{xcolor}
\usepackage{enumitem}
\usepackage{adjustbox}
\usepackage{calc}
\usepackage{tikz}
\usepackage{subcaption}
\usepackage{xspace}
\usepackage{url}
\usepackage{ulem} 
\newcommand{\nop}[1]{}

\usepackage[ruled,vlined,noend]{algorithm2e}
\usepackage{listings}
\usepackage{caption}

\usepackage{hyperref}

\usepackage{verbatimbox}

\usetikzlibrary{decorations.pathreplacing,angles,quotes,positioning}

\newtheorem{theorem}{Theorem}
\newtheorem{definition}[theorem]{Definition}

\newtheorem{lemma}[theorem]{Lemma}
\newtheorem{example}[theorem]{Example}

\newcommand{\removespacetocaption}{}

\newcommand{\db}{\textbf{D}}
\newcommand{\type}[1]{\text{type}(#1)}
\newcommand{\schRel}{\textsf{Rels}}

\newcommand{\objectsRes}[1]{\objects_{#1}}

\newcommand{\Attr}[1]{\ensuremath{\text{Attr}(#1)}}

\newcommand{\eg}{\textit{e.g.}}

\newcommand{\trans}[1][i]{T_{#1}}

\newcommand{\transset}{{\mathcal{T}}}
\newcommand{\schedule}{s}
\newcommand{\objects}{\textbf{Tuples}}

\newcommand{\variables}{\mathbf{Var}}
\newcommand{\templ}[1][i]{\tau_{#1}}
\newcommand{\workload}{\mathcal{P}}

\newcommand{\tmap}{\mu}

\newcommand{\cg}[1]{CG(#1)}

\newcommand{\dependson}[3][\schedule]{#2 \rightarrow_{#1} #3}

\newcommand{\prefix}[2]{{\normalfont\textsf{prefix}}_{#2}(#1)}

\newcommand{\postfix}[2]{{\normalfont\textsf{postfix}}_{#2}(#1)}

\newcommand{\mvrc}{RC\xspace}
\newcommand{\MVRC}{\mvrc}

\newcommand{\readmvcom}{\MVRC\xspace}
\newcommand{\snapshot}{SI\xspace}

\newcommand{\parsnapshot}{{\normalfont\textsc{parallel snapshot isolation}}\xspace}

\newcommand{\coNP}{{\rm co{\sc np}}\xspace}

\newcounter{conditioncounter}

\newcommand{\remove}[1]{}%
\newcommand{\myR}{\ensuremath{\mathtt{R}}}
\newcommand{\myW}{\ensuremath{\mathtt{W}}}
\newcommand{\myUP}{\ensuremath{\mathtt{U}}}
\newcommand{\R}[2][i]{\myR_{#1}\mathtt{[#2]}}
\newcommand{\W}[2][i]{\myW_{#1}\mathtt{[#2]}}
\newcommand{\UP}[2][i]{\myUP_{#1}\mathtt{[#2]}}
\newcommand{\CT}[1][i]{\mathtt{C}_{#1}}

\newcommand{\ReadSet}[1]{\ensuremath{\text{ReadSet}(#1)}}
\newcommand{\WriteSet}[1]{\ensuremath{\text{WriteSet}(#1)}}

\newcommand{\x}{\mathtt{t}}
\newcommand{\y}{\mathtt{v}}
\newcommand{\z}{\mathtt{q}}

\newcommand{\myvs}{S}
\newcommand{\vt}{\mathtt{T}}
\newcommand{\myvv}{\mathtt{V}}

\newcommand{\vx}{\mathtt{X}}
\newcommand{\vy}{\mathtt{Y}}
\newcommand{\vz}{\mathtt{Z}}

\newcommand{\ListAttr}[1]{\ensuremath{\{\text{#1}\}}}
\newcommand{\Account}{\ensuremath{\text{Account}}\xspace}
\newcommand{\Savings}{\ensuremath{\text{Savings}}\xspace}
\newcommand{\Checking}{\ensuremath{\text{Checking}}\xspace}
\newcommand{\Conflict}{\ensuremath{\text{Conflict}}\xspace}

\newcommand{\Balance}{\ensuremath{\text{Balance}}\xspace}
\newcommand{\DepositChecking}{\ensuremath{\text{DepositChecking}}\xspace}
\newcommand{\TransactSavings}{\ensuremath{\text{TransactSavings}}\xspace}
\newcommand{\Amalgamate}{\ensuremath{\text{Amalgamate}}\xspace}

\newcommand{\tpcckv}{TPC-Ckv}
\newcommand{\Warehouse}{\ensuremath{\text{Warehouse}}\xspace}
\newcommand{\District}{\ensuremath{\text{District}}\xspace}
\newcommand{\Customer}{\ensuremath{\text{Customer}}\xspace}
\newcommand{\Order}{\ensuremath{\text{Order}}\xspace}
\newcommand{\OrderLine}{\ensuremath{\text{OrderLine}}\xspace}
\newcommand{\Stock}{\ensuremath{\text{Stock}}\xspace}

\newcommand{\sstart}{\textit{op}_0}

\renewcommand{\vx}{\textsf{X}}

\newcommand{\ptrans}{transaction template}   
\newcommand{\ptranss}{transaction templates} 
\newcommand{\Ptranss}{Transaction templates} 
\newcommand{\PTranss}{Transaction Templates} 
\newcommand{\shortptrans}{template}   
\newcommand{\shortptranss}{templates} 
\newcommand{\shortPTranss}{Templates} 

\newcommand{\canmu}{\bar{\mu}}

\ifthenelse{\boolean{fullversion}}{}{}

\begin{document}

\leftmargini 2.9ex

\renewcommand{\emph}[1]{{\it #1}}

\title{Robustness against Read Committed for Transaction Templates}

\author{Brecht Vandevoort}
\authornote{PhD Fellow of the Research Foundation -- Flanders (FWO)}
\email{brecht.vandevoort@uhasselt.be}
\affiliation{%
	\institution{\hspace{-5mm} Hasselt University and Transnational University of Limburg}
}

\author{Bas Ketsman}
\email{bas.ketsman@vub.be}
\affiliation{%
   \institution{Vrije Universiteit Brussel}
}

\author{Christoph Koch}
\email{christoph.koch@epfl.ch}
\affiliation{%
   \institution{\'Ecole Polytechnique F\'ed\'erale de Lausanne}
}

\author{Frank Neven}
\email{frank.neven@uhasselt.be}
\affiliation{%
   \institution{\hspace{-5mm} Hasselt University and Transnational University of Limburg}
}

\begin{abstract}
The isolation level Multiversion Read Committed (RC), offered by many database systems, is known to trade consistency for increased transaction throughput.
Sometimes, transaction workloads can be safely executed under RC obtaining the perfect isolation of serializability at the lower cost of RC. To identify such cases, we introduce an expressive model of transaction programs to better reason about the serializability of transactional workloads. We develop tractable algorithms to decide whether any possible schedule of a workload executed under RC is serializable (referred to as the robustness problem).
Our approach yields robust subsets that are larger than those identified by previous methods. We provide experimental evidence that workloads that are robust against RC can be evaluated faster under RC compared to stronger isolation levels. We discuss techniques for making workloads robust against RC by promoting selective read operations to updates. Depending on the scenario, the performance improvements can be considerable. Robustness testing and safely executing transactions under the lower isolation level RC can therefore provide a direct way to increase transaction throughput without changing DBMS internals.
\end{abstract}


\maketitle





\section{Introduction}
\label{sec:intro}

Relational database systems provide the ability to trade off isolation guarantees for improved performance by offering a variety of isolation levels, the highest being serializability, which guarantees what is considered to be perfect isolation. Executing transactions concurrently under weaker isolation levels is not without risk, as it can introduce certain anomalies. Sometimes, however, a set of transactions can be executed at an isolation level lower than serializability without introducing any anomalies. This is a desirable scenario: a lower isolation level, usually implementable with a cheaper concurrency control algorithm, gives us the stronger isolation guarantees of serializability for free.
This formal property is called robustness \cite{DBLP:conf/pods/Fekete05,DBLP:conf/concur/0002G16}: a set of transactions $\transset$ is called \emph{robust against a given isolation level} if every possible interleaving of the transactions in $\transset$ that is {allowed} under the specified isolation level is serializable.

There is a famous example that is part of database folklore: the
TPC-C benchmark \cite{TPCC} is robust against Snapshot Isolation (SI), so there is no need to run a stronger, and more expensive, concurrency control algorithm than SI if the workload is just TPC-C. This has
played a role in the incorrect choice of SI as the general concurrency control algorithm for isolation level Serializable in Oracle and PostgreSQL
(before version 9.1, cf.\ 
\cite{DBLP:journals/tods/FeketeLOOS05}).

Robustness is, fundamentally, a static property of workloads, rather than a property detectable online, while a concrete transaction schedule unfolds.
It involves the static or offline analysis of {\it transaction programs}\/ (code) to decide whether all possible interleavings of transactions (that is, instantiations of transaction programs) at runtime are guaranteed to be robust.
Robustness received quite a bit of attention in the literature.
Most existing work focuses on SI~\cite{Alomari:2008:CSP:1546682.1547288,DBLP:conf/cav/BeillahiBE19,DBLP:conf/pods/Fekete05,DBLP:journals/tods/FeketeLOOS05} or higher isolation levels~\cite{DBLP:conf/concur/BeillahiBE19,DBLP:conf/concur/0002G16,DBLP:conf/concur/Cerone0G15,cerone_et_al:LIPIcs:2017:7794}. 
It is particularly interesting to consider robustness against lower level isolation levels like multi-version Read Committed (referred to as \MVRC from now on). Indeed,  \MVRC is widely available, often the default in database systems (see, e.g., [4]), and is generally expected to have better throughput than stronger isolation levels. The work by Alomari and Fekete~\cite{DBLP:conf/aiccsa/AlomariF15} studies
robustness against \MVRC and proposes ways to preanalyse (and then modify) the code of a set of applications allowing to run transactions under RC while still guaranteeing that all executions are serializable.

 In general, robustness is a hopelessly undecidable property and
previous work has therefore only dealt with very simple models of workloads.
In this paper, we focus on pushing the frontier of the robustness problem for \MVRC. Robustness for arbitrary database application code would require the full sophistication of state-of-the-art program analysis and theorem provers and would not allow us to distill general guarantees that can lead to simpler analysis algorithms. 
We take a middle road, proposing a more expressive model of workloads than previously considered, which lets us still craft a complete and tractable decision procedure for robustness. We will show by examples -- specifically the TPC-C and SmallBank benchmarks -- that our model allows us to significantly expand the reach of robustness testing, yielding guaranteed serializability at the cost of just \MVRC isolation for a much larger class of workloads.

Our approach is centered on a novel characterization of robustness against \MVRC in the spirit of \cite{DBLP:conf/pods/Fekete05,DBLP:conf/pods/Ketsman0NV20} that improves over the sufficient condition presented in \cite{DBLP:conf/aiccsa/AlomariF15}, and on a formalization of transaction programs, called {\it transaction templates}, facilitating fine-grained reasoning for robustness against \MVRC.
Key aspects of our formalization are the following:
\begin{itemize}
\addtolength{\topsep}{-0.3ex}
\addtolength{\labelsep}{-0.3ex}
\item
Conceptually, {\it transaction templates}
are functions with parameters, and can, for instance, be derived from stored procedures inside a database system. Our abstraction generalizes transactions as usually studied in concurrency control research -- sequences of read and write operations -- by making the objects worked on variable, determined by input parameters. 
Such parameters are {\it typed} to add additional power to the analysis.

\item
We support {\it atomic updates} (that is, a read followed by a write of the same database object, to make a relative change to its value) allowing us to identify some workloads as robust that otherwise would not be.

\item
Furthermore, we model database objects read and written at the granularity of fields, rather than just entire tuples, decoupling conflicts further and allowing to recognize additional cases that would not be recognizable as robust on the tuple level.
\end{itemize} 

There are also a few restrictions to the model. We assume there is a fixed set of read-only attributes that cannot be updated and which are used to select tuples for update.
The most typical example of this are primary key values passed to transaction templates as parameters.
The inability to update primary keys is not an important restriction in many workloads, where keys, once assigned, never get changed, for regulatory or data integrity reasons. 
{In general, this restriction on updating and query-based selection of the same fields deals with the fact that the static, workload-level analysis of the phantom problem quickly yields undecidability.}
This makes our results inapplicable in certain scenarios, but these assumptions are necessary to make robustness decidable for such a versatile class of workloads, and it seems an acceptable trade-off to obtain such a result. It can be hoped that future work will push this decidability frontier even further.
These choices provide an interesting tradeoff between tractability and the ability to model and decide the robustness of more realistic workloads, as will be argued and illustrated throughout the remainder of the paper (as in Section~\ref{sec:example} for the SmallBank benchmark). 

{The sufficiency of our test for robustness, and the modification techniques we introduce to make programs robust, are practically applicable to programs that fit our model of a template. Programs that contain reads based on a predicate, rather than lookups on unchanging attributes such as a primary key, will need further techniques. Also, the necessity we prove for our decision procedure is only valid within our definition of RC isolation. In practice, it is possible for a set of programs running on a particular platform to always generate serializable executions even if they do not meet our test, in the case that the platform's implementation of RC doesn't allow all the possible interleavings which are covered by our definition of RC.}

\noindent    
\paragraph{In summary, the technical contributions of this paper are the following.}
\smallskip
\noindent
(1) We provide a full characterization for robustness against \MVRC for a workload of mere transactions instances (i.e., in the absence of variables).
The characterization 
forms a main building block for the robustness results for transaction templates mentioned in (3) below. Our result is interesting in its own right as there are not many isolation levels for which complete characterizations are known. The seminal paper by Fekete~\cite{DBLP:conf/pods/Fekete05} was the first to provide a characterisation for SI.
More recently, such characterisations where obtained for \MVRC and Read Uncommitted under a lock based rather than a multiversion semantics~\cite{DBLP:conf/pods/Ketsman0NV20}. In fact, it was shown that robustness against \MVRC under a lock-based semantics is \coNP-complete which should be contrasted with the polynomial time algorithm for multiversion Read Committed obtained in this paper.

\smallskip   
\noindent
(2) We introduce the formalism of transaction templates and formally define how associated sets of workloads are defined. The new formalism takes into account the type of variables in operations, makes atomic updates explicit, and models database objects read and written at the granularity of fields rather than tuples. 

\smallskip
\noindent
{(3) We obtain a polynomial time decision procedure for robustness against \MVRC for workloads of transactions defined by \ptranss{}. This is the first time a sound and complete algorithm for robustness against \MVRC on the level of transaction programs is obtained -- that is, an algorithm that does not produce false positives nor false negatives. In this way, we extend the work in \cite{DBLP:conf/aiccsa/AlomariF15} that is based on a sufficient condition for robustness in the sense that false positives never occur but false negatives can.} We discuss the implications of our algorithm in detail in Section~\ref{sec:detecting:robust:sets}.

\smallskip
\noindent
(4) We assess the effectiveness of our approach by analyzing 
SmallBank and \tpcckv{} (based on TPC-C) 
showing that 
we can identify robust subsets that are 
 larger than those identified by previous methods. 
Still, neither SmallBank nor \tpcckv{} is robust against \MVRC when taking all \ptranss{} into account. We 
consider ways to make \ptranss{} robust by promoting selective read operations to update operations and assess the effectiveness of this method on both benchmarks. With these (save) adaptations, both full benchmarks become robust for RC.

\smallskip
\noindent
(5) We experimentally demonstrate, using these two benchmarks and a well-known and unmodified DBMS, that our approach leads to practical performance improvements compared to when executed under SI or serializable SI, and compared to other robustness techniques for \MVRC~\cite{DBLP:conf/aiccsa/AlomariF15}, especially under higher contention.\footnote{In the absence of contention, the three techniques -- all sharing a common MVCC code base in the DBMS we use for experimentation -- essentially perform the same instructions and no improvements can be expected.}

\ifthenelse{\boolean{fullversion}}{
    \noindent    
    \paragraph{Outline.}
    We provide an extended example illustrating our results in Section~\ref{sec:example} and discuss related work in Section~\ref{sec:relwork}. We introduce the necessary definitions in Section~\ref{sec:defs}. We obtain a characterization for robustness against \MVRC in Section~\ref{sec:robustness-transactions}. In Section~\ref{sec:templates} and \ref{sec:robustness-templates}, we define templates and present our results for deciding robustness for transaction templates. 
    We discuss how to detect robust subsets in Section~\ref{sec:detecting:robust:sets}.
    We experimentally validate our approach in Section~\ref{sec:experiments} and conclude in Section~\ref{sec:concl}.
}{}
\ifthenelse{\boolean{fullversion}}
{}
{
    \smallskip
    \noindent
    Due to space constraints, proofs have been deferred to the online available full version of this
    paper~\cite{fullversion}. 
}


\section{Motivating Example}
\label{sec:example}

The SmallBank~\cite{Alomari:2008:CSP:1546682.1547288} schema consists of the tables
Account(\underline{Name}, CustomerID),
Savings(\underline{CustomerID}, Balance), and
Checking(\uline{Cus\-tomerID}, Balance) (key attributes are underlined).
The \Account table associates customer names with IDs. The other tables contain the balance (numeric value) of the savings and checking accounts of customers identified by their ID. 
The application code interacts with the database via the following transaction programs:
Balance($N$) returns the total balance (savings and checking) for a customer with name $N$.
%
DepositChecking($N$,$V$) makes a deposit of amount $V$ 
in the checking account of the customer with name $N$ (see Figure~\ref{fig:intro:depositchecking}).
%
TransactSavings($N$,$V$) makes a deposit or withdrawal $V$ on the savings account of the customer with name $N$.
%
Amalgamate($N_1$,$N_2$) transfers all the funds from customer $N_1$ to customer $N_2$.
%
Finally, WriteCheck($N$,$V$) writes a check $V$ against the account of the customer with name $N$, penalizing if overdrawing.

\noindent
\textbf{\textit{Formalisation of transactions templates.}}
Figure~\ref{fig:smallbank-abstract-syntax} displays the transaction templates for SmallBank.
\ifthenelse{\boolean{fullversion}}
{
The corresponding SQL code is provided in Figure~\ref{fig:smallbank:SQL} in the appendix.
}
{
}
A \ptrans{} consists of a sequence of read, write, and update operations to a tuple $\vx$ in a specific relation. For instance, $\R[]{\vx :\Account\{N,C\}\}}$ indicates that a read operation is performed to a tuple in relation $\Account$ on the attributes Name and CustomerID. We abbreviate the names of attributes by their first letter to save space. The set $\{N,C\}$ is the read set of the read operation.
Similarly, $\myW$ and $\myUP$ refer to write and update operations to tuples of a specific relation. Write operations have an associated write set while update operations contain a read set followed by a write set: e.g.,
$\UP[]{\vz :\DepositChecking\{C,B\}\{B\}\}}$ first reads the CustomerID and Balance of tuple $\vz$ and then writes to the attribute Balance.
All $\myR$-, $\myW$- and $\myUP$-operations always access exactly one tuple. A $\myUP$-operation is an atomic update that first reads the tuple and then writes to it. {Templates serve as abstractions of transaction programs and represent an infinite number of possible workloads. For instance, disregarding attribute sets, $\{ R[\x]R[\y]R[\z]\allowbreak U[\z], R[\x']R[\y']R[\z']U[\z'], R[\x]U[\z]\}$ is a workload consistent with the SmallBank templates as it contains two instantiations of Write\-Check and one instantiation of DepositChecking;  $\{R[\x]R[\y]\allowbreak R[\z]\allowbreak U[\z']\}$
with $\z\ne\z'$ is not a valid workload as the two final operations in WriteCheck should be on the same object as required by the formalization.} Typed variables effectively enforce domain constraints as we assume that variables that range over tuples of different relations can never be instantiated by the same value. For instance, in the \ptrans{} for \DepositChecking in Figure~\ref{fig:smallbank-abstract-syntax}, $\vx$ and $\vz$ can not be interpreted to be the same object.

\noindent
{\textbf{\textit{Detecting more robust subsets.}}}
Figure~\ref{fig:table:robust} gives an overview of the maximal robust subsets that are detected using our methods for the SmallBank and \tpcckv{} benchmarks (\tpcckv{} is discussed in Section~\ref{sec:robustness-templates} and the  templates are given in Figure~\ref{fig:tpcc-abstract-syntax}). 
Transaction templates are presented in abbreviated form (e.g., Bal refers to Balance). To assess the effect of the different features of our abstraction, we consider different settings: `Only R \& W' is the 
setting where updates are modeled through a read followed by a write
and where read and write sets always specify the whole set of attributes (that is, conflicts are considered on the level of entire tuples). This setting can be seen to correspond to the one of \cite{DBLP:conf/aiccsa/AlomariF15} that only reports the set \{Balance\} as robust against \MVRC. 

The setting `Atomic Updates' is the extension that models updates explicitly as atomic updates and already allows to detect relatively large robust sets compared to the `Only R \& W' setting. Indeed, for SmallBank 
\{Am,DC,TS\} is a robust subset indicating that any schedule using any number of instantiations of just these three \shortptranss{} that satisfies \MVRC is serializable! Also for \tpcckv{} larger robust subsets are detected.

Finally, `Attr conflicts' no longer requires read and write sets to specify all attributes (that is, conflicts are specified on the level of attributes). 
To illustrate its importance, 
consider the operations 
$\R[]{\vx: \Warehouse\ListAttr{W, Inf}}$ and 
$\UP[]{\vx: \Warehouse\ListAttr{W, YTD}\ListAttr{YTD}}$ coming from templates
NewOrder and Payment, respectively, in the \tpcckv{} benchmark as given in Figure~\ref{fig:tpcc-abstract-syntax}. An instantiation of these template mapping $\vx$ in both operations to the same tuple $\x$, does not result in a conflict  as the read set of the former is disjoint from the write set of the latter. However, considering conflicts on the granularity of tuples, that is, read and write sets refer to all attributes, does result in a conflict.
This difference in granularity has a profound effect for \tpcckv{} as can be seen
in the last row of Figure~\ref{fig:table:robust}: a robust subset of four templates (out of five!) is found: \{Del,Pay,NO,SL\}. For SmallBank there is no improvement 
as tuple conflicts always imply attribute conflicts for this benchmark as all attribute conflicts are based on the same Balance attributes in \Savings{} and \Checking{}. We explain in Section~\ref{sec:detecting:robust:sets} how robustness on attribute-level conflicts implies robustness on systems whose concurrency control subsystem works at the granularity of tuples.

We do not claim that all features in our abstraction are novel. The novelty lies in their combination to push the frontier of the robustness problem for \MVRC. Indeed, Figure~\ref{fig:table:robust} clearly shows that when taken together in an explicit formalisation, larger sets of transaction workloads can be safely determined to be robust. This is relevant since robust workloads can be executed under \MVRC at increased throughput compared to SI
 or serializable SI (see Section~\ref{sec:exp:robust_subset}).

Earlier work on robustness against \mvrc~\cite{DBLP:conf/aiccsa/AlomariF15} 
based on counterflow dependencies did not consider atomic updates or attribute-level conflicts, but can be extended to these settings. The robust subsets that are detected by these extension are given in Figure ~\ref{fig:table:relworkrobust}. A comparison with Figure~\ref{fig:table:robust}
reveals that although larger subsets are detected, 
our analysis still detects more and even larger robust subsets for both benchmarks, under both `Atomic Updates' and `Attr conflicts'.

\begin{figure}[t]
\begin{minipage}[c]{0.99\columnwidth}
%
\begin{verbbox}[\small]
DepositChecking(N,V):
  SELECT CustomerId INTO :X FROM Account WHERE Name=:N;
  UPDATE Checking SET Balance = Balance+:V
    WHERE CustomerId=:X;
  COMMIT;
\end{verbbox}

\begin{center}
{
    \theverbbox
}
\end{center}

\removespacetocaption

    \caption{SQL code for \DepositChecking.}
    \label{fig:intro:depositchecking}



\end{minipage}%

\vspace{.5em}
\begin{minipage}[c]{0.99\columnwidth}
    \centering\small

\begin{minipage}[t]{0.5\textwidth-2ex}
\Balance: 
\[
\begin{array}{l}
\R[]{\vx: \Account\ListAttr{N, C}}\\
\R[]{\vy: \Savings\ListAttr{C, B}}\\
\R[]{\vz: \Checking\ListAttr{C, B}}\\
\end{array}
\]
DepositChecking: 
\[
\begin{array}{l}
\R[]{\vx: \Account\ListAttr{N, C}}\\
\UP[]{\vz: \Checking\ListAttr{C, B}\ListAttr{B}}\\
\end{array}
\]
TransactSavings: 
\[
\begin{array}{l}
\R[]{\vx: \Account\ListAttr{N, C}}\\
\UP[]{\vy: \Savings\ListAttr{C, B}\ListAttr{B}}\\
\end{array}
\]
\end{minipage}%
\hfill%
\begin{minipage}[t]{0.50\textwidth-2ex}
Amalgamate: 
\[
\begin{array}{l}
\R[]{\vx_1: \Account\ListAttr{N, C}}\\
\R[]{\vx_2: \Account\ListAttr{N, C}}\\
\UP[]{\vy_1: \Savings\ListAttr{C, B}\ListAttr{B}}\\
\UP[]{\vz_1: \Checking\ListAttr{C, B}\ListAttr{B}}\\
\UP[]{\vz_2: \Checking\ListAttr{C, B}\ListAttr{B}}\\
\end{array}
\]
WriteCheck: 
\[
\begin{array}{l}
\R[]{\vx: \Account\ListAttr{N, C}}\\
\R[]{\vy: \Savings\ListAttr{C, B}}\\
\R[]{\vz: \Checking\ListAttr{C, B}}\\
\UP[]{\vz: \Checking\ListAttr{C, B}\ListAttr{B}}\\
\end{array}
\]
\end{minipage}

\removespacetocaption

    \caption{Transaction templates for SmallBank. 
    }
    \label{fig:smallbank-abstract-syntax}

\end{minipage}

\vspace{.3em}

\begin{center}
\begin{small}
\begin{tabular}{l||l|l}
               & SmallBank 
                                 & \tpcckv{} 
                                 \\ \hline\hline
Only R \& W &  \{Bal\} & \{OS, SL\} \\  \hline                               
Atomic Updates &   \{Am,DC,TS\}, &  \{Del,Pay,SL\}, \{NO, SL\},\\
& \{Bal,DC\}, \{Bal,TS\} & \{Pay, OS, SL\} \\
\hline
Attr conflicts &  \{Am,DC,TS\}, & \{Del,Pay,NO,SL\}, \\
&  \{Bal,DC\}, \{Bal,TS\} & \{Pay, OS, SL\} \\
\end{tabular}
\end{small}
\end{center}

\removespacetocaption

\caption{Robust subsets by analysis setting. 
 \label{fig:table:robust}}

\vspace{.3em}

\begin{center}
\begin{small}
\begin{tabular}{l||l|l}
            & SmallBank 
                              & \tpcckv{} 
                              \\ \hline\hline
Only R \& W &  \{Bal\} & \{OS, SL\} \\  \hline                               
Atomic Updates &   \{Am,DC,TS\}, \{Bal\} &  \{Del,Pay,SL\}, \{NO\}, \{OS,SL\}\\
\hline
Attr conflicts &  \{Am,DC,TS\}, \{Bal\} & \{Del,Pay,SL\}, \{Del,Pay,NO\} \\
& & \{OS,SL\} \\
\end{tabular}
\end{small}
\end{center}

\removespacetocaption

\caption{Detection of robust subsets based on counterflow dependencies ~\cite{DBLP:conf/aiccsa/AlomariF15} extended to updates and attribute level conflicts as introduced in this paper. 
\label{fig:table:relworkrobust}}

\end{figure}

\ifthenelse{\boolean{fullversion}}
{
We refer to the appendix for a detailed robustness analysis for each combination of \ptranss{}.
}
{
}


\section{Related Work}
\label{sec:relwork}

\ifthenelse{\boolean{fullversion}}{
\subsection{Static robustness checking on the application level}
}{\smallskip \noindent {\bf Static robustness checking on the application level.}}
Previous work on static robustness testing~\cite{DBLP:journals/tods/FeketeLOOS05,DBLP:conf/aiccsa/AlomariF15} for transaction programs is based on the following key insight: when a \emph{schedule} is not serializable, then the dependency graph constructed from that schedule contains a cycle satisfying a condition specific to the isolation level at hand: dangerous structure for \snapshot and the presence of a counterflow edge for \MVRC. This is extended to 
a workload of \emph{transaction programs} via a
 so-called static dependency graph, where each program is represented by a node, and there is a conflict edge from one program to another if there can be a schedule that gives rise to that conflict.  The absence of a cycle satisfying the condition specific to that isolation level guarantees robustness, while the presence of a cycle does not necessarily imply non-robustness. {\it We provide a formal approach to static robustness testing by making underlying assumptions more explicit within the formalism of transaction templates and obtain a decision procedure that is sound and complete for robustness testing 
 against \MVRC, allowing to detect larger subsets of transactions to be robust as exemplified in Section~\ref{sec:example}.}

\ifthenelse{\boolean{fullversion}}{
Cerone et al.~\cite{DBLP:conf/concur/Cerone0G15} provide a framework for uniformly specifying different isolation levels in a declarative way. 
A key assumption is \emph{atomic visibility} requiring that either all or none of the updates of each transaction are visible to other transactions.
Based on this framework, Bernardi and Gotsman~\cite{DBLP:conf/concur/0002G16} provide sufficient conditions for robustness against these isolation levels. 
Similar to the work of Fekete et al.~\cite{DBLP:journals/tods/FeketeLOOS05}, they first identify specific properties admitted by cycles in the dependency graphs of schedules that are allowed by the isolation level but not serializable. While analyzing robustness for a given set of program instances, they assume that each program instance is overestimated by three sets of tuples: those that might be read or written to by the program instance, and those that must be written to by the program instance. based on these sets, a static dependency graph is constructed. Analogous to~\cite{DBLP:journals/tods/FeketeLOOS05}, the absence of cycles with the property related to an isolation level in this graph guarantees that the set of program instances is robust against that isolation level. When analyzing robustness for a set of programs instead of specific program instances, a summary dependency graph is constructed, where each program is represented by a node. This graph is similar to static dependy graphs, but has additional information on the edges related to how the programs should be instantiated to create a specific conflict. This additional information reduces the number of workloads that are falsely identified to be non-robust.
Continuing on this line of work, Cerone and Gotsman~\cite{Cerone:2018:ASI:3184466.3152396} later studied the problem of robustness against \parsnapshot towards \snapshot (i.e., whether for a given workload every schedule allowed under \parsnapshot is allowed under \snapshot).
\emph{This declarative framework cannot be used to study robustness against \MVRC, as \MVRC does not admit \emph{atomic visibility}.}
}
{
 Other work studies robustness within a framework for uniformly specifying different isolation levels in a declarative way~\cite{DBLP:conf/concur/Cerone0G15,DBLP:conf/concur/0002G16,Cerone:2018:ASI:3184466.3152396}. A key assumption here is \emph{atomic visibility} requiring that either all or none of the updates of each transaction are visible to other transactions.
 \emph{This aims at higher isolation levels and cannot be used for \MVRC, as \MVRC does not admit \emph{atomic visibility}.}
 }

{
Executing a non-robust workload under a lower isolation level usually increases throughput at the cost of increasing the number of anomalies. 
To better quantify this tradeoff for a given workload, Fekete et al.~\cite{DBLP:journals/pvldb/FeketeGA09} presented a probabilistic model that predicts the rate of integrity violations depending on specific workload configurations. \emph{This line of work is orthogonal to robustness,
as a robust workload will increase throughput without introducing anomalies.
}}

\ifthenelse{\boolean{fullversion}}{

}{
Transaction chopping splits transactions into smaller pieces to obtain performance benefits and is correct if, for every serializable execution of the chopping, there exists an equivalent serializable execution of the original transactions~\cite{DBLP:journals/tods/ShashaLSV95}.  Cerone et al. \cite{Cerone:2018:ASI:3184466.3152396,DBLP:conf/wdag/CeroneGY15}  studied chopping under various isolation levels. \emph{Transaction chopping has no direct relationship with robustness testing against \MVRC.}
}

\ifthenelse{\boolean{fullversion}}{
\subsection{Making transactions robust}
}{
\noindent {\bf Making transactions robust.}}
{When a workload is not robust against an isolation level, robustness can be achieved by modifying the transaction programs~\cite{DBLP:journals/tods/FeketeLOOS05, DBLP:conf/aiccsa/AlomariF15, DBLP:conf/dasfaa/AlomariCFR08,Alomari:2008:CSP:1546682.1547288, DBLP:conf/aiccsa/Alomari13},
using an external lock manager~\cite{DBLP:conf/icde/AlomariFR09, DBLP:conf/aiccsa/AlomariF15, DBLP:conf/aiccsa/Alomari13},
allocating some programs to higher isolation levels~\cite{DBLP:conf/pods/Fekete05,DBLP:conf/dasfaa/AlomariCFR08},
or even a combination of these techniques~\cite{DBLP:conf/aiccsa/Alomari13}.

For \snapshot, two code modification techniques to remove dangerous structures from the static dependency graph have been studied~\cite{DBLP:journals/tods/FeketeLOOS05,DBLP:conf/dasfaa/AlomariCFR08,Alomari:2008:CSP:1546682.1547288, DBLP:conf/aiccsa/Alomari13, DBLP:conf/icde/AlomariFR09}: materialization and promotion. The materialization technique materializes conflicts between two potentially concurrent transactions by adding a new tuple to the database symbolizing this conflict and a write to this tuple is added to both transactions enforcing them to be non-concurrent. Alternatively, an external lock manager can be used~\cite{DBLP:conf/icde/AlomariFR09}. 
The promotion technique promotes a read operation by adding
an identity write to the same object. On some DBMS's, promotion can be implemented by changing the \textsc{\texttt{SELECT}} statement to \textsc{\texttt{SELECT ... FOR UPDATE}}. An alternative to 
code modification techniques is to allocate some transactions to S2PL instead of \snapshot~\cite{DBLP:conf/pods/Fekete05}. Alomari~\cite{DBLP:conf/aiccsa/Alomari13} considered a refinement that adds an additional write to each transaction running under S2PL.

For RC, Alomari and Fekete [5] consider lock materialization to avoid counterflow dependencies using two approaches: (1) in-database, by adding a write on a newly introduced tuple at the start of each transaction; and, (2) introducing an external lock manager outside of the database that application programs need to access.
In contrast, \emph{we employ a code modification technique based on promotion as for \snapshot changing certain read operations to updates. We provide a comparison in Section~\ref{sec:promotion}.
}

\ifthenelse{\boolean{fullversion}}{
\subsection{Other approaches}
}{
}

\ifthenelse{\boolean{fullversion}}{
{Instead of weakening the isolation level, other approaches to increasing transaction throughput without sacrificing ACID guarantees have been studied as well.}
Transactions can for example be split in smaller pieces to obtain performance benefits. However, this approach poses a new challenge, as not every serializable execution of these chopped transactions is necessarily equivalent to some serializable execution over the original transactions. A chopping of a set of transactions is correct if for every serializable execution of the chopping there exists an equivalent serializable execution of the original transactions. Shasha et al.~\cite{DBLP:journals/tods/ShashaLSV95} provide a graph based characterization for this correctness problem.
{This problem has been studied for different isolation levels such as \snapshot~\cite{Cerone:2018:ASI:3184466.3152396} and \parsnapshot~\cite{DBLP:conf/wdag/CeroneGY15} as well. However, in this case a correct chopping does not guarantee serializability. Instead, it verifies whether every execution of the chopped transactions allowed under an isolation level is equivalent to some execution of the original transactions allowed under this isolation level.}
\emph{Transaction chopping has no direct relationship with robustness testing against \MVRC.}

{Another approach is to modify existing algorithms that guarantee serializability. One notable example is a modification of S2PL where a transaction might release some locks before it acquired all locks. Wolfson~\cite{DBLP:journals/jal/Wolfson86,DBLP:journals/jal/Wolfson87} uses a sufficient condition to determine for a given workload at which point each lock acquired by a transaction might be released without risking anomalies.}

{When semantic knowledge of the transaction programs is available, it can be used to weaken the serializability requirement. Farrag and \"Ozsu~\cite{DBLP:journals/tods/FarragO89} use semantic knowledge of allowed interleavings between transactions to construct a new concurrency control algorithm that guarantees relatively consistent schedules. These relatively consistent schedules always preserve consistency, but do not necessarily guarantee serializability. Lu et al.~\cite{DBLP:journals/tkde/LuBL04} provide sufficient conditions under which every execution over a set of transactions under a given lock-based isolation level is semantically correct. A schedule is semantically correct if it has the same semantic effect as a serial schedule. As such, semantic correctness does not necessarily guarantee traditional serializability.
}
}{}

Many approaches to increase transaction throughput have been proposed: improved or novel pessimistic (cf., e.g., ~\cite{DBLP:journals/pvldb/YanC16,DBLP:journals/pvldb/TianHMS18,DBLP:conf/sigmod/RenFA16,DBLP:journals/pvldb/RenTA12,DBLP:journals/pvldb/JohnsonPA09})
or optimistic (cf., e.g., ~\cite{DBLP:conf/sigmod/SharmaSD18,DBLP:conf/sigmod/YuPSD16,DBLP:journals/pvldb/GuoCWQZ19,DBLP:journals/pvldb/HuangQKLS20,DBLP:journals/pvldb/LarsonBDFPZ11,DBLP:conf/sigmod/DiaconuFILMSVZ13,DBLP:conf/sigmod/BernsteinDDP15,DBLP:conf/cidr/BernsteinRD11,DBLP:journals/pvldb/SadoghiCBNR14,DBLP:conf/cloud/DingKDG15,DBLP:conf/sigmod/0001MK15,DBLP:conf/sigmod/KimWJP16,DBLP:conf/sigmod/LimKA17,DBLP:conf/sigmod/JonesAM10,DBLP:journals/pvldb/YuanWLDXBZ16})
algorithms, 
as well as approaches based on coordination avoidance (cf., e.g., ~\cite{DBLP:journals/pvldb/FaleiroAH17,DBLP:conf/sigmod/PrasaadCS20,DBLP:journals/pvldb/LuYCM20,DBLP:conf/sigmod/ShengTZP19,DBLP:journals/pvldb/FaleiroA15,DBLP:journals/pvldb/RenLA19,DBLP:conf/sigmod/ThomsonDWRSA12}). 
\emph{We do not compare to these as our focus lies on a technique that can be applied to standard DBMS's without any modifications to the database internals.}


\section{Definitions}
\label{sec:defs}

\smallskip
\noindent {\bf Databases.}
A \emph{relational schema} is a set $\schRel$ of relation names, and for
each $R\in\schRel$, $\Attr{R}$ is the finite set of associated attribute names.
For every relation $R\in\schRel$, we fix an infinite set $\objectsRes{R}$ of abstract objects called tuples.
We assume that $\objectsRes{R} \cap \objectsRes{S} = \emptyset$ for all  $R,S\in\schRel$ with $R\neq S$. We then denote by $\objects$ the set $\bigcup_{R\in\schRel} \objectsRes{R}$ of all possible tuples. By definition, for every $\x\in\objects$ there is a unique relation $R \in \schRel$ such that $\x\in\objectsRes{R}$. 
In that case, we say that $\x$ is of \emph{type} $R$ and denote the latter by $\type{\x}=R$. 
A \emph{database} $\db$ over schema $\schRel$ assigns to every relation name $R\in \schRel$ a finite set $R^{\db} \subset\objectsRes{R}$.

\smallskip
\noindent {\bf Transactions and Schedules.}
For a tuple $\x \in \objects$, we distinguish three operations $\R[]{\x}$, $\W[]{\x}$, and $\UP[]{\x}$ on $\x$, denoting that tuple $\x$ is read, written, or updated, respectively. {We say that the operation is on the tuple $\x$.}
The operation $\UP[]{\x}$ is an atomic update and should be viewed as an atomic sequence of a read of $\x$ followed by a write to $\x$.
We will use the following terminology: a \emph{read operation} is an $\R[]{\x}$ or a $\UP[]{\x}$, and a \emph{write operation} is a $\W[]{\x}$ or a $\UP[]{\x}$. Furthermore, an \myR-operation is an $\R[]{\x}$, a \myW-operation is a $\W[]{\x}$, and a \myUP-operation is a $\UP[]{\x}$.
We also assume a special \emph{commit} operation denoted $\CT[]$.
{To every operation $o$ on a tuple of type $R$, we associate the set of attributes
$\ReadSet{o}\subseteq\Attr{R}$ and $\WriteSet{o}\subseteq\Attr{R}$ containing, respectively, the set of attributes that $o$ reads from and writes to. When $o$ is a $\myR$-operation then $\WriteSet{o}=\emptyset$. Similarly, when $o$ is a $\myW$-operation then $\ReadSet{o}=\emptyset$. 
} 

A \emph{transaction} $\trans[]$ 
is a sequence of read and write operations 
followed by a commit. 
Formally, we model a transaction as a linear order $(\trans[],\leq_{\trans[]})$, where $\trans[]$ is the set of (read, write and commit) operations occurring in the transaction and $\leq_{\trans[]}$ encodes the ordering of the operations. As usual, we use $<_{\trans[]}$ to denote the strict ordering.

When considering a set $\transset$ of transactions, we assume that every transaction in the set has a unique id $i$ and write $\trans$ to make this id explicit. Similarly, to distinguish the operations of different transactions, we add this id as a subscript to the operation. {That is, we write $\W{\x}$, $\R{\x}$, and $\UP{\x}$ to denote a $\W[]{\x}$, $\R[]{\x}$, and $\UP[]{\x}$ occurring in transaction $\trans$; similarly $\CT[i]$ denotes the commit operation in transaction $T_i$. }
This convention is consistent with the literature (see, \eg\
\cite{DBLP:conf/sigmod/BerensonBGMOO95,DBLP:conf/pods/Fekete05}). 
To avoid ambiguity of notation, we assume that a transaction performs at most one write, one read, and one update per tuple.
The latter is a common assumption (see, \eg~\cite{DBLP:conf/pods/Fekete05}). All our results carry over to the more general setting in which multiple writes and reads per tuple are allowed.

A \emph{(multiversion) schedule} $\schedule$ over a set $\transset$ of transactions is a tuple $(O_\schedule, \leq_\schedule, {\ll_\schedule,} v_\schedule)$ where $O_\schedule$ is the set 
 containing all operations of transactions in $\transset$ as well as a special operation $\sstart$ conceptually writing the initial versions of all existing tuples, $\leq_\schedule$ encodes the ordering of these operations, {$\ll_\schedule$ is a \emph{version order} providing for each tuple $\x$ a total order over all write operations on $\x$ occurring in $\schedule$,} and $v_\schedule$ is a \emph{version function} mapping each read operation $a$ in $\schedule$ to either $\sstart$ or to a write\footnote{Recall that a write operation is either a $\W[]{x}$ or a $\UP[]{x}$.} operation different from $a$ in $\schedule$. 
We require that $\sstart \leq_\schedule a$ for every operation $a \in {O_\schedule}$, {$\sstart \ll_\schedule a$ for every write operation $a \in {O_\schedule}$}, and that $a <_{\trans[]} b$ implies $a <_\schedule b$ for every $\trans[] \in \transset$ and every $a,b \in \trans[]$. 
We furthermore require that for every read operation $a$, $v_\schedule(a) <_\schedule a$ and, if $v_\schedule(a) \neq \sstart$, then the operation $v_\schedule(a)$ is on the same tuple as $a$.
Intuitively, $\sstart$ indicates the start of the schedule, the order of operations in $s$ is consistent with the order of operations in every transaction $\trans[]\in\transset$, and the version function maps each read operation $a$ to the operation that wrote the version observed by $a$.
If $v_\schedule(a)$ is $\sstart$, then $a$ observes the initial version of this tuple.
{The version order $\ll_\schedule$ represents the order in which different versions of a tuple are installed in the database. For a pair of write operations on the same tuple, this version order does not necessarily coincide with $\leq_\schedule$. For example, under \mvrc the version order is based on the commit order instead.}

A schedule $\schedule$ is a \emph{single version schedule} if {$\ll_\schedule$ coincides with $\leq_\schedule$ and} every read operation always reads the last written version of the tuple. Formally, {for each pair of write operations $a$ and $b$ on the same tuple, $a \ll_\schedule b$ iff $a <_\schedule b$, and} for every read operation $a$ there is no write operation $c$ on 
the same tuple as $a$
with $v_\schedule(a) <_\schedule c  <_\schedule a$. 
A single version schedule over a set of transactions $\transset$ is \emph{single version 
serial
}  if its transactions are not interleaved with operations from other transactions. That is, for every $a,b,c \in {O_\schedule}$ with $a <_{\schedule}
b<_{\schedule} c$ and $a,c \in \trans[]$ implies $b \in \trans[]$ for every
$\trans[] \in \transset$. 

\ifthenelse{\boolean{fullversion}}{
{The absence of aborts in our definition of schedule is consistent with the common assumption~\cite{DBLP:conf/pods/Fekete05,DBLP:conf/concur/0002G16} that an underlying recovery mechanism will rollback aborted transactions. We only consider isolation levels that only read committed versions. Therefore there will never be cascading aborts.}
}{}

\smallskip
\noindent {\bf Conflict Serializability.}
{Let $a_j$ and $b_i$ be two operations on the same tuple from different transactions $\trans[j]$ and $\trans[i]$ in a {set of transactions $\transset$}. We then say that {$a_j$ is \emph{conflicting} with $b_i$} if:
\begin{itemize}
	\item \emph{(ww-conflict)} $\WriteSet{a_j} \cap \WriteSet{b_i} \neq \emptyset$; or,
	\item \emph{(wr-conflict)} $\WriteSet{a_j} \cap \ReadSet{b_i} \neq \emptyset$; or, 
    \item \emph{(rw-conflict)} $\ReadSet{a_j} \cap \WriteSet{b_i} \neq \emptyset$.
\end{itemize}}
{In this case, we also say that $a_j$ and $b_i$ are conflicting operations.}
Furthermore, commit operations and the special operation $\sstart$ never conflict with any other operation.
When $a_j$ and $b_i$ are conflicting operations in $\transset$, we say that $a_j$ \emph{depends on} $b_i$ in a schedule $\schedule$ over $\transset$, denoted $b_i \rightarrow_\schedule a_j$ if:\footnote{Throughout the paper, we adopt the following convention:  a $b$ operation can be understood as a `before' while an $a$ can be interpreted as an `after'.}
\begin{itemize}
    \item \emph{(ww-dependency)} {{$b_i$ is ww-conflicting with $a_j$} and $b_i \ll_{\schedule} a_j$}; or, 
    \item \emph{(wr-dependency)} {{$b_i$ is wr-conflicting with $a_j$} and $b_i = v_\schedule(a_j)$ or $b_i \ll_{\schedule} v_\schedule(a_j)$}; or, 
    \item \emph{(rw-antidependency)} {{$b_i$ is rw-conflicting with $a_j$} and $v_\schedule(b_i) \ll_{\schedule} a_j$.}
\end{itemize}

Intuitively, a ww-dependency from $b_i$ to $a_j$ implies that $a_j$ writes a version of a tuple {that is installed} after the version written by $b_i$.
A wr-dependency from $b_i$ to $a_j$ implies that $b_i$ either writes the version observed by $a_j$, or it writes a version that is {installed} before the version observed by $a_j$.
A rw-antidependency from $b_i$ to $a_j$ implies that $b_i$ observes a version {installed} before the version written by $a_j$.

Two schedules $\schedule$ and $\schedule'$ are \emph{conflict equivalent} if they are over the same set $\transset$ of transactions and for every pair of conflicting operations $a_j$ and $b_i$, $b_i \rightarrow_\schedule a_j$ iff $b_i \rightarrow_{\schedule'} a_j$.

\begin{definition}
    A schedule $\schedule$ is \emph{conflict serializable} if it is conflict equivalent to a single version serial schedule.
\end{definition}

A \emph{conflict graph} $\cg{\schedule}$ for schedule $\schedule$ over a set of transactions $\transset$ is the graph whose nodes are the transactions in $\transset$ and where there is an edge from $T_i$ to $T_j$ if $T_i$ has an operation $b_i$ that conflicts with an operation $a_j$ in $T_j$ and $b_i \rightarrow_\schedule a_j$.
The following is immediate from \cite{DBLP:books/cs/Papadimitriou86}:
\begin{theorem}\label{theo:not-conflict-serializable}
    A schedule $\schedule$ is conflict serializable iff the conflict graph for
    $\schedule$ is acyclic.
\end{theorem}


Our formalisation of transactions and conflict serializability is based on \cite{DBLP:conf/pods/Fekete05}, generalized to operations over attributes of tuples and extended with $\myUP$-operations that combine $\myR$- and $\myW$-operations into one atomic operation. These definitions are closely related to the formalization presented by Adya et al.~\cite{DBLP:conf/icde/AdyaLO00}, but we assume a total rather than a partial order over the operations in a schedule. 

\ifthenelse{\boolean{fullversion}}
{
We do not concern ourselves with predicate reads here, as our workload model, formalized in Section~\ref{sec:templates}, assumes that the selection of tuples is exclusively on attributes that do not get written.
(See the remarks on this restriction in Section~\ref{sec:intro}.)
Since predicate reads do not influence conflict serializability in our setting, we omit them in our notation to facilitate presentation. This assumption is in line with other work on robustness (e.g.~\cite{DBLP:conf/aiccsa/AlomariF15,DBLP:conf/pods/Fekete05,DBLP:conf/concur/0002G16}).
}{}


\smallskip
\noindent {\bf Multiversion Read Committed.}
Let $\schedule$ be a schedule for a set $\transset$ of transactions.
Then, $\schedule$ \emph{exhibits a dirty write}
     iff there are two {ww-conflicting} operations $a_j$ and $b_i$ in $\schedule$ on the same tuple $\x$
     with $a_j \in \trans[j]$, $b_i \in \trans[i]$ and $\trans[j] \neq \trans[i]$
     such that 
    \[b_i <_\schedule a_j <_\schedule \CT[i].\]
    That is, transaction $T_j$ writes to {an attribute of} a tuple that has
     been modified earlier by $T_i$, but $T_i$ has not yet issued a commit.

{For a schedule $\schedule$, the version order $\ll_\schedule$ corresponds to the commit order in $\schedule$ if for every pair of write operations $a_j \in \trans[j]$ and $b_i \in \trans[i]$, $b_i \ll_\schedule a_j$ iff $\CT[i] <_\schedule a_j$.}
We say that a schedule $\schedule$ is \emph{read-last-committed (RLC)} if {$\ll_\schedule$ corresponds to the commit order and} for every read operation $a_j$ in $\schedule$ on some tuple $\x$ the following holds:
\begin{itemize}
    \item $v_\schedule(a_j) = \sstart$ or $\CT[i] <_\schedule a_j$ with $v_\schedule(a_j) \in \trans[i]$; and
    \item there is no write\footnote{Recall that a write operation is either a $\myW$ or a $\myUP$-operation.} operation $c_k \in \trans[k]$ on $\x$ with $\CT[k] <_\schedule a_j$ and $v_\schedule(a_j) {\ll_\schedule} c_k$.
\end{itemize}
 That is, $a_j$ observes the most recent version of $\x$ {(according to the order of commits)} that is committed before $a_j$. Note in particular that a schedule cannot exhibit dirty reads, defined in the traditional way~\cite{DBLP:conf/sigmod/BerensonBGMOO95}, if it is read-last-committed.

\begin{definition} \label{def:isolationlevels}
A schedule is \emph{allowed under isolation level} read committed (RC) if
it is read-last-committed and does not exhibit dirty writes. 
\end{definition}

\smallskip
\noindent {\bf Robustness.} 
The robustness property~\cite{DBLP:conf/concur/0002G16} (also called \emph{acceptability} in~\cite{DBLP:conf/pods/Fekete05,DBLP:journals/tods/FeketeLOOS05}) guarantees serializability for all schedules of a given set of transactions for a given isolation level.

\begin{definition}[Robustness]
\label{def:robustness}
    A set\/ $\transset$ of transactions is \emph{robust} against RC
    if every schedule for\/ $\transset$ that is allowed under RC is
    conflict serializable.
\end{definition}

It is beneficial to model operations on the granularity of the attributes that are read or written.

\begin{example} 
Consider transactions $\trans[1]:\, \R[1]{\x\{a, b, c\}} \,  \W[1]{\y\{a\}} \, \CT[1]$ and $\trans[2]:\, \R[2]{\y\{b\}} \, \W[2]{\x\{a,b,d\}} \, \CT[2]$.
Here, for example, $\R[1]{\x\{a, b, c\}}$ is shorthand for operation
$\R[1]{\x}$ with read set $\{ a,b,c \}$.
The two operations on $\y$ are in conflict if the concurrency control system of the DBMS works with tuple-level objects, but are not conflicting on the level of attributes.
The workload is not robust on the tuple-level,
{as witnessed by the following schedule that is not (tuple-)conflict equivalent to a serial schedule $\schedule:\, \R[1]{\x\{a, b, c\}} \, \R[2]{\y\{b\}} \, \W[2]{\x\{a,b,d\}} \, \CT[2]\allowbreak \, \W[1]{\y\{a\}} \, \CT[1].$
However,} these two transactions are robust against \mvrc at attribute-level granularity.\footnote{This is under the reasonable assumption that a database system can read and/or update all attributes of a tuple in one atomic step.}
The order of the two operations on tuple $\x$ determines the order of the transactions in a conflict equivalent single version serial schedule.
{For example, the schedule $\schedule$ is conflict equivalent to the serial schedule $\trans[1] \cdot \trans[2]$.}
{So, modeling conflicts on the level of attributes allows to identify more workloads as robust.}
\hfill $\Box$
\end{example}


\section{Robustness for Transactions}
\label{sec:robustness-transactions}

Before introducing our formalisation for \ptranss{} in the next section, we start by studying the robustness problem for 
transactions. The results of the present section serve as a building block for our robustness algorithm for \ptranss{}. 

A naive way to decide the robustness property for a set of transactions is to iterate over all possible schedules allowed under \MVRC 
and verify that none violates conflict serializability. 
We show in the present section that only schedules with a very particular structure have to be considered which form the basis of a tractable decision procedure. We call these schedules \emph{multiversion split schedules}.

\begin{figure}

\begin{tikzpicture}
    [
        t0/.style={color=black,fill=black!10},
        t1/.style={color=red,fill=red!30},
        t2/.style={color=blue,fill=blue!30},
        t3/.style={color=orange,fill=orange!30},
        t4/.style={color=pink,fill=pink!30}
    ]

    \draw[t0](.5,1) rectangle (2.3,.7);         \draw[t0](5.8,1) rectangle (7.5,.7);
    \draw[t0](2.4,.6) rectangle (3.5,.3);
    \draw[t0](3.6,.2) rectangle (4.7,-.1);
    \draw[t0](4.8,-.2) rectangle (5.7,-.5);
    \node[anchor=west] at (0,.8) {$T_1$};
    \node[anchor=west] at (0,.4) {$T_2$};
    \node[anchor=west] at (0,.0) {$T_3$};
    \node[anchor=west] at (0,-.4) {$T_4$};

    \node(Aa) at (2.1,.85) {\footnotesize{$b_1$}};          \node(Ab) at (6.75,.85) {\footnotesize{$a_1$}};
    \node(Ba) at (3,.45) {\footnotesize{$b_2$}};
    \node(Bb) at (3.3,.45) {\footnotesize{$a_2$}};
    \node(Ca) at (4.15,.05) {\footnotesize{$a_3$}};
    \node(Cb) at (4.45,.05) {\footnotesize{$b_3$}};
    \node(Da) at (5.2,-.35) {\footnotesize{$a_4$}};
    \node(Db) at (5.5,-.35) {\footnotesize{$b_4$}};
   
   \path[->](Aa) edge[bend left] node{} (Bb);
   \path[->](Ba) edge[bend right] node{} (Ca);
   \path[->](Cb) edge[bend left] node{} (Da);
   \path[->](Db) edge[bend right] node{} (Ab);
   
    \draw[->](0,-.7) -- (7.5,-.7);
\end{tikzpicture}

\removespacetocaption

\caption{\label{fig:mvschedule} Multiversion split schedule for four transactions.}
\end{figure}
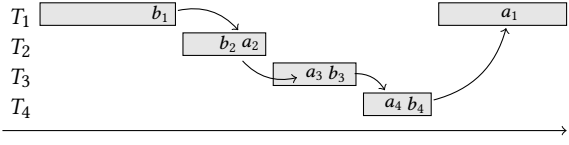

In the next definition, we represent conflicting operations from transactions in a set  $\transset$ 
as quadruples $(T_i, b_i, a_j, T_j)$ with $b_i$ and $a_j$ conflicting operations, and $T_i$ and $T_j$ their respective transactions in $\transset$. We call these quadruples \emph{conflict quadruples} for $\transset$. 
Further, for an operation $b\in \trans[]$, we denote by $\prefix{\trans[]}{b}$ the restriction of $\trans[]$ to all operations that are before or equal to $b$ according to $\leq_{\trans[]}$. Similarly, we denote by $\postfix{\trans[]}{b}$ the restriction of $\trans[]$ to all operations that are strictly after $b$ according to $\leq_{\trans[]}$. Throughout the paper, we interchangeably consider transactions both as linear orders as well as sequences.
Therefore, $\trans[]$ is then equal to the sequence $\prefix{\trans[]}{b}$ followed by $\postfix{\trans[]}{b}$ which we denote by $\prefix{\trans[]}{b}\cdot \postfix{\trans[]}{b}$ for every $b\in T$.

\begin{definition}[Multiversion split schedule]\label{def:mvsplitschedule}
Let $\transset$ be a set of transactions and $C = (T_1, b_1, a_2, T_2), (T_2, b_2, a_3, T_3), \ldots, (T_m,\allowbreak b_m,\allowbreak a_1,\allowbreak T_1)$ a sequence of  conflict quadruples for $\transset$ s.t.\ each transaction in $\transset$ occurs in at most two quadruples. A \emph{multiversion split schedule} for $\transset$ based on $C$ is a multiversion 
schedule that has the form
    $$ \prefix{\trans[1]}{b_1}\cdot \trans[2]\cdot \ldots \cdot \trans[m] \cdot \postfix{\trans[1]}{b_1}\cdot \trans[m+1] \cdot \ldots \cdot \trans[n],$$
    where
    {\begin{enumerate}
        \item \label{c:1} there is no write operation in $\prefix{\trans[1]}{b_1}$ ww-conflicting with a write operation in any of the transactions $T_2, \ldots, T_m$; 
        \item \label{c:2} $b_1 <_{\trans[1]} a_1$ or $b_{m}$ is rw-conflicting with $a_{1}$; and,
        \item \label{c:3} $b_{1}$ is rw-conflicting with $a_{2}$.
    \end{enumerate}}
    Furthermore, $\trans[m+1],\ldots,\trans[n]$ are the remaining transactions in $\transset$ (those not mentioned in $C$) in an arbitrary order.

\end{definition}

Figure~\ref{fig:mvschedule} depicts a schematic multiversion split schedule. The name stems from the fact that the schedule is obtained by splitting one transaction in two ($T_1$ at operation $b_1$ in Figure~\ref{fig:mvschedule}) and placing all other transactions in $C$ in between. The figure does not display the trailing transactions $\trans[m+1], \trans[m+2], \ldots$ and assumes $b_1<_{T_1} a_1$. 
{Intuitively, Condition~(\ref{c:1}) guarantees that $\schedule$ is allowed under \mvrc, while Condition~(\ref{c:2}) and (\ref{c:3}) ensure that $C$ corresponds to a cycle in $\cg{\schedule}$.}

The following theorem characterizes non-robustness in terms of the existence of a multiversion split schedule. The proof shows that for any 
counterexample schedule allowed under RC, a counterexample schedule  can be constructed that is a multiversion split schedule, and that, conversely, any multiversion split schedule $\schedule$ gives rise to a cycle in the conflict-graph $\cg{\schedule}$.
\begin{theorem} \label{theo:characterization:split-shedules}
    For a set of transactions $\transset$, this is equivalent:
    \begin{enumerate}
        \item \label{theo:char:split1} $\transset$ is not robust against \MVRC; 

        \item \label{theo:char:split3} there is a multiversion split schedule $\schedule$ for $\transset$ {based on some $C$}. 
    \end{enumerate}
\end{theorem}

\newcommand{\prefixfreegraph}{\text{prefix-conflict-free-graph}}

The above characterization for robustness against \mvrc leads to a polynomial time algorithm that cycles through all possible split schedules.
For this, we need to introduce the following notion. For a transaction $T_1$, an operation $b_1\in T_1$ and a set of transactions $\transset$ with $T_1\not\in\transset$, define $\prefixfreegraph(b_1,T_1, \transset)$ as the graph containing as nodes all transactions in $\transset$ that do not contain a {ww-}conflict with an operation in $\prefix{T_1}{b_1}$. Furthermore, there is an edge between two transactions $T_i$ and $T_j$  if $T_i$ has an operation that conflicts with an operation in $T_j$.

\begin{theorem}\label{theo:robustness-transactions:complexity}
Algorithm~\ref{alg:transaction_robustness} decides whether a  
set of transactions $\transset$ is robust against \MVRC in
time ${O}(\text{max}\{k.|\transset|^3, k^3.\ell\})$,
with $k$ the total number of operations in $\transset$ and $\ell$ the maximum number of operations in a transaction in $\transset$.
\end{theorem}

\begin{algorithm}[t]
\SetKwInOut{Input}{Input}
\SetKwInOut{Output}{Output}
\SetAlgoLined

\Input{\ Set of transactions $\transset$}
\Output{\ \emph{True} iff $\transset$ is robust against \MVRC} 
\BlankLine

\For{$T_1 \in \transset$}{
    \For{$b_1$ a {read operation} 
                   in $T_1$}{
        $G := \prefixfreegraph(b_1,T_1,\transset\setminus\{T_1\})$\;
        $TC := \text{reflexive-transitive-closure of } G$\;
        \For{$(T_2, T_m)$ in $TC$}{
                \For{$a_1 \in T_1$, $a_2 \in T_2$, $b_m \in T_m$}{

                \If{$a_1$ conflicts with $b_m$ {\bf and} {$b_1$ is rw-conflicting with $a_2$} 
                {\bf and} ($b_1 <_{T_1} a_1$ {\bf or} {$b_m$ is rw-conflicting with $a_1$} 
                )}
                {\Return{False}}}
            
        }    
    }  
}
\Return{True}

\caption{\label{alg:transaction_robustness} Deciding transaction robustness against RC.}
\end{algorithm}


\section{\PTranss{}}
\label{sec:templates}

\Ptranss{} are transactions where operations are defined over typed variables. 
Types of variables are relation names in $\schRel$ and indicate that variables can only be instantiated by tuples from the respective type.

We fix an infinite set of variables $\variables$ that is disjoint from $\objects$. Every variable $\vx\in\variables$ has an associated relation name in $\schRel$ as type that we denote by $\type{\vx}$.

\begin{definition}\label{def:template}
A \emph{\ptrans{}} $\tau$ is a transaction over $\variables$. 
{In addition, for every operation $o$ in $\tau$ over a variable $\vx$,
$\ReadSet{o}\subseteq \Attr{\type{\vx}}$ and $\WriteSet{o}\subseteq \Attr{\type{\vx}}$}.
\end{definition}
Notice that operations in \ptranss{} are defined over typed variables whereas they are over $\objects$ in transactions. Indeed, the \ptrans{} for Balance in Figure~\ref{fig:smallbank-abstract-syntax} contains a read operation $o=\R[]{\vx: \Account\{N,C\}}$.
{As explained in Section~\ref{sec:example}, the notation $\mathtt{\vx}: \Account\{N,C\}$ is a shorthand for $\type{\mathtt{\vx}}=\Account$ and $\ReadSet{o}=\{N,C\}$.}

Recall that we denote variables by capital letters $\mathtt{\vx},\mathtt{\vy},
\mathtt{\vz}$ and tuples by small letters $\x,\y$. 
A variable assignment $\mu$ is a mapping from $\variables$ to $\objects$
such that $\mu(\vx)\in \objects_{\type{\vx}}$.
By $\tmap (\templ[])$, we denote the transaction obtained by replacing each variable $\vx$ in $\templ[]$ with $\tmap (\vx)$. {A variable assignment
for a database $\db$ maps every variable to a tuple occurring in a relation in $\db$.}

A set of transactions $\transset$ is \emph{consistent} with a set of \ptranss{} $\workload$ {and 
database $\db$}, if for every transaction $\trans[]$ in $\transset$ there is a \ptrans{} $\tau\in \workload$ and a variable assignment $\mu_{\trans[]}$ 
for $\db$ such that $\mu_{\trans[]}(\tau) = \trans[]$.

Let $\workload$ be a set of \ptranss{} and $\db$ be a database.
Then, $\workload$ is \emph{robust against RC over $\db$} if for every set of transactions $\transset$ that is consistent with $\workload$ and $\db$, it holds that $\transset$ is robust against RC.

\begin{definition}[Robustness]\label{def:template_robustness}
 A set of \ptranss{} 
$\workload$ is \emph{robust} against RC if\/ $\workload$ is robust against RC for every database $\db$.
\end{definition}

{
\begin{example}\label{ex:parameterized_transactions}
Consider the 
database $\db$ over the SmallBank schema: 
$\Account^{\bf D}=\{a_1,a_2\}$; 
$\Savings^{\bf D}=\{s_1,s_2\}$; and \allowbreak
$\Checking^{\bf D}=\{c_1,\allowbreak c_2\}$.
For simplicity, we ignore read and write sets.
 Let $\transset_1 = \{ \R[]{a_1}\allowbreak \R[]{s_1}\R[]{c_1},\allowbreak \R[]{a_1}\R[]{a_2}\UP[]{s_1}\UP[]{c_1}\UP[]{c_2}\}$.
 Then $\transset_1$ is consistent with the SmallBank \ptranss{} and $\db$ as witnessed by the \ptranss{} Balance and Amalgamate, and the variable assignments $\mu_1 = \{\vx\to a_1,\vy\to s_1,\vz\to c_1\}$ and $\mu_2 = \{\vx_1\to a_1,\vx_2\to a_2,\allowbreak \vy_1\to s_1,\vy_2\to s_2,\vz_2\to c_2\}$.
    The set $\{\Balance, \Amalgamate\}$ is not robust against \mvrc, witnessed by $\db$ and $\transset_1$. Indeed, we can construct a multiversion split schedule over $\transset_1$:
	\begin{center}
	\begin{tabular}{l @{\hspace{.5em}} l @{\hspace{.5em}} l}
		$\trans[1]: \ \R[1]{a_1}\, \R[1]{s_1}$ & & $\R[1]{c_1} \, \CT[1]$\\
		$\trans[2]:$ & $\R[2]{a_1}\R[2]{a_2}\UP[2]{s_1}\UP[2]{c_1}\UP[2]{c_2} \CT[2]$ &
	\end{tabular}
	\end{center}
	\vspace{-\baselineskip}\hfill $\Box$
\end{example}
}


\section{Robustness for \shortPTranss{}}
\label{sec:robustness-templates}

\newcommand{\context}{context\xspace}
\newcommand{\mygraph}[1]{\textsf{Context}(#1)}
\newcommand{\mygraphext}[1]{\textsf{TGE}(#1)}
\newcommand{\mycomponent}[1]{\textsf{component}(#1)}
\newcommand{\mygraphb}[2][\workload]{\mathcal{G}_{#1}(#2)}
\newcommand{\prefixfreegraphB}{\text{pt-prefix-conflict-free-graph}}
\newcommand{\myvar}[1]{\text{Var}(#1)}

Algorithm~\ref{alg:transaction_robustness} cannot be applied directly to test robustness for transaction templates as there are infinitely many sets of transactions $\transset$ consistent with a given set of transaction templates $\workload$. 
We use a different approach that resembles Algorithm~\ref{alg:transaction_robustness} but that operates directly over \ptranss{}.

Central to the proposed algorithm (Algorithm~\ref{alg:ptime:template}) is a generalization of conflicting operations: For \ptranss{} $\tau_i$ and $\tau_j$ in $\workload$, we say that an operation $o_i \in \tau_i$ is \emph{potentially conflicting} with an operation $o_j \in \tau_j$ if {$o_i$ and $o_j$ are operations over a variable of the same type, and at least one of the following holds:}
\begin{itemize}
    \item $\WriteSet{o_i}\cap \WriteSet{o_j}\neq \emptyset$ (potentially ww-conflicting);
    \item $\WriteSet{o_i}\cap \ReadSet{o_j}\neq \emptyset$ (potentially wr-conflicting); or
    \item $\ReadSet{o_i}\cap \WriteSet{o_j}\neq \emptyset$ (potentially rw-conflicting).
\end{itemize}

Intuitively, potentially conflicting operations lead to conflicting operations when the variables of these operations are mapped to the same tuple by a variable assignment. Analogously to conflicting quadruples over a set of transactions {as in Definition~\ref{def:mvsplitschedule}}, we consider \emph{potentially conflicting quadruples} $(\tau_i, o_i, p_j,  \tau_j)$ over a set of \ptranss{} $\workload$ with $\tau_i, \tau_j \in \workload$, and $o_i \in \tau_i$ an operation that is potentially conflicting with an operation $p_j \in \tau_j$.
A sequence of potentially conflicting quadruples 
$D = (\tau_1, o_1, \allowbreak p_2,\allowbreak \tau_2), \ldots, (\tau_m, o_m, p_1, \tau_1)$ over $\workload$  (in which multiple occurrences of the same \ptrans{} are allowed) {induces} a sequence of conflicting quadruples $C = (\trans[1], b_1, a_2, \trans[2]),\allowbreak \ldots,\allowbreak (\trans[m], b_m, a_1, \trans[1])$ by  applying a variable mapping $\mu_i$ to {each $\tau_i$ in $D$}.
{We call such a set of variable mappings simply
a \emph{variable mapping} for $D$, denoted $\bar{\mu}$, and write $\bar{\mu}(D) = C$.}

{A basic insight is the following: if there is a multiversion split schedule $\schedule$ for some $C$ over a set of transactions $\transset$ consistent with $\workload$ and a database $\db$
with the properties of Definition~\ref{def:mvsplitschedule}, {then there is a sequence of potentially conflicting quadruples $D$ such that $\bar{\mu}(D) = C$ for some $\bar{\mu}$. }
The approach followed in Algorithm~\ref{alg:ptime:template} is then to enumerate sequences $D$ together with mappings $\bar{\mu}$ in search of $\bar{\mu}(D)$ for which the conditions of Definition~\ref{def:mvsplitschedule} are satisfied.} 
If a counterexample exists, Algorithm~\ref{alg:ptime:template} needs at most three tuples per type to construct a counterexample.
We encode this choice for each variable by assigning the numbers $1$ to $3$ to specific operations.

To cycle through all possible sequences $D$, Algorithm~\ref{alg:ptime:template} iterates over the possible split \ptranss{} $\tau_1\in \workload$ and its possible operations $o_1,p_1\in\tau_1$, and relies on a graph referred to as 
$\prefixfreegraphB({o_1},\allowbreak {p_1}, h, \tau_1, \workload)$.
Here, $h\in\{1,2\}$ signals that the prefix and suffix of the split of $\tau_1$ use the same tuple of each type when $h=1$ and that the suffix uses the second tuple of each type when $h=2$.
The graph has as nodes the quadruples $(\tau, o, i, j)$ with $\tau \in \workload$, $o \in \tau$, $i\in\{1,2,3\}$ and $j \in \{\text{in},\text{out}\}$.
Here, $i\in\{1,2,3\}$ encodes the tuple assigned to $o$ in $\tau$.
There will be two types of edges: (1) inner edges $(\tau,o,i,\text{in}) \to (\tau, p,i',\text{out})$ that stay within the same transaction $\tau$ and indicate how the chosen tuple version changes (or stays the same) from $c_i$ for $o$ to $c_{i'}$ for $p$; and (2) outer edges $(\tau,o,i,\text{out}) \to (\tau', p,i',\text{in})$ between different occurrences of \ptranss{} encoding a potentially conflicting quadruple $(\tau,o,p,\tau')$ and maintaining information on the chosen tuple as well.

More formally, a quadruple node $(\tau, o, i, j)$ in the graph satisfies the following properties:
\begin{itemize}
\item[(a)] $i=1$ implies that there is no operation $o_1' \in \prefix{\tau_1}{o_1}$ over the same variable as ${o}_1$ in
$\tau_1$ s.t.\ $o_1'$ is potentially ww-conflicting with an operation over the same variable as $o$ in $\tau$.
\item[(b)] $i=h$ implies that there is no operation $o_1' \in \prefix{\tau_1}{o_1}$ over the same variable as ${p}_1$ in
    $\tau_1$ s.t.\ $o_1'$ is potentially ww-conflicting with an operation over the same variable as $o$ in $\tau$. 
\end{itemize}
Conditions (a) and (b) on the nodes ensure that condition (1) of
Definition~\ref{def:mvsplitschedule} is always guaranteed for all possible variable mappings that are consistent with the particular choice of tuples.  
Furthermore, two nodes $(\tau, o, i, j)$ and $(\tau', o', i', j')$ are connected by a directed edge if either
\begin{itemize}
    \item[$(\dagger)$] $\tau = \tau'$, $j = \text{in}$, $j' = \text{out}$, and
    if $o$ and $o'$ are over the same variable in $\tau$, then $i = i'$
        (i.e., remain within the same transaction and
        change the chosen tuple version only
        when $o$ and $o'$ are not over the same variable); or,
    \item[$(\ddagger)$] $j = \text{out}$, $j' = \text{in}$, $i = i'$ and $o$ and $o'$ are
        potentially conflicting (i.e., the analogy of
        $b$ and $a$ for consecutive transactions in a
        split schedule, but here defined for transaction templates).
\end{itemize}

\begin{algorithm}[t]
\SetKwInOut{Input}{Input}
\SetKwInOut{Output}{Output}
\SetAlgoLined

\Input{\ Set of \ptranss{} $\workload$}
\Output{\ \emph{True} iff $\workload$ is robust against \MVRC} 
\BlankLine

\For{$\tau_1 \in \workload$}{
    \For{$o_1$ an operation in $\tau_1$, $(p_1, i) \in \tau_1\times \{1,2\}$}{
        $G := \prefixfreegraphB(o_1, p_1, i, \tau_1, \workload)$\;
        $TC := \text{transitive-closure of } G$\; 
        \For{$\tau_2, \tau_m$ in $\workload$}{
            \For{ $p_2 \in \tau_2$, $o_m \in
            \tau_m$}{

                \If{$p_1$ is potentially conflicting with $o_m$ {\bf and} $o_1$ is potentially
                rw-conflicting with $p_2$ {\bf and} ($o_1 <_{\tau_1} p_1$ {\bf
                or} $o_m$ is potentially rw-conflicting with $p_1$) {\bf and} $\langle(\tau_2,
                p_2,1,\text{in}), (\tau_m, o_m, i, \text{out})\rangle$ in $TC$ }                                 {\Return{False}}}
            
        }    
    }  
}
\Return{True}
\caption{\label{alg:ptime:template} Deciding \shortptrans{} robustness against RC.}
\end{algorithm}

\begin{theorem}\label{theo:ptime:noconstraints}
    Algorithm~\ref{alg:ptime:template} 
    decides whether a set of \ptranss{}
     $\workload$ is robust against RC in time 
    $\mathcal{O}(k^4.\ell)$
    with $k$ the total number of operations in $\workload$ and $\ell$ the
    maximum number of operations in transactions of $\workload$.
\end{theorem}

\ifthenelse{\boolean{fullversion}}{
    \begin{example}
    We illustrate Algorithm~\ref{alg:ptime:template} via an example run on the SmallBank benchmark.
    Take $\tau_1 = \Amalgamate$, $o_1 = \R[]{\vz}$, $p_1 = \UP[]{\vz}$ and $i = 1$. Then we can choose $\tau_2 = \tau_m = \DepositChecking$ and $p_2 = o_m = \UP[]{\vz}$ to satisfy all conditions in Algorithm~\ref{alg:ptime:template}.
    Notice in particular that there is an edge from $(\tau_2, \UP[]{\vz}, 1, \textit{in})$ to $(\tau_m, \UP[]{\vz}, 1, \textit{out})$ in $\prefixfreegraphB(o_1, p_1, 1, \tau_1, \workload)$ by $(\dagger)$. The corresponding counterexample is the multiversion split schedule based on instances of respectively \Amalgamate and \DepositChecking over the same customer, where \Amalgamate is split after $\R[]{\vz}$.
    \hfill$\Box$
    \end{example}
}{}


\section{Detecting robust sets}
\label{sec:detecting:robust:sets}
\label{sec:isolation-level:attribute-granularity}

{As every subset of a robust set of templates is robust as well, maximal robust subsets of a workload $\workload$ can be detected by running Algorithm~\ref{alg:ptime:template} first on $\workload$ itself and if necessary on smaller subsets. Even though there are exponentially many possible subsets, $\workload$ is expected to be small and robustness tests can be performed in a static and offline analysis phase.}

Algorithm~\ref{alg:ptime:template} allows for a complete characterization of robustness at attribute-level granularity. We discuss the ramifications of using these results with a DBMS whose concurrency control subsystem works at the granularity of tuples.
In this case, an \MVRC implementation isolates more strongly than actually needed to assure serializability on workloads our techniques identify as robust.\footnote{For instance, \MVRC in PostgreSQL acquires locks on the granularity of tuples rather than attributes  -- see Section~\ref{sec:postgres} for a more detailed description.}

There are two ways to employ our decision procedures in this case.
The first is to simply coarsen the workload model to the tuple-level by
setting, for each operation, the read and write sets to all the attributes of 
the tuple. In this way, our algorithms give a correct and complete answer 
at tuple-level granularity. {As discussed in Section~\ref{sec:example}, the row `Atomic updates' in Figure~\ref{fig:table:robust} indicates which sets are robust under this method for SmallBank and \tpcckv{} and, how this improves over considering only reads and writes.}

The second approach is to simply work with the attribute-level model
and accept that the DBMS is more conservative than necessary.
When our algorithm determines a workload to be robust, that workload will still be robust on systems that assure \MVRC with tuple-level database objects, for the simple reason that every conflict on the granularity of attributes implies a conflict on the granularity of tuples.  As a result, every schedule that can be created by these systems is allowed under our definition of \MVRC.
However, when our algorithm determines a workload \emph{not}\/ to be robust,
they may be too conservative: they might do so by identifying a complete set of counterexample schedules, none of which may actually be allowed under RC at the granularity of tuples.
Thus, \emph{our attribute-level algorithm technically provides only a 
{sufficient}\/ rather than a complete condition for robustness on such systems.}
The second technique nevertheless strictly dominates the first {on SmallBank and \tpcckv{} (as can be seen in the row `Attr conflicts' in Figure~\ref{fig:table:robust})}, even when the DBMS works with tuple-level objects. It detects all the robust cases of the former approach, plus potentially additional ones that can only be found by attribute-level analysis, but which still are robust on a DBMS with tuple-level concurrency control.
The latter approach leads to a more general observation with practical value: our algorithm provides a sufficient condition to guarantee serializability for every implementation that can only generate a subset of the schedules allowed by \mvrc.


\section{Experiments}
\label{sec:experiments}

We discussed the effectiveness of our approach in detecting larger robust subsets in comparison with \cite{DBLP:conf/aiccsa/AlomariF15} 
at the end of Section~\ref{sec:example}. 
  We focus here on how robustness can improve transaction throughput.

\begin{figure*}[t]
    \centering \footnotesize

\begin{minipage}[t]{0.37\textwidth-2ex}    
NewOrder:
\[
\begin{array}{l}
\R[]{\vx: \Warehouse\ListAttr{W, Inf}}\\
\UP[]{\vy: \District\ListAttr{W, D, Inf, N}\ListAttr{N}}\\
\R[]{\vz: \Customer\ListAttr{W, D, C, Inf}}\\
\W[]{\myvs: \Order\ListAttr{W, D O, C, Sta}}\\
\UP[]{\vt_1: \Stock\ListAttr{W, I, Qua}\ListAttr{Qua}}\\
\W[]{\myvv_1: \OrderLine\ListAttr{W, D, O, OL, I, Del, Qua}}\\
\UP[]{\vt_2: \Stock\ListAttr{W, I, Qua}\ListAttr{Qua}}\\
\W[]{\myvv_2: \OrderLine\ListAttr{W, D, O, OL, I, Del, Qua}}\\
\end{array}
\]
\end{minipage}
\begin{minipage}[t]{0.63\textwidth}
\begin{minipage}[t]{\textwidth/2-1ex}    
Delivery:
\[
\begin{array}{l}
\UP[]{\myvs: \Order\ListAttr{W, D, O}\ListAttr{Sta}}\\
\UP[]{\myvv_1: \OrderLine\ListAttr{W, D, O, OL, Del}\ListAttr{Del}}\\
\UP[]{\myvv_2: \OrderLine\ListAttr{W, D, O, OL, Del}\ListAttr{Del}}\\
\UP[]{\vz: \Customer\ListAttr{W, D, C, Bal}\ListAttr{Bal}}\\
\end{array}
\]
\end{minipage}
\begin{minipage}[t]{\textwidth/2-1ex}    
OrderStatus:
\[
\begin{array}{l}
\R[]{\vz: \Customer\ListAttr{W, D, C, Inf, Bal}}\\
\R[]{\myvs: \Order\ListAttr{W, D, O, C, Sta}}\\
\R[]{\myvv_1: \OrderLine\ListAttr{W, D, O, OL, I, Del, Qua}}\\
\R[]{\myvv_2: \OrderLine\ListAttr{W, D, O, OL, I, Del, Qua}}\\
\end{array}
\]
\end{minipage}

\medskip

\begin{minipage}[t]{\textwidth/2-1ex}    
Payment:
\[
\begin{array}{l}
\UP[]{\vx: \Warehouse\ListAttr{W, YTD}\ListAttr{YTD}}\\
\UP[]{\vy: \District\ListAttr{W, D, YTD}\ListAttr{YTD}}\\
\UP[]{\vz: \Customer\ListAttr{W, D, C, Bal}\ListAttr{Bal}}\\
\end{array}
\]
\end{minipage}
\begin{minipage}[t]{\textwidth/2-1ex}
StockLevel:
\[
\begin{array}{l}
\R[]{\vt: \Stock\ListAttr{W, I, Qua}}\\
\end{array}
\]
\end{minipage}
\end{minipage}

\removespacetocaption

    \caption{Abstraction for the \tpcckv{} \ptranss{}. Attribute names are abbreviated.}
    \label{fig:tpcc-abstract-syntax}
\end{figure*}

\subsection{Experimental Setup}
\label{exp:setup}

\subsubsection{PostgreSQL}
\label{sec:postgres}
We used PostgreSQL 12.4 as a database engine. PostgreSQL uses multiversion concurrency control to implement three different isolation levels: Read Committed (\mvrc), Snapshot isolation (SI), and Serializable Snapshot Isolation (SSI)~\cite{DBLP:journals/tods/FeketeLOOS05}.\footnote{In PostgreSQL 12.4, these three isolation levels are referred to as Read Committed, Repeatable Read, and Serializable, respectively.}
When reading a tuple, RC reads the last committed version before this read operation, whereas SI and SSI see the last committed version before the start of the transaction.
All isolation levels use write locks to avoid dirty writes. If a transaction $\trans[1]$ wants to update a tuple that has been changed by a concurrent transaction $\trans[2]$, transaction $\trans[1]$ will wait for $\trans[2]$ to commit or abort, thereby releasing the write lock, before proceeding. Notice that in specific cases this can lead to deadlocks, e.g. when multiple concurrent transactions try to update the same set of tuples. Under SI and SSI, $\trans[1]$ will abort if $\trans[2]$ successfully committed, according to the first-updater-wins principle.
When using SSI, PostgreSQL will furthermore monitor for possible conditions that could lead to unserializable executions, and possibly abort transactions to preserve serializability.

The database system runs on a server with two 2.3 GHz Xeon Gold 6140 CPUs with 18 cores each, 192 GB RAM, and a 200 GB SSD local disk. A separate machine is used to issue the transactional workload to the database system through a low-latency connection.
The workload is supplied via
a number of concurrently running client processes. Each client sequentially runs transactions from randomly selected \ptranss{} through this same database connection. When a transaction is aborted, the client immediately retries this transaction with the same parameters, until it eventually commits. For 60 seconds, we measure the number of transactions that are committed and the number of aborts. Each experiment is repeated 5 times. The graphs in this section show both the average values, as well as 95\% confidence intervals.

\subsubsection{SmallBank benchmark (see Section~\ref{sec:example})}
\label{sec:smallbank}
The database is populated with 18000 randomly generated accounts with corresponding checking and savings accounts -- as in earlier experiments on the SmallBank benchmark in~\cite{DBLP:conf/aiccsa/AlomariF15, Alomari:2008:CSP:1546682.1547288}.
Each client uses a uniform distribution when selecting one of the possible \shortptranss{}. To select which accounts to address, we considered two approaches. The first approach fixes a small subset of accounts, referred to as the \emph{hotspot}, and a probability for an account selected for use in a transaction to be from among the hotspot accounts, referred to as the \emph{hotspot probability}. Within the hotspot, each account has an equal probability of being selected. The second approach uses a Zipfian distribution to randomly select accounts~\cite{DBLP:conf/sigmod/GraySEBW94}.

\subsubsection{\tpcckv{} benchmark}
\label{sec:tpcc}
The second benchmark is based on the TPC-C benchmark~\cite{TPCC}. We modified the schema and \shortptranss{} to turn all predicate reads into key-based accesses.
The schema consists of six relations:
\sloppy
\begin{itemize}
    \item Warehouse(\underline{WarehouseID}, Info, YTD),
    \item District(\underline{WarehouseID}, \underline{DistrictID}, Info, YTD, NextOrderID),
    \item Customer(\underline{WarehouseID}, \underline{DistrictID}, \underline{CustID}, Info, Balance),
    \item Order(\underline{WarehouseID}, \underline{DistrictID}, \underline{OrderID}, CustID, Status),
    \item OrderLine(\underline{WarehouseID}, \underline{DistrictID}, \underline{OrderID}, \underline{OrderLineID}, ItemID, DeliveryInfo, Quantity), and
    \item Stock(\underline{WarehouseID}, \underline{ItemID}, Quantity).
\end{itemize}
\fussy

We focus on five different \ptranss{}:
%
\begin{enumerate}
    \item NewOrder($W$, $D$, $C$, $I_1$, $Q_1$, $I_2$, $Q_2$, \ldots): creates a new order for the customer identified by $(W,D,C)$. The id for this order is obtained by increasing the NextOrderID attribute of the District tuple identified by $(W,D)$ by one. Each order consists of a number of items $I_1, I_2, \ldots$ with respectively quantities $Q_1, Q_2, \ldots$. For each of these items, a new OrderLine tuple is created and the related stock quantity is decreased.
    \item Payment($W$, $D$, $C$, $A$): represents a customer identified by $(W,D,\allowbreak C)$ paying an amount $A$. This payment is reflected in the database by increasing the balance of this customer by $A$. This amount is furthermore added to the YearToDate (YTD) income of both the related warehouse and district.
    \item OrderStatus($W$, $D$, $C$, $O$): requests information about the current status of the order identified by $(W,D,O)$. This \ptrans{} collects information of the customer identified by $(W,D,C)$ who created the order, the order itself, and the different OrderLine tuples related to this order.
    \item Delivery($W$, $D$, $C$, $O$): delivers the order represented by $(W,D,O)$. The status of the order is updated, as well as the DeliveryInfo attribute of each OrderLine tuple related to this order. The total price of the order is deduced from the balance of the customer who made this order, identified by $(W,D,C)$.
    \item StockLevel($W$, $I$): returns the current stock level of item $I$ in $W$. 
\end{enumerate}

An abstraction of each \ptrans{} is given in Figure~\ref{fig:tpcc-abstract-syntax}.
The experiments adhere to the requirements of the official TPC-C benchmark~\cite{TPCC}, with a scaling factor of 25 warehouses. This means that the database is populated with 25 warehouses, where each warehouse is assigned 10 districts and 100000 different stock items. Each district has 3000 customers, and each customer initially has 10 orders. We randomly assign between 5 and 15 orderlines per order
{(Figure~\ref{fig:tpcc-abstract-syntax} shows only two orderlines per order to simplify presentation).}
Each client uses a uniform distribution when selecting one of the possible \shortptranss{}.
When generating parameters for each transaction, we remain consistent with the TPC-C benchmark. That is, we use a uniform distribution to randomly pick warehouses, districts, items within a warehouse and orders for a customer. 
Customers within a district are non-uniformly selected based on a Zipfian distribution. We consider one additional setting where warehouses are selected according to a Zipfian distribution.

\subsection{Robust workloads}
\label{sec:exp:robust_subset}

\begin{figure}[t]
    \centering
    \begin{subfigure}[t]{0.495\linewidth}
        \centering
        \includegraphics[width=\linewidth]{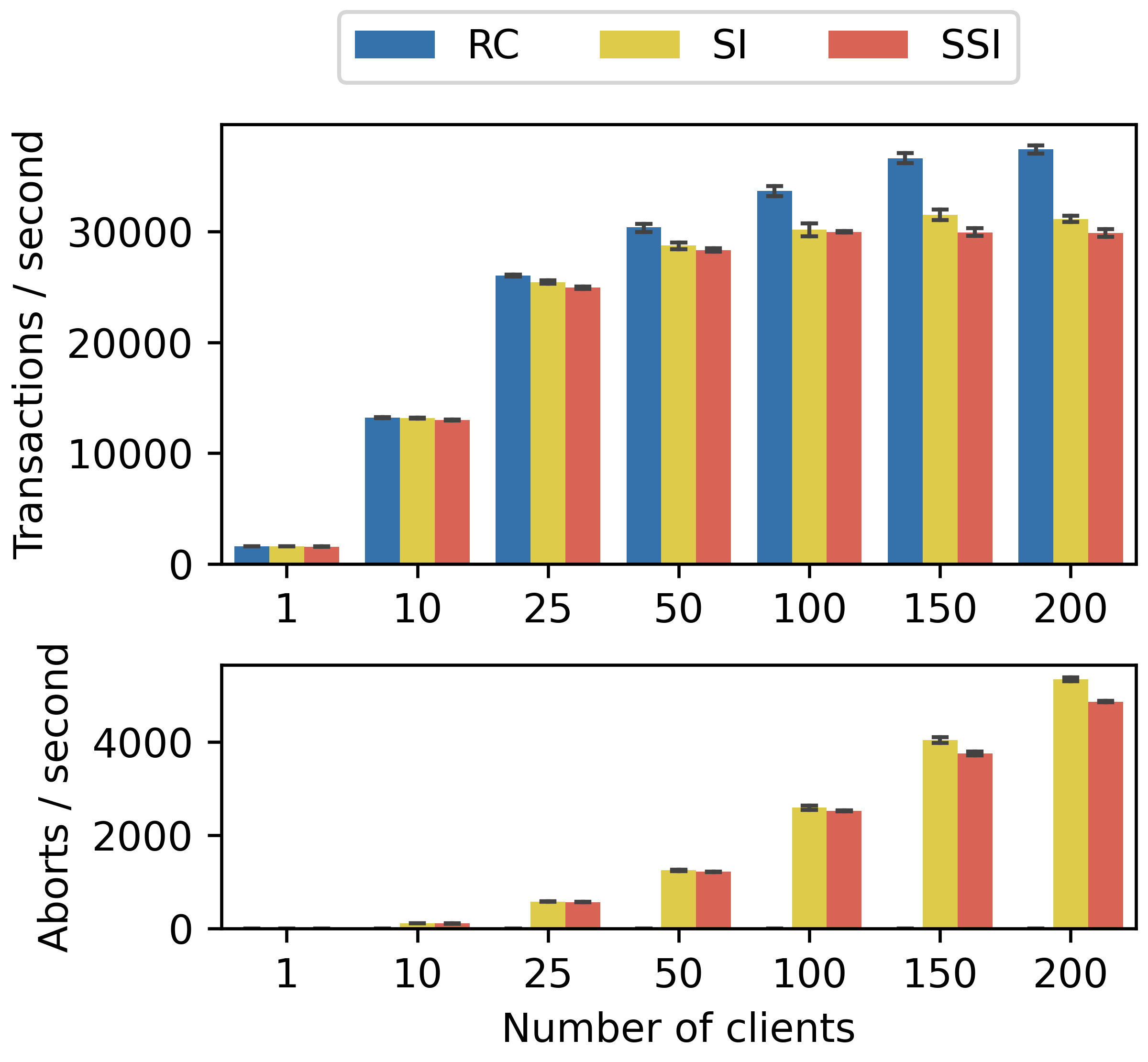}
        \caption{Robust subset.}
        \label{fig:exp1}
    \end{subfigure}
    \hfill
    \begin{subfigure}[t]{0.495\linewidth}
        \centering
        \includegraphics[width=\linewidth]{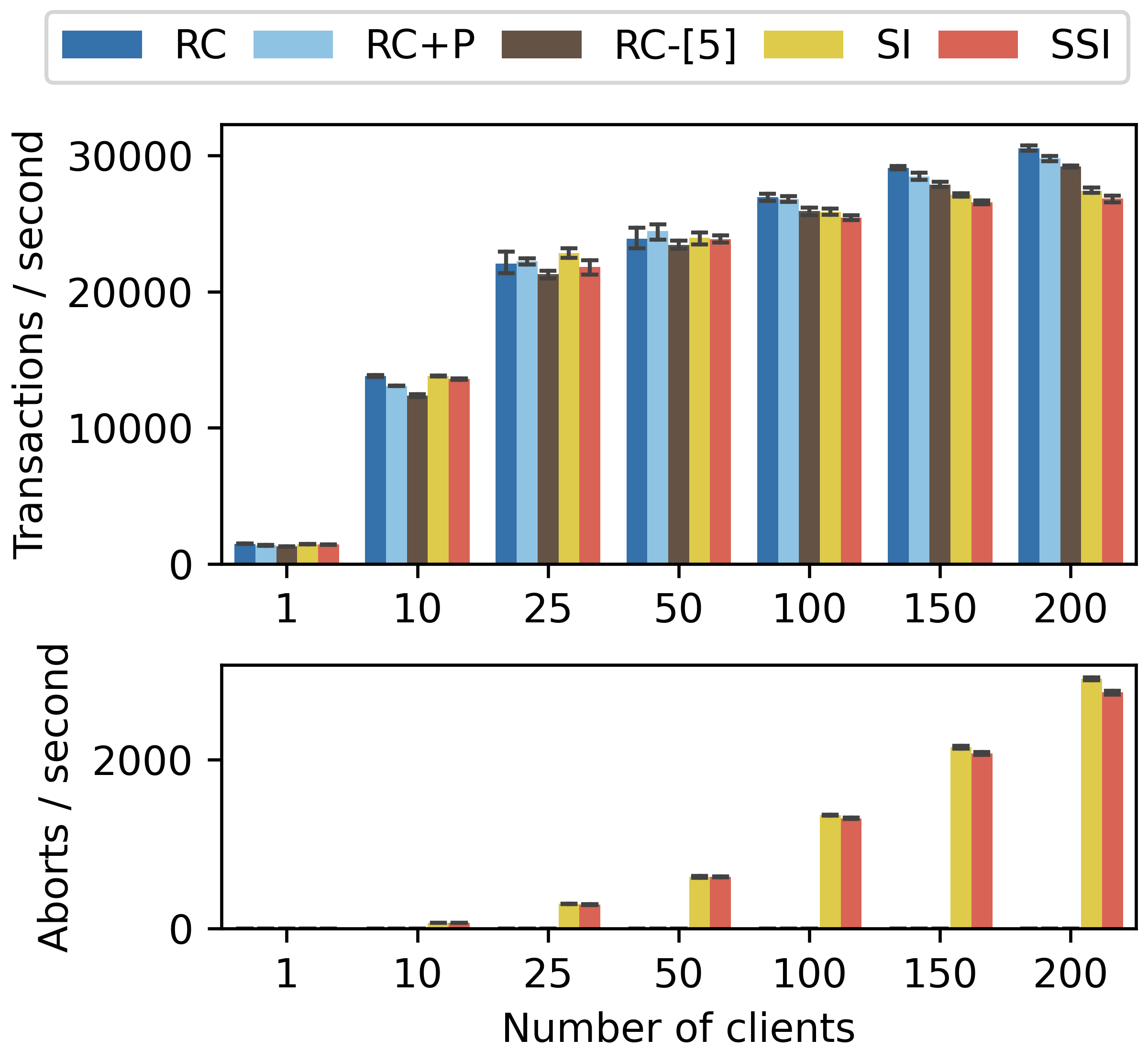}
        \caption{{Complete benchmark.}}
        \label{fig:exp8}
    \end{subfigure}
    \caption{Throughput and abort rate
per number of concurrent clients for (a subset of) SmallBank. The hotspot consists of 1000 accounts with a hotspot probability of 90\%.}
    \label{fig:exp1and8}

\end{figure}

\begin{figure*}[t]
    \centering
    \begin{subfigure}[t]{0.24\linewidth}
        \centering
        \includegraphics[width=.95\linewidth]{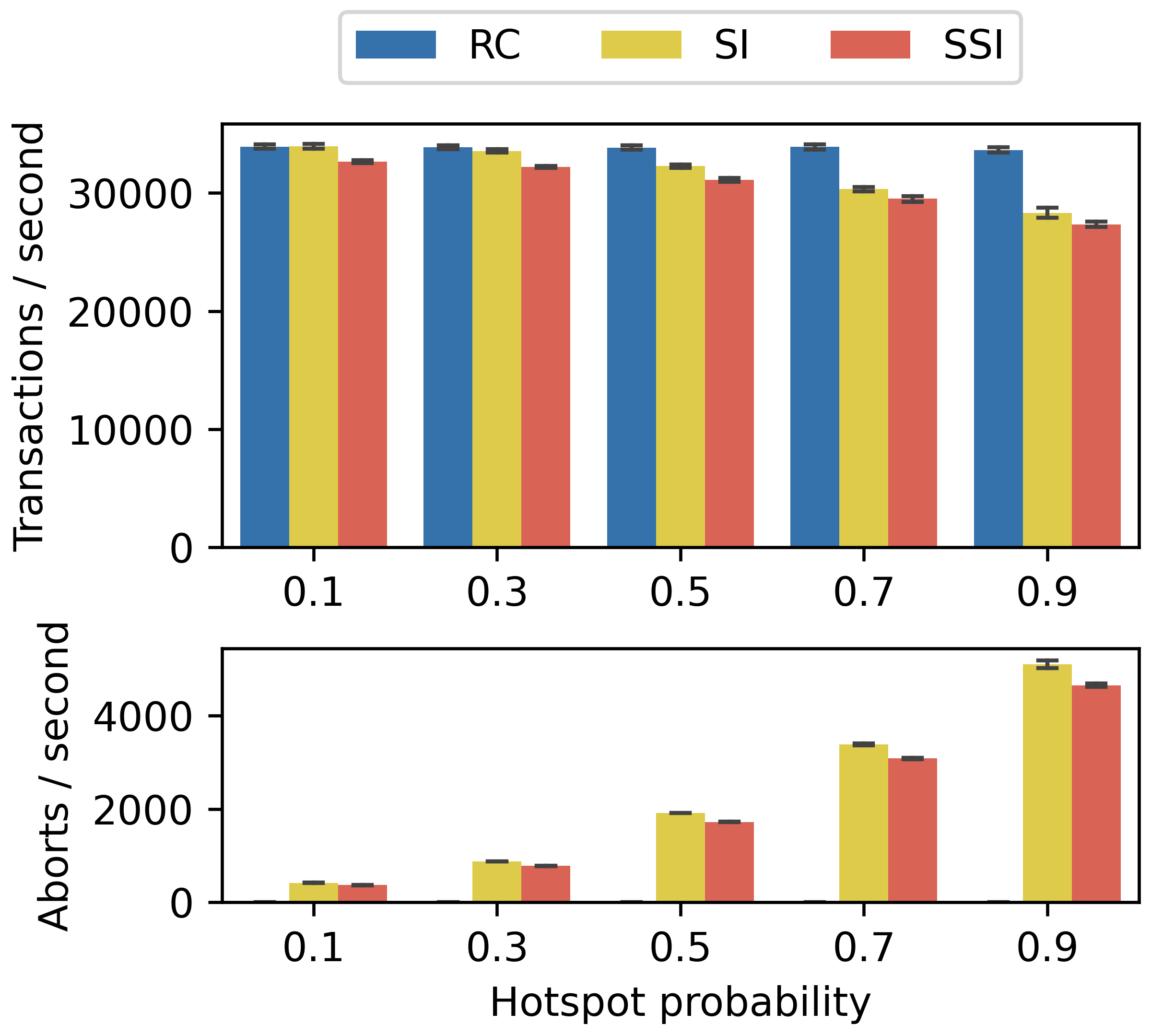}
        \caption{Hotspot of 1000 accounts}
        \label{fig:exp2_1000}
    \end{subfigure}
    \hfill
    \begin{subfigure}[t]{0.24\linewidth}
        \centering
        \includegraphics[width=.95\linewidth]{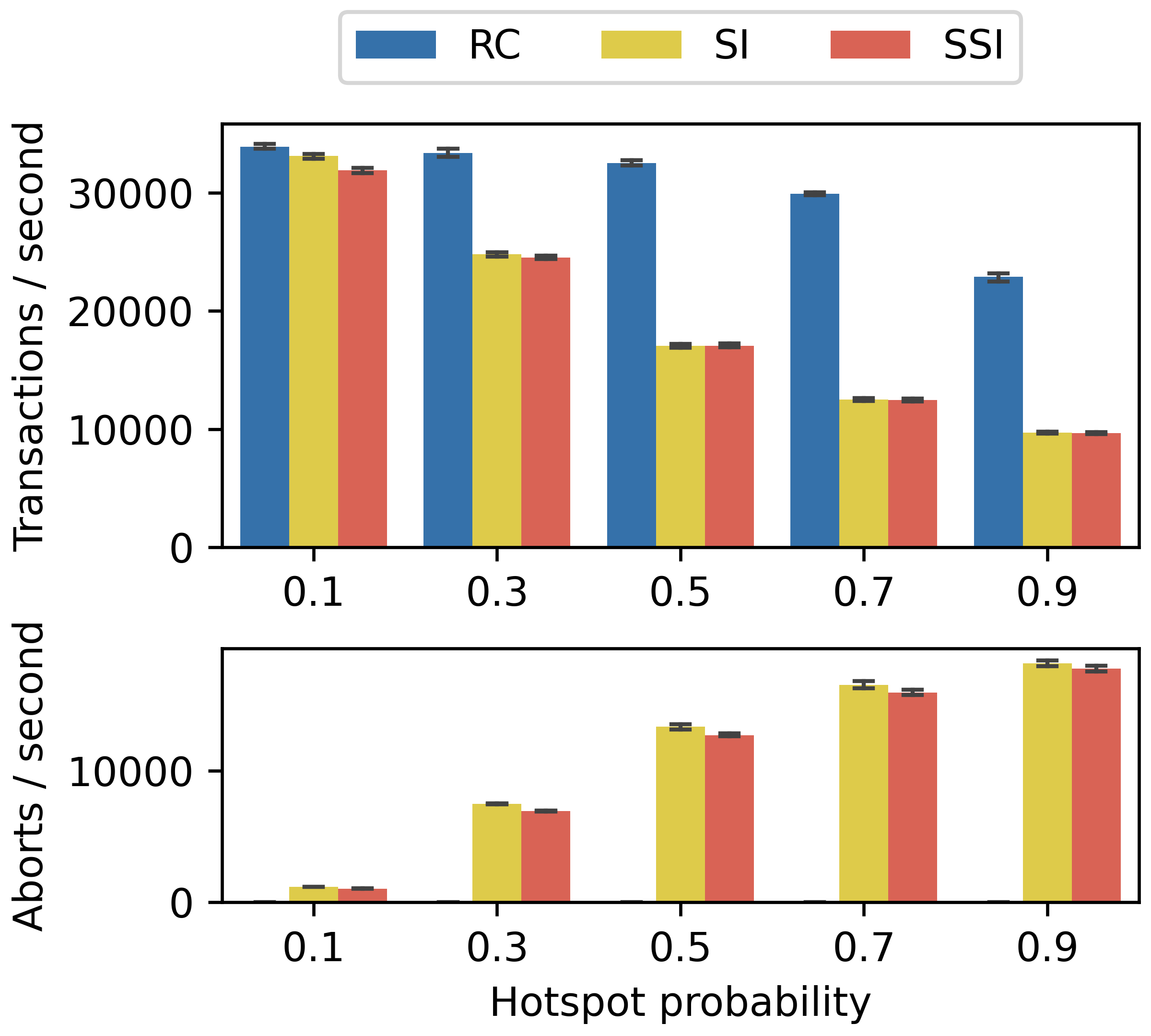}
        \caption{Hotspot of 100 accounts}
        \label{fig:exp2_100}
    \end{subfigure}
    \hfill
    \begin{subfigure}[t]{0.24\linewidth}
        \centering
        \includegraphics[width=.95\linewidth]{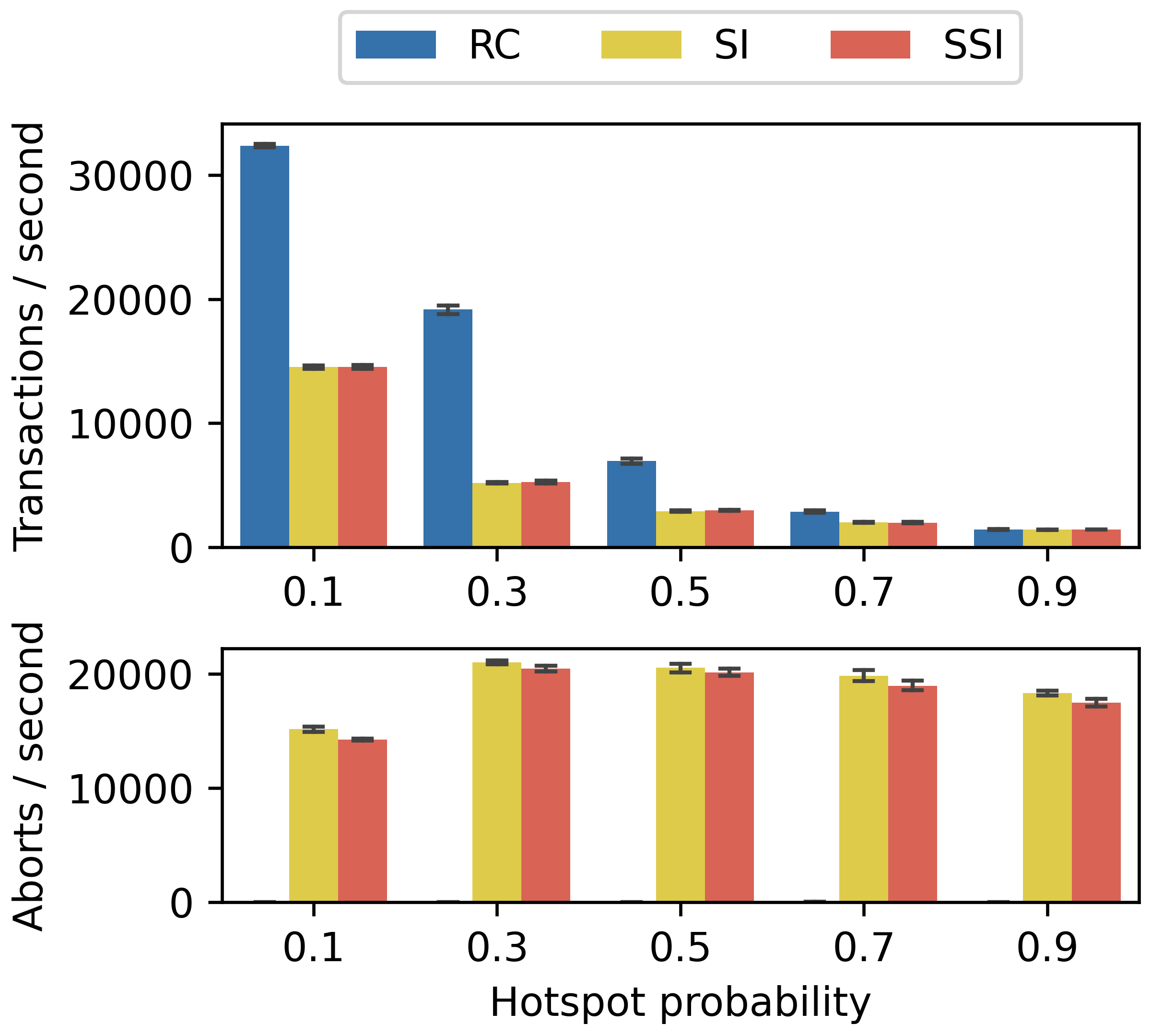}
        \caption{Hotspot of 10 accounts}
        \label{fig:exp2_10}
    \end{subfigure}
    \hfill
    \begin{subfigure}[t]{0.24\linewidth}
        \centering
        \includegraphics[width=.95\linewidth]{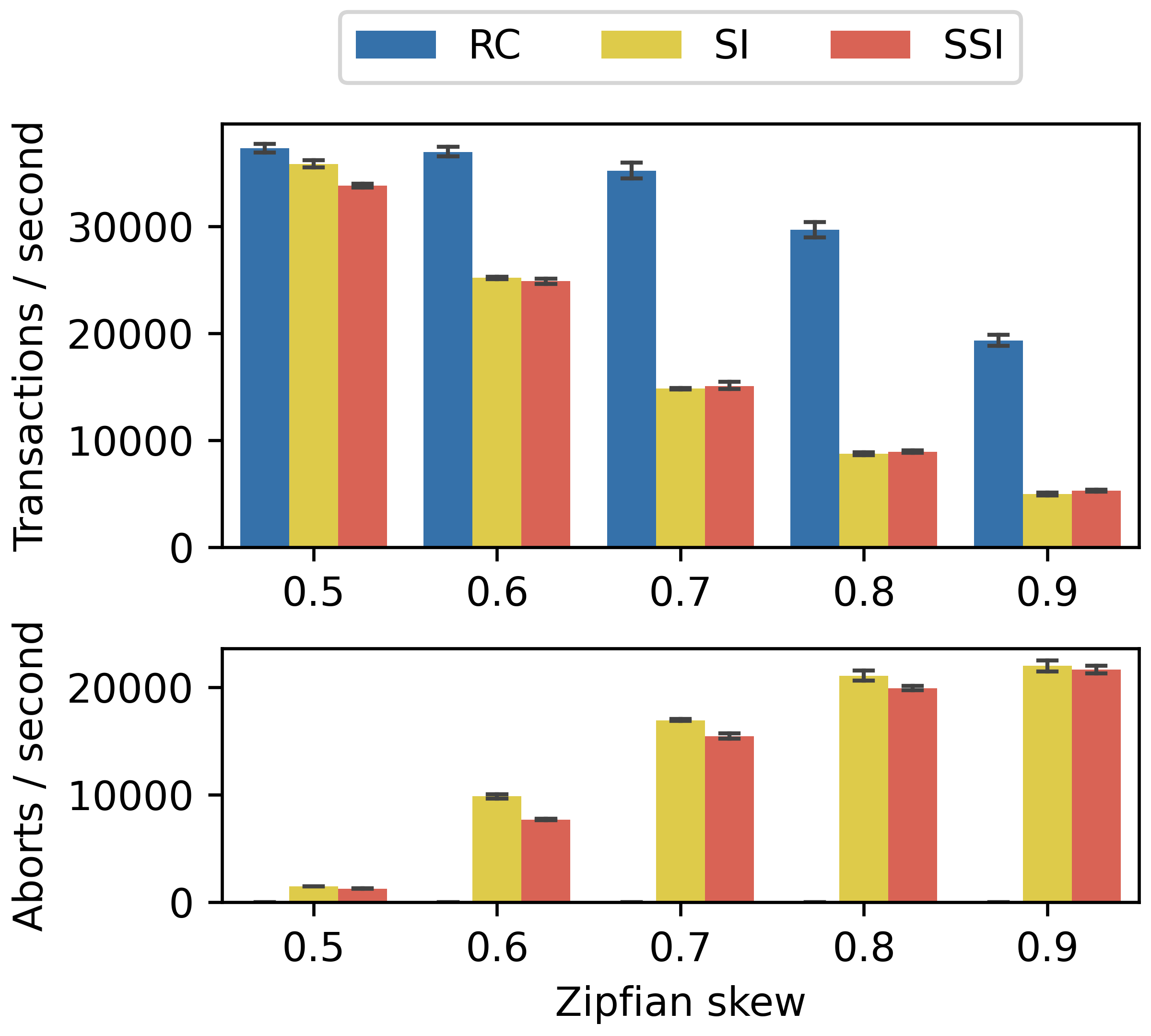}
        \caption{Zipfian distribution}
        \label{fig:exp3}
    \end{subfigure}
    \removespacetocaption
    \caption{[Robust subset of SmallBank] 
    Throughput and abort rate 
    with 200 clients 
    and different contention parameters.}
    \label{fig:exp2_all}
\end{figure*}

{

In the experiments below, we show the potential performance benefits of using a lower isolation level over a robust subset of the SmallBank benchmark.}
The first experiment explores the influence of the number of concurrent clients on both throughput and abort rate. For this experiment, we used a workload of three \ptranss{} \{DepositChecking, TransactSavings and Amalgamate\}, since this workload is the largest subset of the Smallbank benchmark that is robust against \mvrc. For this experiment, a hotspot size of 1000 accounts with
a hotspot probability of 90\% was used. The results of this experiment are shown in Figure~\ref{fig:exp1}.
When the number of clients is low, the different isolation levels result in a similar throughput. However, if the number of concurrent clients increases, \mvrc clearly outperforms both SI and SSI. This is to be expected, since the high number of concurrent clients leads to more concurrent transactions trying to update the same tuple, and consequently more aborts under SI and SSI due to the first-updater-wins principle. It should be noted that under \mvrc, aborts can still occur due to deadlocks, 
but these aborts are quite rare. In this experiment, the number of aborts under \mvrc never exceeded 0.15 aborts per second.

We next consider different levels of data skew on the dataset. Figure~\ref{fig:exp2_1000}, Figure~\ref{fig:exp2_100} and Figure~\ref{fig:exp2_10} show the throughput for different hotspot probabilities when there are respectively 1000, 100 and 10 accounts in the hotspot. Figure~\ref{fig:exp3} shows the throughput for different skew parameters when using a Zipfian distribution. When the data skew increases, \mvrc greatly outperforms the other two isolation levels. However, when contention further increases, the throughput of \mvrc decreases drastically due to transactions waiting for write locks to be released. In Figure~\ref{fig:exp2_10}, the number of aborts under \mvrc due to detected deadlocks increases to around 33 aborts per second when the hotspot probability is 90\%.

Similar findings are obtained when considering maximal subsets of the \tpcckv{} benchmark, for instance, \{Payment, OrderStatus and StockLevel\}, that are robust against \MVRC.

\paragraph{Conclusion.}
When a set of \ptranss{} is robust against \mvrc, choosing this lower isolation level never results in a performance loss. This is to be expected, since SI and SSI require additional overhead when checking for possible serialization failures that require an abort. \mvrc greatly outperforms the other isolation levels for settings with higher contention. Indeed, due to the first-updater-wins principle, SI and SSI need to abort a transaction when two concurrent transactions write to the same object. Higher contention increases this probability, resulting in an increased abort rate.

\subsection{Promoted workloads}
\label{sec:promotion}

\begin{figure}[t]
\footnotesize
    \begin{minipage}[t]{0.37\columnwidth-2ex}
    \Balance (RC+P): 
    \[
    \begin{array}{l}
    \R[]{\vx: \Account\ListAttr{N,C}}\\
    \UP[]{\vy: \Savings\ListAttr{C,B}\ListAttr{B}}\\
    \UP[]{\vz: \Checking\ListAttr{C,B}\ListAttr{B}}\\
    \end{array}
    \]
    \end{minipage}%
    \begin{minipage}[t]{0.39\columnwidth-2ex}
    WriteCheck (RC+P): 
    \[
    \begin{array}{l}
    \R[]{\vx: \Account\ListAttr{N,C}}\\
    \UP[]{\vy: \Savings\ListAttr{C,B}\ListAttr{B}}\\
    \UP[]{\vz: \Checking\ListAttr{C,B}\ListAttr{B}}\\
    \UP[]{\vz: \Checking\ListAttr{C,B}\ListAttr{B}}\\
    \end{array}
    \]
    \end{minipage}%
    \vline%
    \begin{minipage}[t]{0.32\columnwidth-2ex}
    \, \Balance (RC-\cite{DBLP:conf/aiccsa/AlomariF15}): 
    \[
    \begin{array}{l}
    \W[]{\myvv: \Conflict\ListAttr{N}}\\
    \R[]{\vx: \Account\ListAttr{N,C}}\\
    \R[]{\vy: \Savings\ListAttr{C,B}}\\
    \R[]{\vz: \Checking\ListAttr{C,B}}\\
    \end{array}
    \]
    \end{minipage}
    \removespacetocaption

    \caption{Left: \emph{all} modified templates in RC+P. Right: Example of \emph{one} template modification for RC-\cite{DBLP:conf/aiccsa/AlomariF15}.}
    \label{fig:balance-promoted}
\end{figure}

\begin{figure*}[h]
    \centering
    \begin{subfigure}[t]{0.24\linewidth}
        \centering
        \includegraphics[width=.95\linewidth]{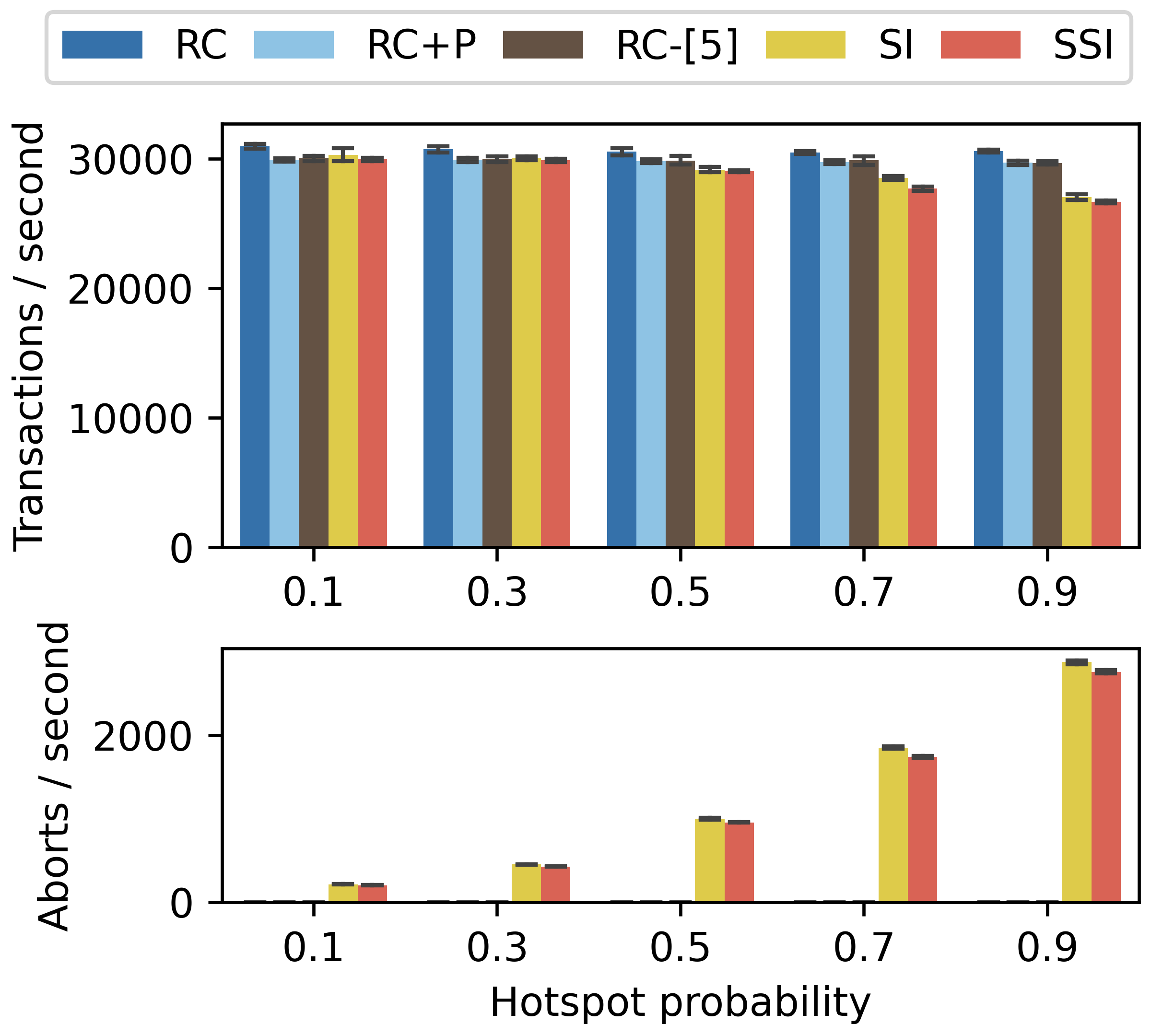}
        \caption{Hotspot of 1000 accounts}
        \label{fig:exp4_1000}
    \end{subfigure}
    \hfill
    \begin{subfigure}[t]{0.24\linewidth}
        \centering
        \includegraphics[width=.95\linewidth]{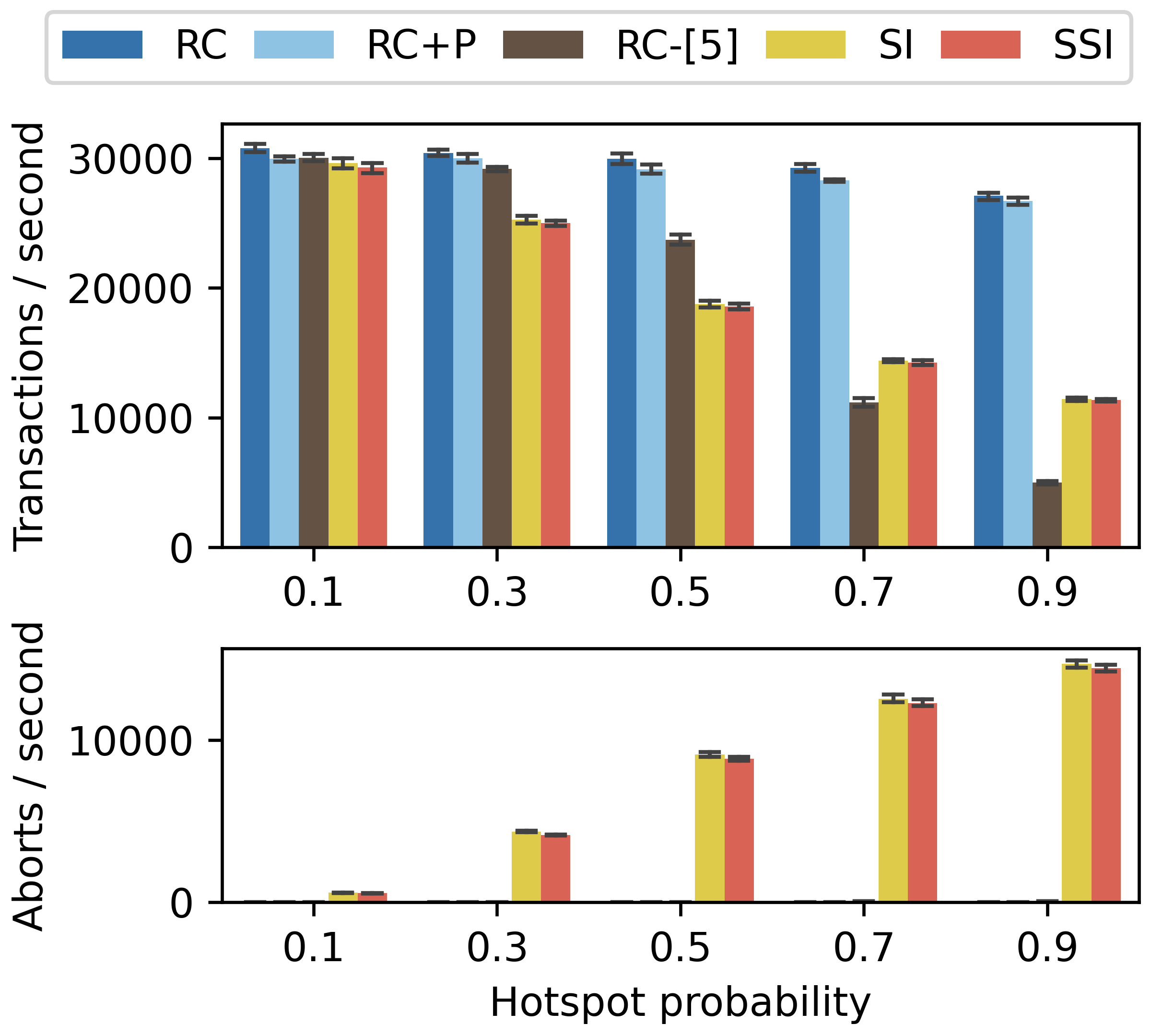}
    \caption{Hotspot of 100 accounts}
        \label{fig:exp4_100}
    \end{subfigure}
    \hfill
    \begin{subfigure}[t]{0.24\linewidth}
        \centering
        \includegraphics[width=.95\linewidth]{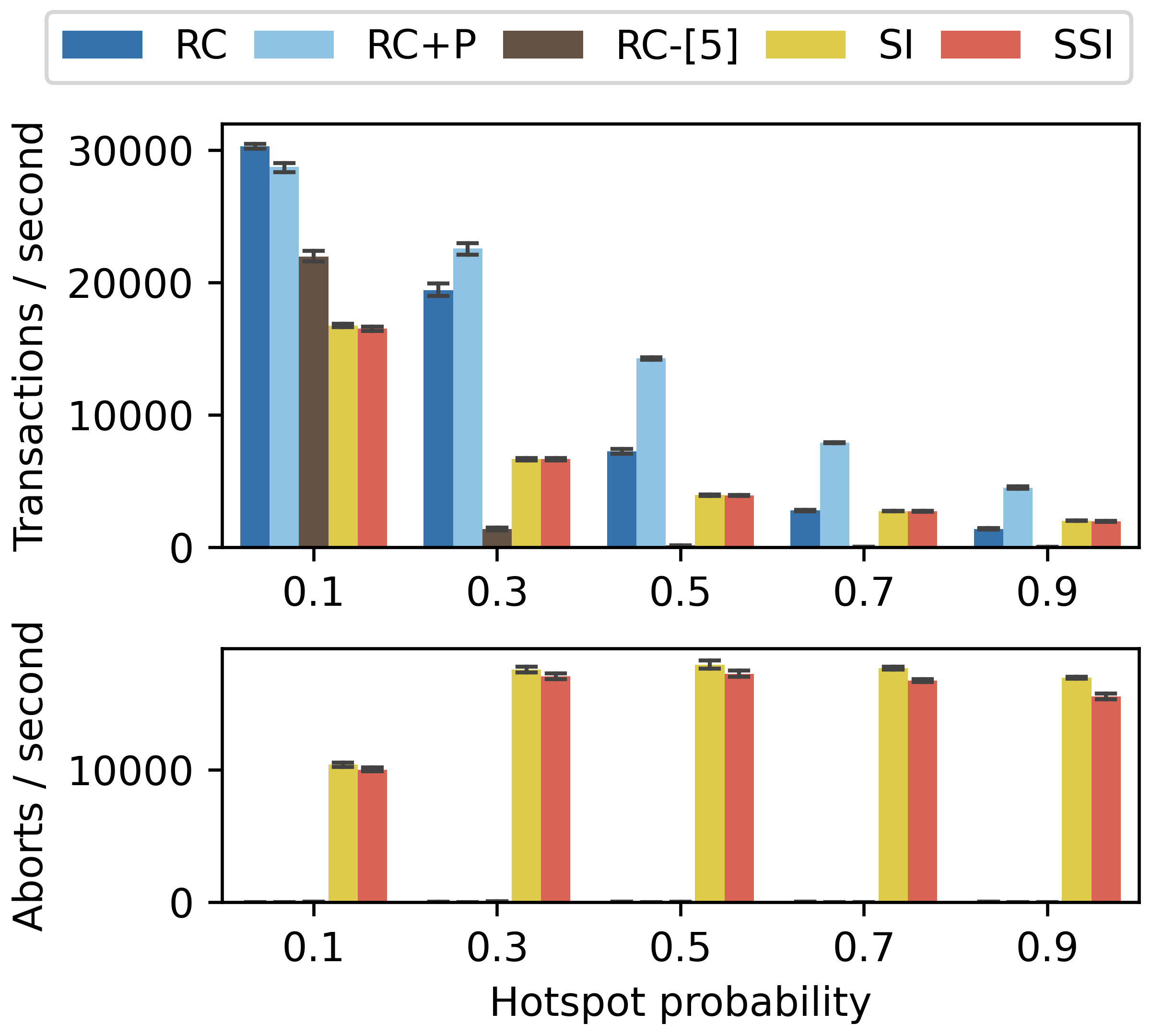}
    \caption{Hotspot of 10 accounts}
        \label{fig:exp4_10}
    \end{subfigure}
    \hfill
    \begin{subfigure}[t]{0.24\linewidth}
        \centering
        \includegraphics[width=.95\linewidth]{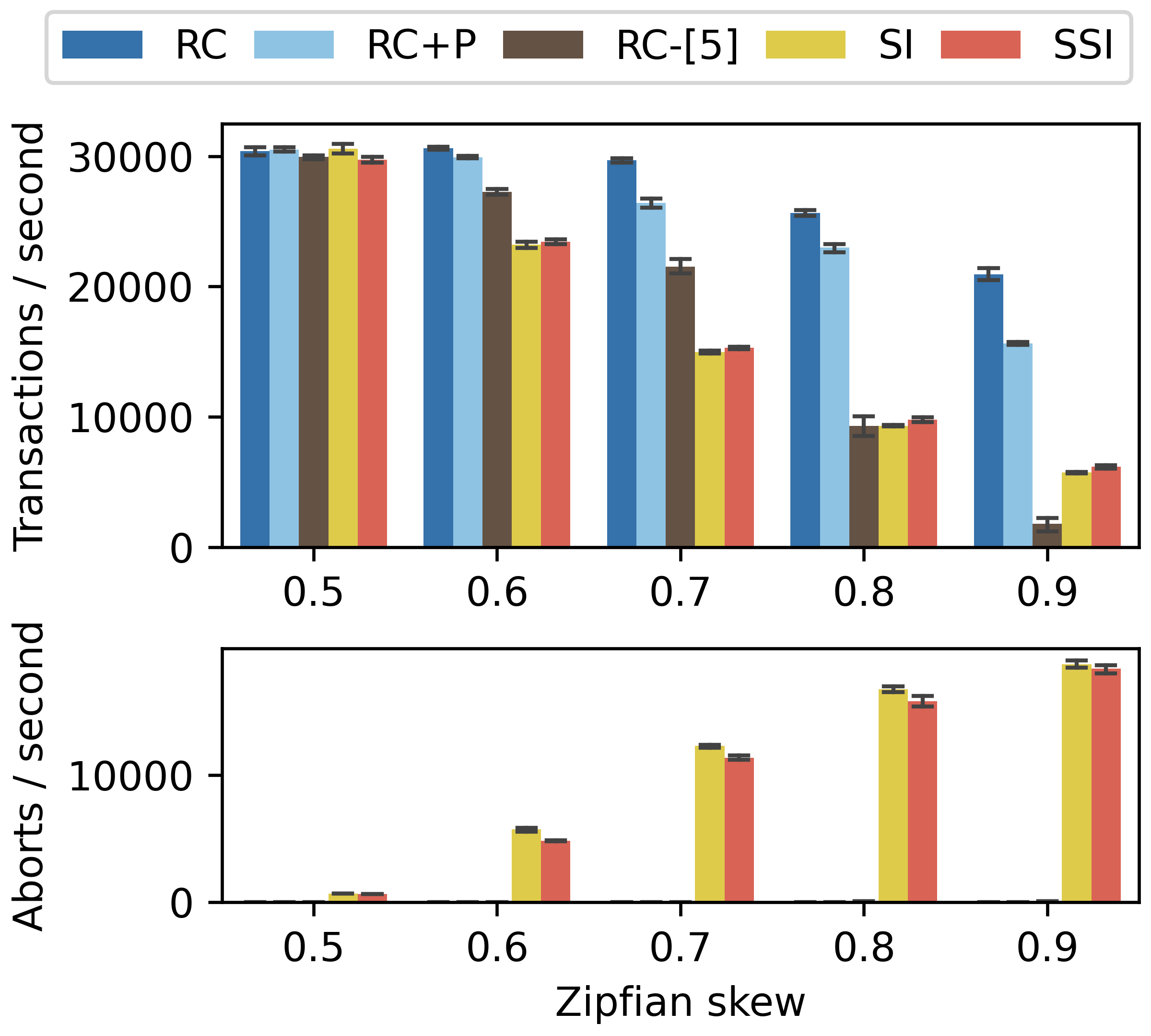}
    \caption{Zipfian distribution}
        \label{fig:exp5}
    \end{subfigure}
    \removespacetocaption
    \caption{[SmallBank 
    with promotion] Throughput and abort rate 
    with 200 clients 
    and different contention parameters.}
    \label{fig:exp45_all}
\end{figure*}

\begin{figure*}[h]
    \centering
    \begin{subfigure}[t]{0.27\linewidth}
        \centering
        \includegraphics[width=.95\linewidth]{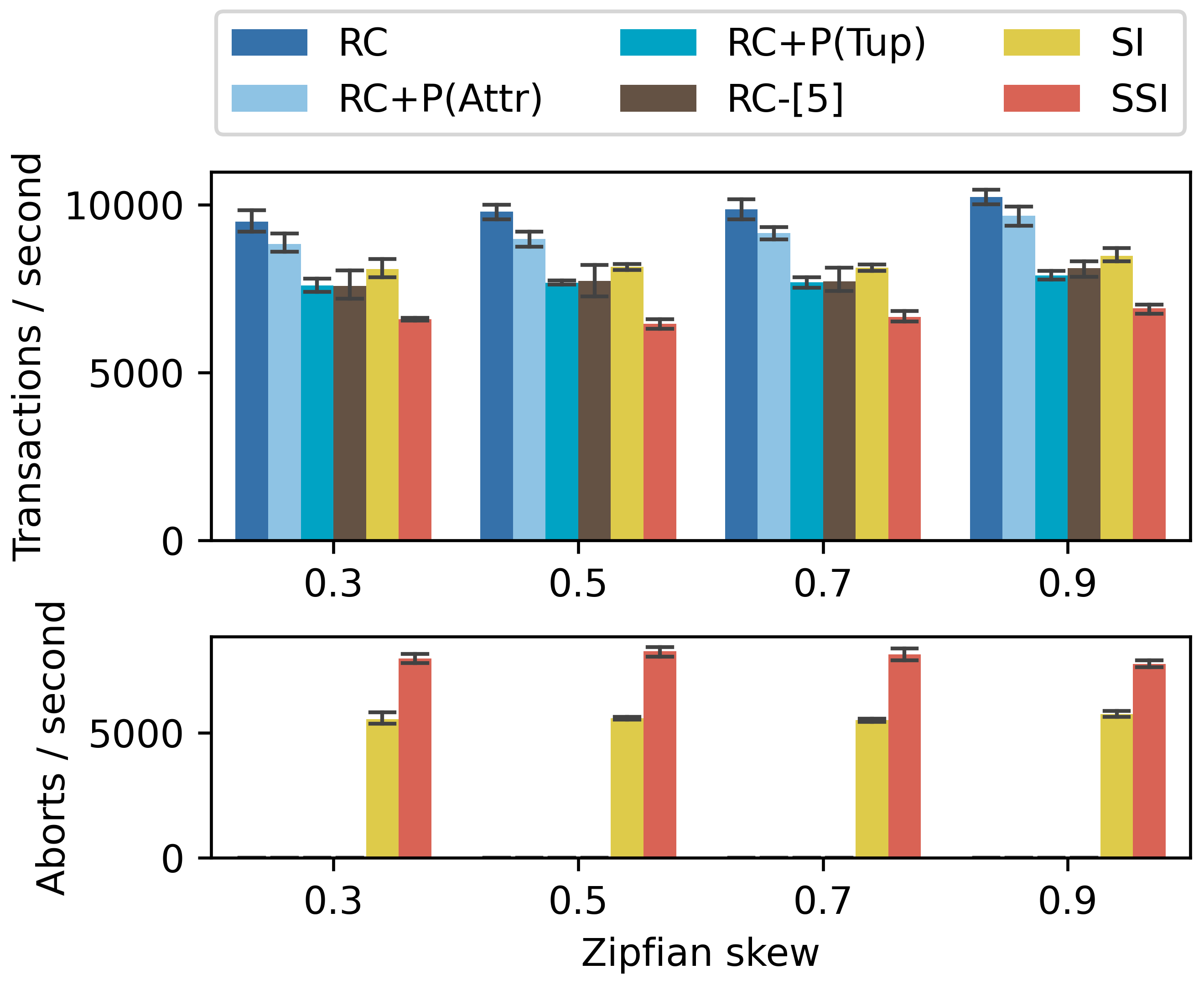}
        \caption{Influence of Zipfian skew on customers; 200 concurrent clients,  uniform distribution over warehouses.}
        \label{fig:tpcc_exp6}
    \end{subfigure}
    \hfill
    \begin{subfigure}[t]{0.27\linewidth}
        \centering
        \includegraphics[width=.95\linewidth]{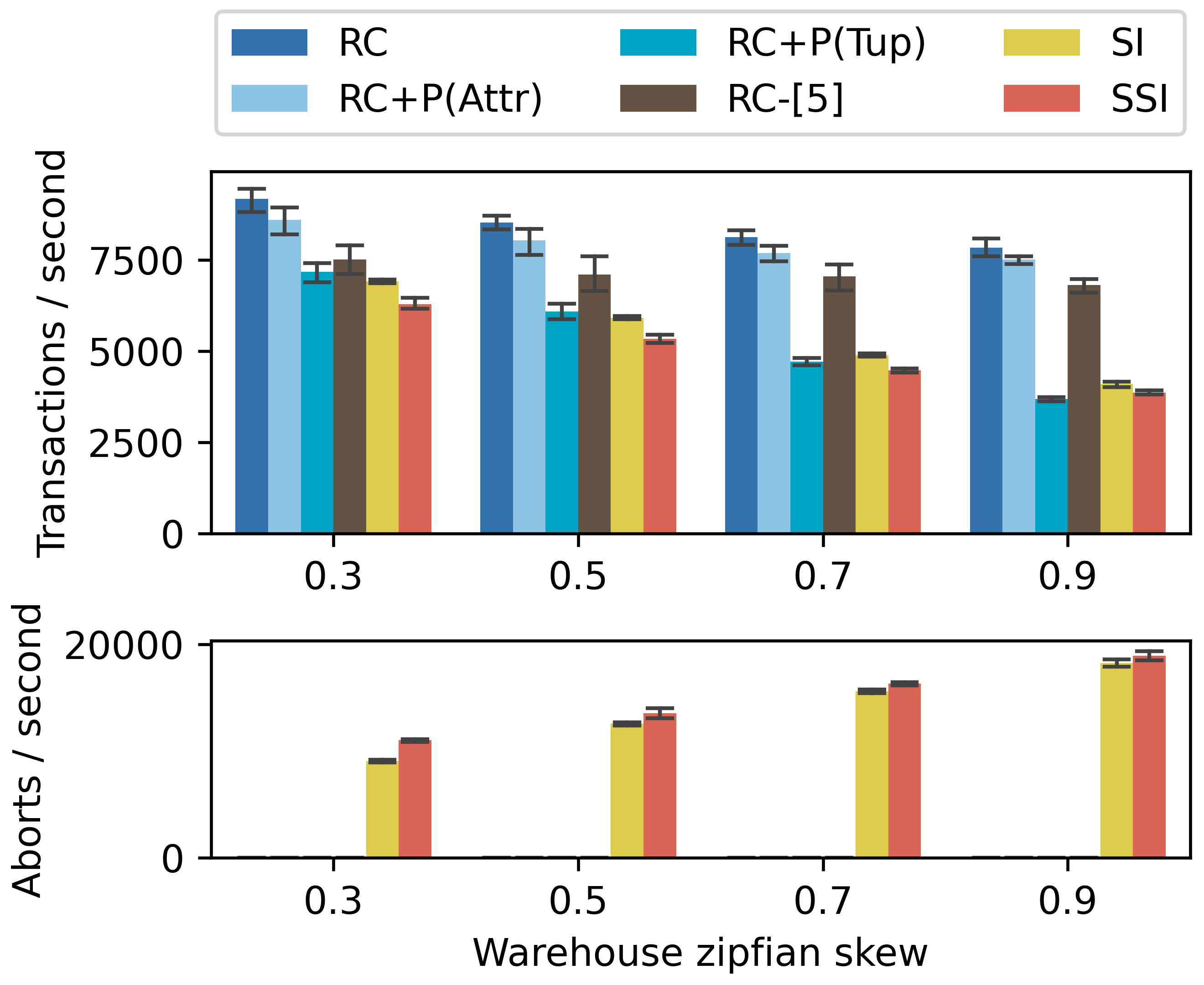}
        \caption{Influence of Zipfian skew over warehouses; fixed Zipfian skew of 0.7 over customers, 200 clients.}
        \label{fig:tpcc_exp10}
    \end{subfigure}
    \hfill
    \begin{subfigure}[t]{0.405\linewidth}
        \centering
        \includegraphics[width=.95\linewidth]{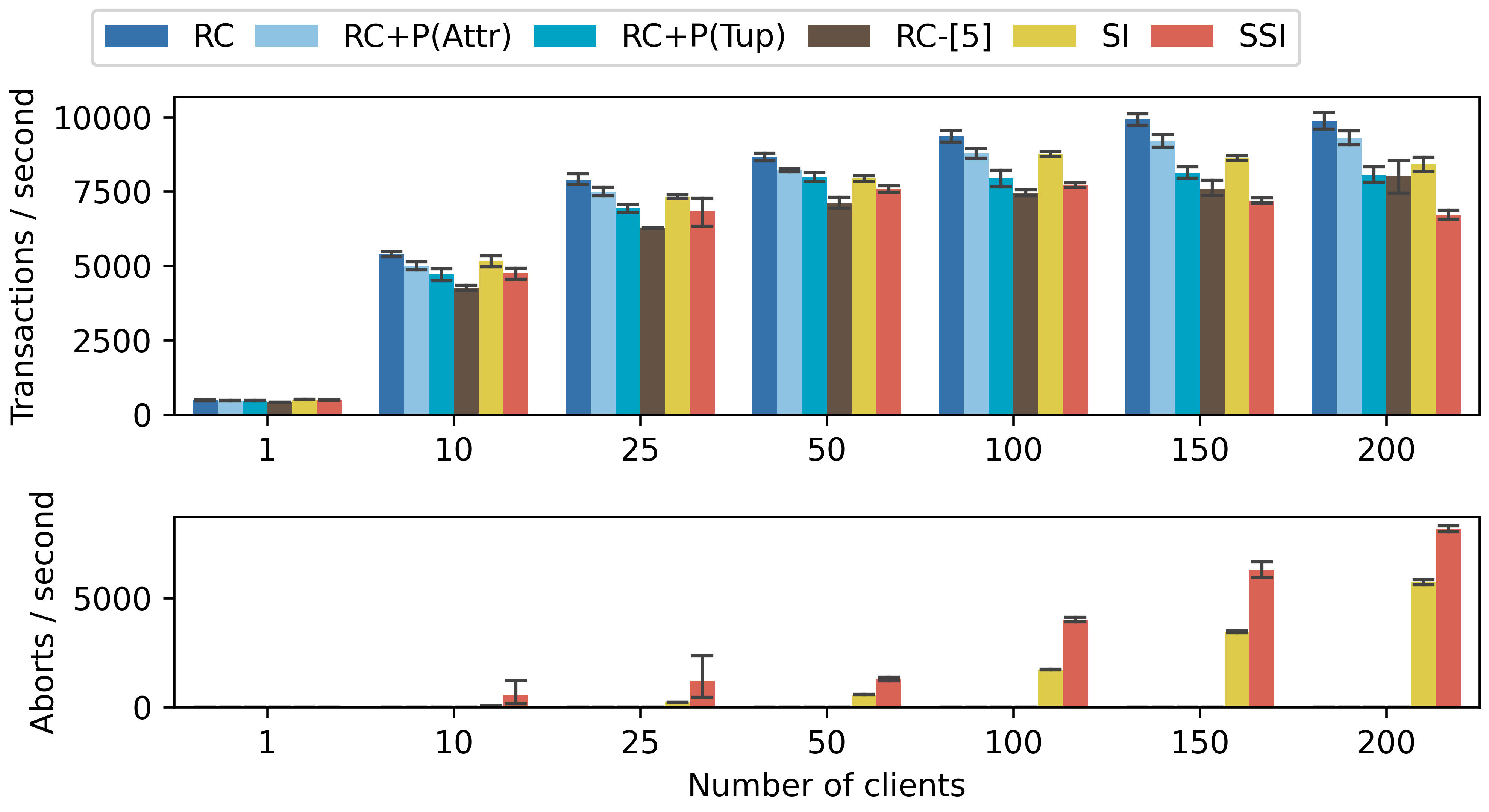}
        \caption{Influence of the number of clients; Zipfian skew of 0.7 over customers, uniform distribution over warehouses.}
        \label{fig:tpcc_exp2}
    \end{subfigure}
    \removespacetocaption
    \caption{[\tpcckv{} 
    with promotion] Throughput and abort rate 
    for 25 warehouses 
    and different contention parameters.}
    \label{fig:tpcc_all}
\end{figure*}

\begin{figure}[h]
    \centering
    \begin{subfigure}[t]{0.480\linewidth}
        \centering
    \includegraphics[width=\linewidth]{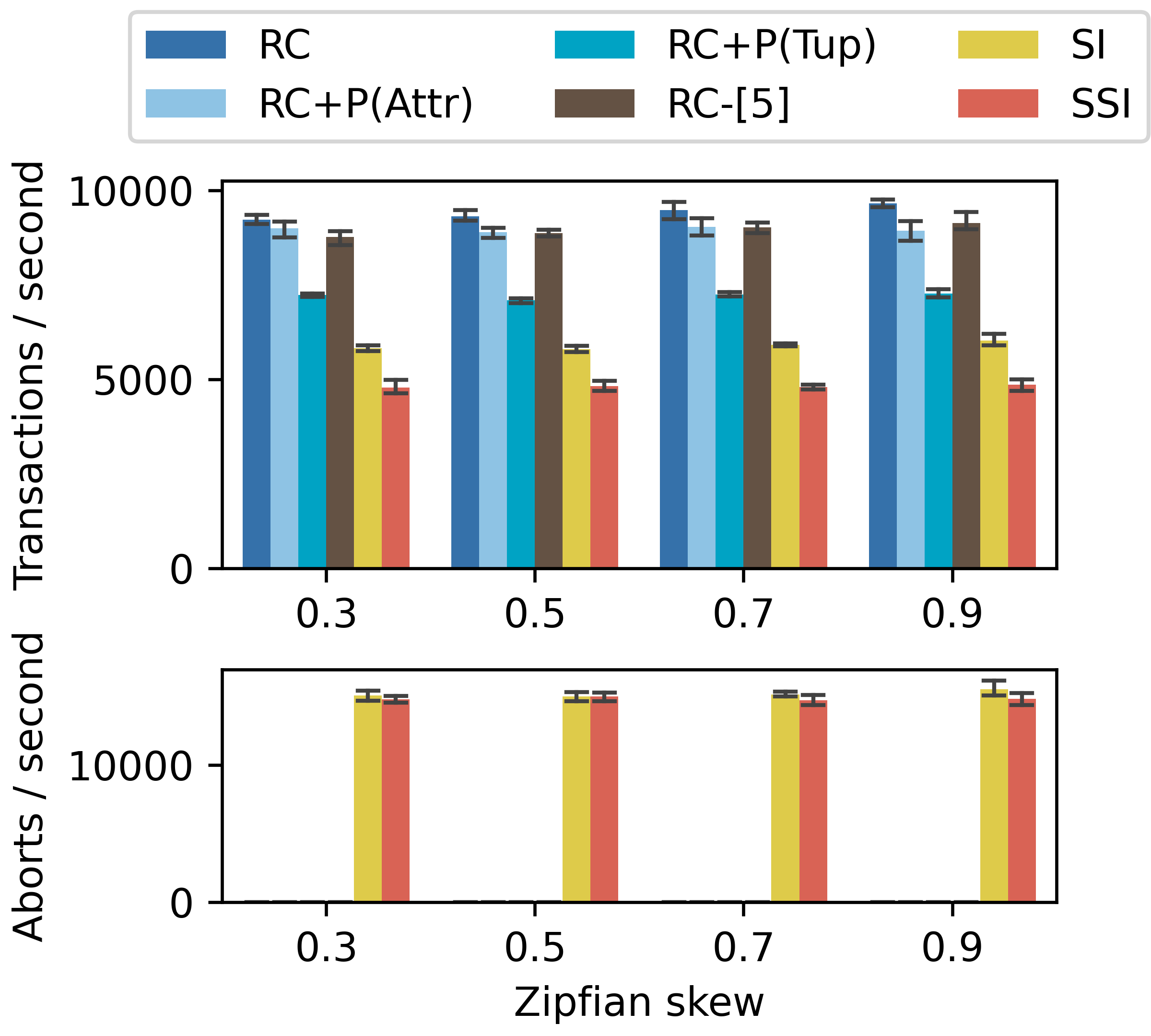}
    \caption{Influence of Zipfian skew over customers when the dataset consists of 10 warehouses.}
        \label{fig:tpcc_exp7}
    \end{subfigure}
    \hfill
    \begin{subfigure}[t]{0.480\linewidth}
        \centering
    \includegraphics[width=\linewidth]{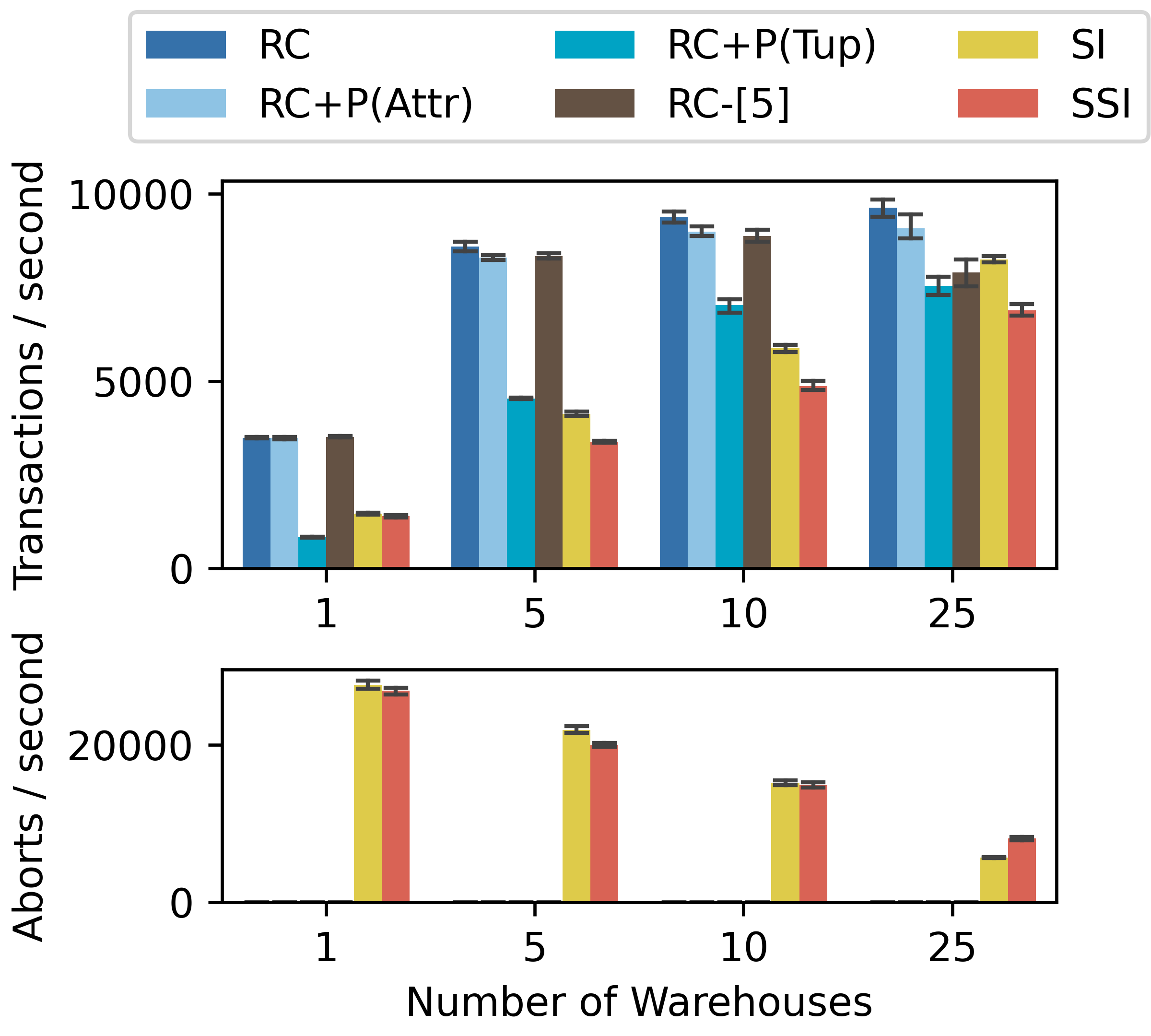}
    \caption{Influence of the number of warehouses with a Zipfian skew of 0.7 over the customers.}
        \label{fig:tpcc_exp12}
    \end{subfigure}
    \removespacetocaption
    \caption{[\tpcckv{} + promotion] Throughput and abort rate;
    200 clients, uniform distribution over warehouses.}
    \label{fig:exp7and12}
\end{figure}

\subsubsection{Promotion.}
{
When a set of \ptranss{} is not robust, we propose a template modification technique based on insights from Definition~\ref{def:mvsplitschedule}:
an equivalent set of \ptranss{} robust against \mvrc can be created by promoting \myR-operations to \myUP-operations that write back the read value.
Such a change does not alter the effect of the \ptrans{}, but the newly introduced write operation will trigger concurrency mechanisms in the database.
We emphasize that 
this is a general technique that can \emph{always} be used to construct an equivalent robust set of \shortptranss{}: 
Definition~\ref{def:mvsplitschedule} requires that operation $b_1$ is rw-conflicting with $a_1$ (Condition~(\ref{c:3})), but not ww-conflicting with $a_1$ (Condition~(\ref{c:1})), so promoting \emph{all\/} \myR-operations to \myUP-operations is sufficient to guarantee robustness against \mvrc.

The promotion approach is inspired by a technique introduced by Fekete et al.~\cite{DBLP:journals/tods/FeketeLOOS05} to make a workload robust against SI. However, in contrast to their approach, which introduces \emph{additional} write operations, we promote an \emph{existing} $\myR$-operation into a $\myUP$-operation.
Fortunately, it is not always necessary to promote all \myR-operations to obtain robustness against \mvrc:
}
to find a minimal set of \myR-operations to promote, 
we can iteratively promote $\myR$-operations to $\myUP$-operations and apply Algorithm~\ref{alg:ptime:template} to check whether the resulting workload is robust. We applied this technique on both SmallBank and \tpcckv{} to guarantee robustness with a minimal number of promotions.

For {SmallBank}, we can obtain robustness by only promoting all \myR-operations over the Checking and Savings relations to \myUP-operations leaving all other \myR-operation intact (cf.~Figure~\ref{fig:balance-promoted} and notice that only 2/5 templates are modified and in total only four reads need to be promoted). 
In our experiments, we refer to this promotion as \emph{RC+P}. Furthermore, this set of promoted \myR-operations is minimal: if one of the \myR-operations over the Checking or Savings relations remains,the application of Algorithm~\ref{alg:ptime:template} reveals that the resulting set of \ptranss{} is not robust against \mvrc.

For {\tpcckv{}}, we can obtain robustness by promoting all \myR-operations over the Customer, Order and OrderLine relations in the OrderStatus template while all other templates remain unchanged. {We refer to this promotion as \emph{RC+P(Attr)}.}
To contrast our approach based on attribute-level conflicts with the one based on tuple-level conflicts, we also investigate how to make \tpcckv{} robust when the read and write sets of operations refer to \emph{all} attributes in the corresponding relations. Again we applied Algorithm~\ref{alg:ptime:template} and obtained that all \myR-operations on tuples over the Warehouse-, Customer- Order- and OrderLine-relations need to be promoted to \myUP-operations, requiring changes in both NewOrder and OrderStatus. {We refer to this promotion as \emph{RC+P(Tup)}.} 
Both promotion strategies are again minimal, since we cannot promote only a strict subset of these \myR-operations to \myUP-operations without losing robustness. When comparing RC+P(Attr) to RC+P(Tup), we see that for \tpcckv{} an analysis on the granularity of tuples requires strictly \emph{more} \myR-operations to be promoted leading to a smaller throughput compared to RC+P(Attr)
as the experiments will show.

We also compare with the Internal Concurrency Exclusion~\cite{DBLP:conf/aiccsa/AlomariF15} approach (referring to the latter as \emph{RC-\cite{DBLP:conf/aiccsa/AlomariF15}}) for making workloads robust against RC extended to attribute-level conflicts.
Each \shortptrans{} is changed by adding additional leading $\myW$-operations overwriting tuples in a newly introduced Conflict relation 
such that every pair of instances that might produce a counterflow edge is guaranteed to write to the same tuple in the Conflict relation. 
The latter is achieved by selecting specific relations in the original benchmark and adding a tuple to relation Conflict for each tuple in the selected relations. Each \shortptrans{} is then changed so that if it accesses one of the tuples of these selected relations, it also writes to the corresponding tuple in relation Conflict.
For SmallBank, it suffices to select relation Account. Figure~\ref{fig:balance-promoted}(right) illustrates the required changes for template \Balance. We stress that this change needs to be done for \emph{every} template.
There are even two $\myW$-operations needed at the beginning of Amalgamate, as it considers two different customers. 
For \tpcckv{}, the selected relations are Stock (thereby requiring additional writes in \shortptranss{} NewOrder and StockLevel) and Customer (thereby requiring additional writes in \shortptranss{} NewOrder, Delivery, OrderStatus and Payment). 
\ifthenelse{\boolean{fullversion}}
{
We refer to the appendix for an overview of the concrete templates used for RC-\cite{DBLP:conf/aiccsa/AlomariF15}.
}
{
We refer 
to~\cite{fullversion} for 
the concrete templates used for RC-\cite{DBLP:conf/aiccsa/AlomariF15}.
}

We also include the performance of the unmodified \shortptranss{} under RC, SI and SSI as a baseline. Recall that both SmallBank and \tpcckv{} are not robust against RC, and SmallBank is not robust against SI. So, to be fair, the performance of the promoted workloads 
should be compared to SSI (for SmallBank) and SI (for \tpcckv).

\subsubsection{SmallBank}

Figure~\ref{fig:exp8} compares the throughput for different numbers of concurrent clients. When contention is lower due to fewer clients, the throughput of RC+P is comparable to \mvrc, SI and SSI. When the number of clients increases, RC+P still outperforms SI and SSI (and also RC-\cite{DBLP:conf/aiccsa/AlomariF15}), although the performance gain is less as 
compared to Figure~\ref{fig:exp1}.

Similarly to Figure~\ref{fig:exp2_all}, we use 200 clients to query the database with different levels of skew, but this time with all \ptranss{} in the SmallBank benchmark and using promoted operations {in RC+P.} 
Figures~\ref{fig:exp4_1000}, \ref{fig:exp4_100} and \ref{fig:exp4_10} show the throughput for different hotspot sizes and probabilities. Figure~\ref{fig:exp5} shows the throughput when a Zipfian distribution is used instead. The experiments show that RC+P
outperforms SI, SSI, and RC-\cite{DBLP:conf/aiccsa/AlomariF15} when the hotspot is smaller, the hotspot probability increases, or skew increases. 
The improvements over SI and SSI result from the high number of aborts under SI and SSI. The improvement over RC-\cite{DBLP:conf/aiccsa/AlomariF15} (which can be an order of magnitude depending on the setting) can be explained by noting that RC+P allows for more concurrency than RC-\cite{DBLP:conf/aiccsa/AlomariF15}. Indeed, for RC-\cite{DBLP:conf/aiccsa/AlomariF15}, two instances of \shortptranss{} that access the same customer can never be concurrent, as they both initially write to the same tuple of type Conflict. For RC+P, some concurrency is still possible in this case: i.e., an execution of DepositChecking can interleave with an execution of TransactSavings over the same customer, as the former only updates the tuple of type Checking, whereas the latter only updates the tuple of type Savings.

Furthermore, the performance of RC+P is usually comparable to RC. In fact, when the hotspot is small (Figure~\ref{fig:exp4_10}), RC+P is even able to outperform the original \shortptranss{} under RC due to a reduced number of deadlocks: the abort rate for RC+P increases to around 20 aborts per second under a hotspot probability of 0.9, whereas RC increases to around 45 aborts per second in this setting.

\subsubsection{\tpcckv{}}

Figure~\ref{fig:tpcc_exp6} and Figure~\ref{fig:tpcc_exp7} show the throughput and abort rate depending on the Zipfian skew over the Customer relation, over a dataset consisting of respectively 25 and 10 warehouses. Noticeably, changing the skew over the Customer relation does not result in a significant change in throughput.
The reason for this is that the throughput bottleneck is not caused by multiple transactions accessing the same customer, but by accessing the same warehouse instead. This is to be expected, since the number of warehouses in our dataset is several magnitudes smaller than the total number of customers.
We further investigate the influence of the number of warehouses (Figure~\ref{fig:tpcc_exp12}), as well as a Zipfian skew over the Warehouse relation (Figure~\ref{fig:tpcc_exp10}). Figure~\ref{fig:tpcc_exp2} investigates the influence of different levels of contention on performance by testing different numbers of concurrent clients.

When comparing both promotion strategies, we conclude that RC+P(Attr) always outperforms RC+P(Tup), especially when the number of warehouses is lowered (Figure~\ref{fig:tpcc_exp12}), or when a larger Zipfian skew is used over the Warehouse relation (Figure~\ref{fig:tpcc_exp10}). In these cases, it should also be noted that RC+P(Attr) clearly dominates the higher isolation levels SI and SSI. For example, when the dataset consists of 5 warehouses in Figure~\ref{fig:tpcc_exp12}, the throughput of RC+P(Attr) averages around 8300 transactions per second, contrasting the average of 4100 transactions per second for SI. Furthermore, the performance loss of RC+P(Attr) compared to RC over the original templates is always relatively small, indicating that our approach allows to achieve serializability guarantees in exchange for a minor performance loss.

When analyzing RC-\cite{DBLP:conf/aiccsa/AlomariF15}, we see that its throughput relative to the throughput of RC+P(Attr) is highly dependent on the number of warehouses (Figure~\ref{fig:tpcc_exp12}), as well as the Zipfian skew over the Warehouse relation (Figure~\ref{fig:tpcc_exp10}). In particular, when the number of warehouses is lowered or the skew increased, the throughput of RC-\cite{DBLP:conf/aiccsa/AlomariF15} is similar to that of RC+P(attr). If on the other hand the number of warehouses is larger and a uniform distribution over the Warehouse relation is used, then RC+P(Attr) clearly outperforms RC-\cite{DBLP:conf/aiccsa/AlomariF15}. Consider for example the setting with 25 warehouses in Figure~\ref{fig:tpcc_exp12}. Then, the average throughputs of RC+P(Attr) and RC-\cite{DBLP:conf/aiccsa/AlomariF15} are respectively around 9000 and 7800 transactions per second.

\subsubsection{Conclusion.}
Our promotion technique outperforms the isolation levels SI and SSI 
under higher contention and guarantees serializability under RC while requiring only a minor performance cost compared to RC over the original templates. When comparing to earlier work based on Internal Concurrency Exclusion~\cite{DBLP:conf/aiccsa/AlomariF15}, on SmallBank our approach significantly outperforms RC-\cite{DBLP:conf/aiccsa/AlomariF15} when contention increases (an order of magnitude in extreme cases). For \tpcckv{}, the performance gain is similar to that of RC-\cite{DBLP:conf/aiccsa/AlomariF15}, although we are still able to identify cases where our technique outperforms RC-\cite{DBLP:conf/aiccsa/AlomariF15} by more than 15\%. Finally, comparing RC+P(Attr) versus RC+P(Tup) for 
\tpcckv{} shows that promotion based on attribute-level analysis significantly outperforms tuple-level analysis.

\ifthenelse{\boolean{fullversion}}{

\section{Conclusion}
\label{sec:concl}

We pushed the frontier of the robustness problem for RC and showed that 
an explicit formalisation detects larger sets of transaction workloads to be robust.
The throughput of a relational database system processing transactions under isolation level Serializable can be improved by an approach based on robustness testing and safely executing transactions under the lower isolation level RC. 
In the future we plan to build further on the gained insights to cover more expressive transaction programs and also explore sufficient conditions
for robustness against \MVRC.

}{}

\begin{acks}
The resources and services used in this work were provided by the VSC (Flemish Supercomputer Center), funded by the Research Foundation -- Flanders (FWO) and the Flemish Government.
\end{acks}

\clearpage

\balance
\bibliographystyle{ACM-Reference-Format}
\bibliography{references}

\ifthenelse{\boolean{fullversion}}
{
\clearpage
\onecolumn

\appendix
\section*{Appendix}


\section{Detailed Benchmark Analysis}

This section provides a detailed overview of the robustness properties for both the SmallBank and TPC-C benchmark. We analyse robustness against \mvrc both on the granularity of attributes and tuples, providing concrete counterexample schedules for all subsets that are not considered robust.

\subsection{SmallBank Transaction Templates}

\begin{figure*}
\begin{minipage}[t]{\textwidth/2-2ex}
\begin{verbatim}
Balance(N):
    SELECT CustomerId INTO :x
      FROM Account
     WHERE Name=:N;
        
    SELECT Balance INTO :a 
      FROM Savings
     WHERE CustomerId=:x;
     
    SELECT Balance + :a 
      FROM Checking
     WHERE CustomerId=:x;
    COMMIT;

Amalgamate(N1,N2):
    SELECT CustomerId INTO :x1
      FROM Account
     WHERE Name=:N1;
    
    SELECT CustomerId INTO :x2
      FROM Account
     WHERE Name=:N2;
    
    UPDATE Savings AS new
       SET Balance = 0
      FROM Savings AS old
     WHERE new.CustomerId=:x1
           AND old.CustomerId=new.CustomerId
    RETURNING old.Balance INTO :a;
    
    UPDATE Checking AS new
       SET Balance = 0
      FROM Checking AS old
     WHERE new.CustomerId=:x1
           AND old.CustomerId=new.CustomerId
    RETURNING old.Balance INTO :b;
    
    UPDATE Checking
       SET Balance = Balance + :a + :b
     WHERE CustomerId=:x2;

DepositChecking(N,V):
    SELECT CustomerId INTO :x
      FROM Account
     WHERE Name=:N;
     
    UPDATE Checking
       SET Balance = Balance + :V 
     WHERE CustomerId=:x;
    COMMIT;
\end{verbatim}
\end{minipage}   
\begin{minipage}[t]{\textwidth/2}
\begin{verbatim}
TransactSavings(N,V):
    SELECT CustomerId INTO :x
      FROM Account
     WHERE Name=:N;
        
    UPDATE Savings
       SET Balance = Balance + :V 
     WHERE CustomerId=:x;
    COMMIT;

WriteCheck(N,V):
    SELECT CustomerId INTO :x
      FROM Account
     WHERE Name=:N;
        
    SELECT Balance INTO :a 
      FROM Savings
     WHERE CustomerId=:x;
     
    SELECT Balance INTO :b 
      FROM Checking
     WHERE CustomerId=:x;
     
    IF (:a + :b) < :V THEN 
        UPDATE Checking
           SET Balance = Balance - (:V + 1) 
         WHERE CustomerId=:x;
    ELSE
        UPDATE Checking
           SET Balance = Balance - :V
         WHERE CustomerId=:x;
    END IF;
    COMMIT;
\end{verbatim}
\end{minipage}
\caption{SmallBank SQL Transaction Templates.}
\label{fig:smallbank:SQL}
\end{figure*}

Figure~\ref{fig:smallbank:SQL} contains the SQL code for the SmallBank transaction templates presented in Figure~\ref{fig:smallbank-abstract-syntax}. We identified three maximal robust subsets of \ptranss{} that are robust against \mvrc:
\begin{itemize}
    \item \{DepositChecking, TransactSavings, Amalgamate\},
    \item \{Balance, DepositChecking\}, and
    \item \{Balance, TransactSavings\}.
\end{itemize}
Figure~\ref{fig:counterex_SB} shows that these are indeed the only robust subsets by providing counterexample multiversion split schedules for sets of templates that are not robust against \mvrc.
We only provide counterexamples over minimal subsets that are not robust against \mvrc, as these schedules immediately serve as counterexamples over larger subsets as well.
An analysis of SmallBank on the granularity of tuples instead of attributes reveals that the robustness analysis remains unchanged. This is to be expected, since for this benchmark all conflicts on the granularity of tuples coincide with conflicts on the granularity of attributes. Indeed, all conflicting operations access the same attribute Balance in the Checking and Savings relations.

\begin{figure*}
    \begin{minipage}[t]{\textwidth}
        $$
        \begin{array}{lccccc}
            \text{$\trans[1]$ (WriteCheck):}&  \R[1]{x} \, \R[1]{y} \, \R[1]{z\ListAttr{C, B}} & & \UP[1]{z\ListAttr{C, B}\ListAttr{B}} \, \CT[1]\\
            \text{$\trans[2]$ (WriteCheck):}& & \R[2]{x} \, \R[2]{y} \, \R[2]{z\ListAttr{C, B}} \, \UP[2]{z\ListAttr{C, B}\ListAttr{B}} \, \CT[2]&
         \end{array}
        $$
        \subcaption{\{WriteCheck\} is not robust against \readmvcom.}\label{fig:counterex_SB_1}
    \end{minipage}
    \begin{minipage}[t]{\textwidth}
        $$
        \begin{array}{lccccc}
            \text{$\trans[1]$ (\Balance):}&  \R[1]{x_1} \, \R[1]{y_1\ListAttr{C, B}} & & \R[1]{z_1\ListAttr{C, B}} \, \CT[1]\\
            \text{$\trans[2]$ (\Amalgamate):}& &  \R[2]{x_1} \, \R[2]{x_2} \, \UP[2]{y_1\ListAttr{C, B}\ListAttr{B}} \, \UP[2]{z_1\ListAttr{C, B}\ListAttr{B}} \, \UP[2]{z_2} \, \CT[2]&
         \end{array}
        $$
        \subcaption{\{Balance, Amalgamate\} is not robust against \readmvcom.}\label{fig:counterex_SB_3}
    \end{minipage}
    \begin{minipage}[t]{\textwidth}
        $$\arraycolsep=.1em
        \begin{array}{lccccc}
            \text{$\trans[1]$ (\Balance):}&  \R[1]{x} \, \R[1]{y\ListAttr{C, B}} & & & & \R[1]{z\ListAttr{C, B}} \, \CT[1]\\
            \text{$\trans[2]$ (\TransactSavings):}& &  \R[2]{x} \, \UP[2]{y\ListAttr{C, B}\ListAttr{B}} \, \CT[2]&\\
            \text{$\trans[3]$ (\Balance):}& & &  \R[3]{x} \, \R[3]{y\ListAttr{C, B}} \, \R[3]{z\ListAttr{C, B}} \, \CT[3]&\\
            \text{$\trans[4]$ (\DepositChecking):}& & & &  \R[4]{x} \, \UP[4]{z\ListAttr{C, B}\ListAttr{B}} \, \CT[4]&
         \end{array}
        $$
        \subcaption{\{Balance, DepositChecking, TransactSavings\} is not robust against \readmvcom.}\label{fig:counterex_SB_7}
    \end{minipage}
    \caption{Counterexamples for robustness against \mvrc for the SmallBank transaction templates. To facilitate readability, we only specify attributes for conflicting operations.}
    \label{fig:counterex_SB}
\end{figure*}


\subsection{\tpcckv{} Transaction Templates}

For the \tpcckv{} transaction templates given in Figure~\ref{fig:tpcc-abstract-syntax}, the corresponding SQL code is given in Figure~\ref{fig:tpcc:SQL}. For this set of templates, the maximal subsets robust against \mvrc are:
\begin{itemize}
    \item \{NewOrder, Payment, Delivery, StockLevel\}, and
    \item \{Payment, OrderStatus, StockLevel\}.
\end{itemize}
For each minimal subset not robust against \mvrc, a counterexample schedule is given in Figure~\ref{fig:counterex_TPCC}.

When analysing the \tpcckv{} transaction templates on the granularity of tuples instead of attributes, we get the following (smaller) subsets robust against \mvrc:
\begin{itemize}
    \item \{Payment, Delivery, StockLevel\},
    \item \{Payment, OrderStatus, StockLevel\}, and
    \item \{NewOrder, StockLevel\}.
\end{itemize}
The schedules given in Figure~\ref{fig:counterex_TPCC} immediately serve as counterexamples on the granularity of tuples, since the schedules in Figure~\ref{fig:counterex_TPCC} exhibit no dirty writes on the granularity of tuples. Counterexample schedules for the remaining minimal subsets not robust against \mvrc are given in Figure~\ref{fig:counterex_TPCC_tuples}.

\begin{figure*}
    \begin{minipage}[t]{\textwidth}
        $$\arraycolsep=.1em
        \begin{array}{lccccc}
            \text{$\trans[1]$ (OrderStatus):}& \R[1]{z} \, \R[1]{s\ListAttr{W,D,O,C,Sta}} & & \R[1]{v_1\ListAttr{$\alpha$}} \, \R[1]{v_2\ListAttr{$\alpha$}} \, \CT[1]\\
            \text{$\trans[2]$ (NewOrder):}& & \R[2]{x} \, \UP[2]{y} \, \R[2]{z} \, \W[2]{s\ListAttr{W,D,O,C,Sta}} \, \UP[2]{t_1} \, \W[2]{v_1\ListAttr{$\alpha$}} \, \UP[2]{t_2} \, \W[2]{v_2\ListAttr{$\alpha$}} \, \CT[2] &
         \end{array}
        $$
        \subcaption{\{NewOrder, OrderStatus\} is not robust against \mvrc. To shorten notation, we use $\alpha$ to denote the set of attributes $\ListAttr{W,D,O,OL,I,Del,Qua}$}\label{fig:counterex_TPCC_2}
    \end{minipage}
    \begin{minipage}[t]{\textwidth}
        $$\arraycolsep=.4em
        \begin{array}{lccccc}
            \text{$\trans[1]$ (OrderStatus):}& \R[1]{z\ListAttr{W,D,C,Inf,Bal}}  & & \R[1]{s} \, \R[1]{v_1\ListAttr{$\alpha$}} \, \R[1]{v_2\ListAttr{$\alpha$}} \, \CT[1]\\
            \text{$\trans[2]$ (Delivery):}& & \UP[2]{s} \, \UP[2]{v_1\ListAttr{$\beta$}\ListAttr{Del}} \, \UP[2]{v_2\ListAttr{$\beta$}\ListAttr{Del}} \, \UP[2]{z\ListAttr{W,D,C,Bal}\ListAttr{Bal}} \, \CT[2]&
         \end{array}
        $$
        \subcaption{\{OrderStatus, Delivery\} is not robust against \mvrc. To shorten notation, we use $\alpha$ and $\beta$ to denote respectively the sets of attributes $\ListAttr{W,D,O,OL,I,Del,Qua}$ and $\ListAttr{W,D,O,OL,Del}$}\label{fig:counterex_TPCC_4}
    \end{minipage}
    \caption{Counterexamples for robustness against \mvrc for the \tpcckv{} transaction templates. To facilitate readability, we only specify attributes for conflicting operations.}
    \label{fig:counterex_TPCC}
\end{figure*}


\begin{figure*}
    \begin{minipage}[t]{\textwidth}
        $$
        \begin{array}{lccccc}
            \text{$\trans[1]$ (NewOrder):}&  \R[1]{x}  & & \UP[1]{y} \, \R[1]{z} \, \W[1]{s} \, \UP[1]{t_1} \, \W[1]{v_1} \, \UP[1]{t_2} \, \W[1]{v_2} \, \CT[1]\\
            \text{$\trans[2]$ (Payment):}& & \UP[2]{x} \, \UP[2]{y} \, \UP[2]{z} \, \CT[2]&
         \end{array}
        $$
        \subcaption{\{NewOrder, Payment\} is not robust against \mvrc.}\label{fig:counterex_TPCC_1}
    \end{minipage}
    \begin{minipage}[t]{\textwidth}
        $$
        \begin{array}{lccccc}
            \text{$\trans[1]$ (NewOrder):}&  \R[1]{x} \, \UP[1]{y} \, \R[1]{z} & & \W[1]{s} \, \UP[1]{t_1} \, \W[1]{v_1} \, \UP[1]{t_2} \, \W[1]{v_2} \, \CT[1]\\
            \text{$\trans[2]$ (Delivery):}& & \UP[2]{s} \, \UP[2]{v_1} \, \UP[2]{v_2} \, \UP[2]{z} \, \CT[2]&
         \end{array}
        $$
        \subcaption{\{NewOrder, Delivery\} is not robust against \mvrc.}\label{fig:counterex_TPCC_3}
    \end{minipage}
    \caption{Counterexamples for robustness against \mvrc for the \tpcckv{} transaction templates when considering conflicts on the granularity of tuples instead of attributes. We omit attributes in our notation, as they are no longer important to decide conflict serializability.}
    \label{fig:counterex_TPCC_tuples}
\end{figure*}

\begin{figure*}
\begin{minipage}[t]{\textwidth/2-2ex}
\begin{verbatim}
NewOrder(WID, DID, CID, ITEMS):
    SELECT Info into :winfo
      FROM Warehouse
     WHERE WarehouseID = :WID;
        
    UPDATE District
       SET NextOrderID = NextOrderID + 1
     WHERE WarehouseID = :WID AND DistrictID = :DID
 RETURNING NextOrderID INTO :nid, Info INTO :dinfo;
     
    SELECT Info
      FROM Customer
     WHERE WarehouseID = :WID AND DistrictID = :DID
           AND CustomerID = :CID;
 
    INSERT INTO Orders
    VALUES (:WID, :DID, :nid, :CID, 'created');
 
 :ordline_id = 1;
 FOR :item_id, :quantity IN :ITEMS {
    UPDATE Stock
       SET Quantity = Quantity - :quantity
     WHERE WarehouseID = :WID AND ItemID = :item_id;
     
    INSERT INTO OrderLine
    VALUES (:WID, :DID, :nid, :ordline_id,
            :item_id, "created", ":quantity");
    
    :ordline_id += 1;
 }
 COMMIT;

Payment(WID, DID, CID, AMOUNT):
    UPDATE Warehouse
       SET YTD = YTD + :AMOUNT
     WHERE WarehouseID = :WID;
    
    UPDATE District
       SET YTD = YTD + :AMOUNT
     WHERE WarehouseID = :WID AND DistrictID = :DID;
    
    UPDATE Customer
       SET Balance = Balance + :AMOUNT
     WHERE WarehouseID = :WID AND DistrictID = :DID
           AND CustomerID = :CID;
 COMMIT;
\end{verbatim}
\end{minipage}   
\begin{minipage}[t]{\textwidth/2}
\begin{verbatim}
OrderStatus(WID, DID, CID, OID):
    SELECT Info INTO :cinfo, Balance INTO :balance
      FROM Customer
     WHERE WarehouseID = :WID AND DistrictID = :DID
           AND CustomerID = :CID;
     
    SELECT Status INTO :ostatus
      FROM Orders
     WHERE WarehouseID = :WID AND DistrictID = :DID
       AND OrderID = :OID;
     
    SELECT DeliveryInfo INTO :delinfos, Quantity INTO :quantities
      FROM OrderLine
     WHERE WarehouseID = :WID AND DistrictID = :DID
           AND OrderID = :OID AND;
 COMMIT;

Delivery(WID, DID, CID, OID, PRICE):
    UPDATE Orders
       SET Status = 'delivered'
     WHERE WarehouseID = :WID AND DistrictID = :DID
           AND OrderID = :OID;

    UPDATE OrderLine
       SET DeliveryInfo = 'delivered'
     WHERE WarehouseID = :WID AND DistrictID = :DID
           AND OrderID = :OID;

    UPDATE Customer
       SET Balance = Balance - :PRICE
     WHERE WarehouseID = :WID AND DistrictID = :DID
           AND CustomerID = :CID;
 COMMIT;

StockLevel(WID, IID):
   SELECT quantity INTO :quantity
     FROM Stock
    WHERE WarehouseID = :WID and ItemID = :IID;
 COMMIT;
\end{verbatim}
\end{minipage}
\caption{\tpcckv{} SQL Transaction Templates.}
\label{fig:tpcc:SQL}
\end{figure*}

\clearpage 

\section{Promoted Benchmarks}

This section provides a detailed overview of all required changes in both benchmarks to obtain robustness (cf. Section~\ref{sec:promotion}). For SmallBank, the required changes for RC+P(CS) and RC-\cite{DBLP:conf/aiccsa/AlomariF15} are presented respectively in Figure~\ref{fig:smallbank-rc_prom} and Figure~\ref{fig:smallbank-rc_ice}. For \tpcckv, the changes in each \shortptrans{} for RC+P(Attr), RC+P(Tup) and RC-\cite{DBLP:conf/aiccsa/AlomariF15} are given in Figure~\ref{fig:tpcc-RC_attr}, Figure~\ref{fig:tpcc-RC_tup} and Figure~\ref{fig:tpcc-RC_ICE}, respectively.

\begin{figure}[h]
\begin{minipage}[c]{0.45\linewidth}
    \centering\small

\begin{minipage}[t]{0.5\textwidth-2ex}
\Balance: 
\[
\begin{array}{l}
\R[]{\vx: \Account\ListAttr{N, C}}\\
\UP[]{\vy: \Savings\ListAttr{C, B}\ListAttr{B}}\\
\UP[]{\vz: \Checking\ListAttr{C, B}\ListAttr{B}}\\
\end{array}
\]
DepositChecking: 
\[
\begin{array}{l}
\R[]{\vx: \Account\ListAttr{N, C}}\\
\UP[]{\vz: \Checking\ListAttr{C, B}\ListAttr{B}}\\
\end{array}
\]
TransactSavings: 
\[
\begin{array}{l}
\R[]{\vx: \Account\ListAttr{N, C}}\\
\UP[]{\vy: \Savings\ListAttr{C, B}\ListAttr{B}}\\
\end{array}
\]
\end{minipage}%
\hfill%
\begin{minipage}[t]{0.50\textwidth-2ex}
Amalgamate: 
\[
\begin{array}{l}
\R[]{\vx_1: \Account\ListAttr{N, C}}\\
\R[]{\vx_2: \Account\ListAttr{N, C}}\\
\UP[]{\vy_1: \Savings\ListAttr{C, B}\ListAttr{B}}\\
\UP[]{\vz_1: \Checking\ListAttr{C, B}\ListAttr{B}}\\
\UP[]{\vz_2: \Checking\ListAttr{C, B}\ListAttr{B}}\\
\end{array}
\]
WriteCheck: 
\[
\begin{array}{l}
\R[]{\vx: \Account\ListAttr{N, C}}\\
\UP[]{\vy: \Savings\ListAttr{C, B}\ListAttr{B}}\\
\UP[]{\vz: \Checking\ListAttr{C, B}\ListAttr{B}}\\
\UP[]{\vz: \Checking\ListAttr{C, B}\ListAttr{B}}\\
\end{array}
\]
\end{minipage}

\removespacetocaption

    \caption{SmallBank \ptranss{} for RC+P(CS). In particular, all $\myR$-operations on \Balance and \Savings are promoted to $\myUP$-operations.
    }
    \label{fig:smallbank-rc_prom}

\end{minipage}%
\hfill%
\begin{minipage}[c]{0.45\linewidth}
    \centering\small

\begin{minipage}[t]{0.5\textwidth-2ex}
\Balance: 
\[
\begin{array}{l}
\W[]{\myvv: \Conflict\ListAttr{N}}\\
\R[]{\vx: \Account\ListAttr{N, C}}\\
\R[]{\vy: \Savings\ListAttr{C, B}}\\
\R[]{\vz: \Checking\ListAttr{C, B}}\\
\end{array}
\]
DepositChecking: 
\[
\begin{array}{l}
\W[]{\myvv: \Conflict\ListAttr{N}}\\
\R[]{\vx: \Account\ListAttr{N, C}}\\
\UP[]{\vz: \Checking\ListAttr{C, B}\ListAttr{B}}\\
\end{array}
\]
TransactSavings: 
\[
\begin{array}{l}
\W[]{\myvv: \Conflict\ListAttr{N}}\\
\R[]{\vx: \Account\ListAttr{N, C}}\\
\UP[]{\vy: \Savings\ListAttr{C, B}\ListAttr{B}}\\
\end{array}
\]
\end{minipage}%
\hfill%
\begin{minipage}[t]{0.50\textwidth-2ex}
Amalgamate: 
\[
\begin{array}{l}
\W[]{\myvv_1: \Conflict\ListAttr{N}}\\
\W[]{\myvv_2: \Conflict\ListAttr{N}}\\
\R[]{\vx_1: \Account\ListAttr{N, C}}\\
\R[]{\vx_2: \Account\ListAttr{N, C}}\\
\UP[]{\vy_1: \Savings\ListAttr{C, B}\ListAttr{B}}\\
\UP[]{\vz_1: \Checking\ListAttr{C, B}\ListAttr{B}}\\
\UP[]{\vz_2: \Checking\ListAttr{C, B}\ListAttr{B}}\\
\end{array}
\]
WriteCheck: 
\[
\begin{array}{l}
\W[]{\myvv: \Conflict\ListAttr{N}}\\
\R[]{\vx: \Account\ListAttr{N, C}}\\
\R[]{\vy: \Savings\ListAttr{C, B}}\\
\R[]{\vz: \Checking\ListAttr{C, B}}\\
\UP[]{\vz: \Checking\ListAttr{C, B}\ListAttr{B}}\\
\end{array}
\]
\end{minipage}

\removespacetocaption

    \caption{SmallBank \shortptranss{} for RC-\cite{DBLP:conf/aiccsa/AlomariF15}. Instances that assign the same tuple to variable $\vx$ (respectively $\vx_i$ in \Amalgamate), should assign the same tuple to variable $\myvv$ (respectively $\myvv_i$ in \Amalgamate) as well.
    }
    \label{fig:smallbank-rc_ice}

\end{minipage}

\end{figure}

\begin{figure}[h]
\begin{minipage}[c]{0.99\linewidth}
    \centering \small

\begin{minipage}[t]{0.37\textwidth-2ex}    
NewOrder:
\[
\begin{array}{l}
\R[]{\vx: \Warehouse\ListAttr{W, Inf}}\\
\UP[]{\vy: \District\ListAttr{W, D, Inf, N}\ListAttr{N}}\\
\R[]{\vz: \Customer\ListAttr{W, D, C, Inf}}\\
\W[]{\myvs: \Order\ListAttr{W, D O, C, Sta}}\\
\UP[]{\vt_1: \Stock\ListAttr{W, I, Qua}\ListAttr{Qua}}\\
\W[]{\myvv_1: \OrderLine\ListAttr{W, D, O, OL, I, Del, Qua}}\\
\UP[]{\vt_2: \Stock\ListAttr{W, I, Qua}\ListAttr{Qua}}\\
\W[]{\myvv_2: \OrderLine\ListAttr{W, D, O, OL, I, Del, Qua}}\\
\end{array}
\]
\end{minipage}
\begin{minipage}[t]{0.63\textwidth}
\begin{minipage}[t]{\textwidth/2-1ex}    
Delivery:
\[
\begin{array}{l}
\UP[]{\myvs: \Order\ListAttr{W, D, O}\ListAttr{Sta}}\\
\UP[]{\myvv_1: \OrderLine\ListAttr{W, D, O, OL, Del}\ListAttr{Del}}\\
\UP[]{\myvv_2: \OrderLine\ListAttr{W, D, O, OL, Del}\ListAttr{Del}}\\
\UP[]{\vz: \Customer\ListAttr{W, D, C, Bal}\ListAttr{Bal}}\\
\end{array}
\]
\end{minipage}
\begin{minipage}[t]{\textwidth/2-1ex}    
OrderStatus:
\[
\begin{array}{l}
\UP[]{\vz: \Customer\ListAttr{W, D, C, Inf, Bal}\ListAttr{Bal}}\\
\UP[]{\myvs: \Order\ListAttr{W, D, O, C, Sta}\ListAttr{Sta}}\\
\UP[]{\myvv_1: \OrderLine\ListAttr{W, D, O, OL, I, Del, Qua}\ListAttr{Del}}\\
\UP[]{\myvv_2: \OrderLine\ListAttr{W, D, O, OL, I, Del, Qua}\ListAttr{Del}}\\
\end{array}
\]
\end{minipage}

\medskip

\begin{minipage}[t]{\textwidth/2-1ex}    
Payment:
\[
\begin{array}{l}
\UP[]{\vx: \Warehouse\ListAttr{W, YTD}\ListAttr{YTD}}\\
\UP[]{\vy: \District\ListAttr{W, D, YTD}\ListAttr{YTD}}\\
\UP[]{\vz: \Customer\ListAttr{W, D, C, Bal}\ListAttr{Bal}}\\
\end{array}
\]
\end{minipage}
\begin{minipage}[t]{\textwidth/2-1ex}
StockLevel:
\[
\begin{array}{l}
\R[]{\vt: \Stock\ListAttr{W, I, Qua}}\\
\end{array}
\]
\end{minipage}
\end{minipage}

\removespacetocaption

    \caption{\tpcckv{} \ptranss{} for RC+P(Attr). In particular, all operations in OrderStatus are promoted. Attribute names are abbreviated.}
    \label{fig:tpcc-RC_attr}
\end{minipage}
\end{figure}

\begin{figure}[h]
\begin{minipage}[c]{0.99\linewidth}
    \centering \small

\begin{minipage}[t]{0.37\textwidth-2ex}    
NewOrder:
\[
\begin{array}{l}
\UP[]{\vx: \Warehouse\ListAttr{W, Inf}\ListAttr{Inf}}\\
\UP[]{\vy: \District\ListAttr{W, D, Inf, N}\ListAttr{N}}\\
\UP[]{\vz: \Customer\ListAttr{W, D, C, Inf}\ListAttr{Inf}}\\
\W[]{\myvs: \Order\ListAttr{W, D O, C, Sta}}\\
\UP[]{\vt_1: \Stock\ListAttr{W, I, Qua}\ListAttr{Qua}}\\
\W[]{\myvv_1: \OrderLine\ListAttr{W, D, O, OL, I, Del, Qua}}\\
\UP[]{\vt_2: \Stock\ListAttr{W, I, Qua}\ListAttr{Qua}}\\
\W[]{\myvv_2: \OrderLine\ListAttr{W, D, O, OL, I, Del, Qua}}\\
\end{array}
\]
\end{minipage}
\begin{minipage}[t]{0.63\textwidth}
\begin{minipage}[t]{\textwidth/2-1ex}    
Delivery:
\[
\begin{array}{l}
\UP[]{\myvs: \Order\ListAttr{W, D, O}\ListAttr{Sta}}\\
\UP[]{\myvv_1: \OrderLine\ListAttr{W, D, O, OL, Del}\ListAttr{Del}}\\
\UP[]{\myvv_2: \OrderLine\ListAttr{W, D, O, OL, Del}\ListAttr{Del}}\\
\UP[]{\vz: \Customer\ListAttr{W, D, C, Bal}\ListAttr{Bal}}\\
\end{array}
\]
\end{minipage}
\begin{minipage}[t]{\textwidth/2-1ex}    
OrderStatus:
\[
\begin{array}{l}
\UP[]{\vz: \Customer\ListAttr{W, D, C, Inf, Bal}\ListAttr{Bal}}\\
\UP[]{\myvs: \Order\ListAttr{W, D, O, C, Sta}\ListAttr{Sta}}\\
\UP[]{\myvv_1: \OrderLine\ListAttr{W, D, O, OL, I, Del, Qua}\ListAttr{Del}}\\
\UP[]{\myvv_2: \OrderLine\ListAttr{W, D, O, OL, I, Del, Qua}\ListAttr{Del}}\\
\end{array}
\]
\end{minipage}

\medskip

\begin{minipage}[t]{\textwidth/2-1ex}    
Payment:
\[
\begin{array}{l}
\UP[]{\vx: \Warehouse\ListAttr{W, YTD}\ListAttr{YTD}}\\
\UP[]{\vy: \District\ListAttr{W, D, YTD}\ListAttr{YTD}}\\
\UP[]{\vz: \Customer\ListAttr{W, D, C, Bal}\ListAttr{Bal}}\\
\end{array}
\]
\end{minipage}
\begin{minipage}[t]{\textwidth/2-1ex}
StockLevel:
\[
\begin{array}{l}
\R[]{\vt: \Stock\ListAttr{W, I, Qua}}\\
\end{array}
\]
\end{minipage}
\end{minipage}

\removespacetocaption

    \caption{\tpcckv{} \ptranss{} for RC+P(Tup). In particular, all \myR-operations on tuples over the Warehouse-, Customer- Order- and OrderLine-relations need to be promoted to \myUP-operations, requiring changes in both NewOrder and OrderStatus. Attribute names are abbreviated.}
    \label{fig:tpcc-RC_tup}
\end{minipage}
\end{figure}

\newcommand{\myvp}{\mathtt{P}}
\newcommand{\myvq}{\mathtt{Q}}

\begin{figure}[h]
\begin{minipage}[c]{0.99\linewidth}
    \centering \small

\begin{minipage}[t]{0.37\textwidth-2ex}    
NewOrder:
\[
\begin{array}{l}
\W[]{\myvp: \Conflict\ListAttr{W, D, C}}\\
\W[]{\myvq_1: \Conflict\ListAttr{W, I}}\\
\W[]{\myvq_2: \Conflict\ListAttr{W, I}}\\
\R[]{\vx: \Warehouse\ListAttr{W, Inf}}\\
\UP[]{\vy: \District\ListAttr{W, D, Inf, N}\ListAttr{N}}\\
\R[]{\vz: \Customer\ListAttr{W, D, C, Inf}}\\
\W[]{\myvs: \Order\ListAttr{W, D O, C, Sta}}\\
\UP[]{\vt_1: \Stock\ListAttr{W, I, Qua}\ListAttr{Qua}}\\
\W[]{\myvv_1: \OrderLine\ListAttr{W, D, O, OL, I, Del, Qua}}\\
\UP[]{\vt_2: \Stock\ListAttr{W, I, Qua}\ListAttr{Qua}}\\
\W[]{\myvv_2: \OrderLine\ListAttr{W, D, O, OL, I, Del, Qua}}\\
\end{array}
\]
\end{minipage}
\begin{minipage}[t]{0.63\textwidth}
\begin{minipage}[t]{\textwidth/2-1ex}    
Delivery:
\[
\begin{array}{l}
\W[]{\myvp: \Conflict\ListAttr{W, D, C}}\\
\UP[]{\myvs: \Order\ListAttr{W, D, O}\ListAttr{Sta}}\\
\UP[]{\myvv_1: \OrderLine\ListAttr{W, D, O, OL, Del}\ListAttr{Del}}\\
\UP[]{\myvv_2: \OrderLine\ListAttr{W, D, O, OL, Del}\ListAttr{Del}}\\
\UP[]{\vz: \Customer\ListAttr{W, D, C, Bal}\ListAttr{Bal}}\\
\end{array}
\]
\end{minipage}
\begin{minipage}[t]{\textwidth/2-1ex}    
OrderStatus:
\[
\begin{array}{l}
\W[]{\myvp: \Conflict\ListAttr{W, D, C}}\\
\R[]{\vz: \Customer\ListAttr{W, D, C, Inf, Bal}}\\
\R[]{\myvs: \Order\ListAttr{W, D, O, C, Sta}}\\
\R[]{\myvv_1: \OrderLine\ListAttr{W, D, O, OL, I, Del, Qua}}\\
\R[]{\myvv_2: \OrderLine\ListAttr{W, D, O, OL, I, Del, Qua}}\\
\end{array}
\]
\end{minipage}

\medskip

\begin{minipage}[t]{\textwidth/2-1ex}    
Payment:
\[
\begin{array}{l}
\W[]{\myvp: \Conflict\ListAttr{W, D, C}}\\
\UP[]{\vx: \Warehouse\ListAttr{W, YTD}\ListAttr{YTD}}\\
\UP[]{\vy: \District\ListAttr{W, D, YTD}\ListAttr{YTD}}\\
\UP[]{\vz: \Customer\ListAttr{W, D, C, Bal}\ListAttr{Bal}}\\
\end{array}
\]
\end{minipage}
\begin{minipage}[t]{\textwidth/2-1ex}
StockLevel:
\[
\begin{array}{l}
\W[]{\myvq: \Conflict\ListAttr{W, I}}\\
\R[]{\vt: \Stock\ListAttr{W, I, Qua}}\\
\end{array}
\]
\end{minipage}
\end{minipage}

\removespacetocaption

    \caption{\tpcckv{} \ptranss{} for RC-\cite{DBLP:conf/aiccsa/AlomariF15}. Instances that assign the same tuple to variable $\vz$ (i.e. a tuple of type \Customer), should assign the same tuple to variable $\myvp$ as well. Analogously, Instances that assign the same tuple to variable $\vt$ (i.e. a tuple of type \Stock), should assign the same tuple to variable $\myvq$ as well. Attribute names are abbreviated.}
    \label{fig:tpcc-RC_ICE}
\end{minipage}
\end{figure}

\clearpage 
\section{Proofs for Section~\ref{sec:robustness-transactions}}
\subsection{Proof for Theorem~\ref{theo:characterization:split-shedules}}

    \emph{(\ref{theo:char:split1} $\to$ \ref{theo:char:split3})} Assume $\transset$ is not robust against $\mvrc$. Then there is a schedule $\schedule$ over $\transset$ allowed under \mvrc with a cycle $C$ in $\cg{\schedule}$.
    We next construct a multiversion split schedule $\schedule'$ based on a sequence $C'$ of conflict quadruples as defined in Definition~\ref{def:mvsplitschedule}. Without loss of generality, we assume that $C$ is a minimal cycle in $\cg{\schedule}$. Let $\trans[1], \trans[2], \ldots \trans[m]$ be the transactions in the order that they appear in $C$, such that $\trans[2]$ is the transaction (among those in $C$) that commits first in $\schedule$.
    In other words, for every transaction $\trans[i] \in C$ different from $\trans[2]$, $\CT[2] <_\schedule \CT[i]$.
    Let $$C' = (T_1, b_1, a_2, T_2), (T_2, b_2, a_3, T_3), \ldots, (T_m,b_m, a_1, T_1)$$ be a sequence of conflict quadruples where for each conflict quadruple $(\trans[i], b_i, a_{i+1}, \trans[i+1])$, we have that $a_{i+1}$ depends on $b_i$ in $\schedule$, that is, $b_i \to_s a_{i+1}$.
    Notice that, since there is an edge from $\trans[i]$ to $\trans[i+1]$ in $\cg{\schedule}$, we can always find such a pair of operations. We take $T_{m+1}$ to be $T_1$.
    We now show that the multiversion split schedule $\schedule'$ based on $C'$ satisfies the conditions in Definition~\ref{def:mvsplitschedule}. 
  
  \noindent
  (Condition \ref{def:mvsplitschedule}.\ref{c:3})
    We assumed that $\CT[2] <_\schedule \CT[1]$. As $\schedule$ is allowed under \mvrc, the existence of a wr- or a ww-dependency from $b_1$ to $a_2$ would imply that $\CT[1] <_\schedule a_2 <_\schedule \CT[2]$. 
    Therefore, $b_1\to_s a_2$ is an rw-antidependency from $b_1$ to $a_2$.
    As a result, {$b_1 <_s \CT[2]$, and}
    $b_1$ and $a_2$ are {rw-conflicting}.
    
    \noindent
    (Condition \ref{def:mvsplitschedule}.\ref{c:1})
    Next, we prove that there is no {ww-conflict between a write operation in $\prefix{\trans[1]}{b_1}$ and a write operation in any of the transactions $\trans[2], \ldots, \trans[m]$.}
    Towards a contradiction, assume that there is a transaction $\trans[i]$ with a write operation $c_i$, {ww-}conflicting with a write operation $c_1$ in $\prefix{\trans[1]}{b_1}$.
    Notice that $c_1 <_\schedule c_i$, as otherwise $c_i <_\schedule \CT[i] <_\schedule c_1 \leq_\schedule b_1 <_\schedule \CT[2]$, contradicting our assumption that $\trans[2]$ commits first.
    Moreover, $\trans[1]$ commits before $c_i$ in $\schedule$, as otherwise $c_1$ and $c_i$ would imply a dirty write.
    Since $\dependson{c_1}{c_i}$ and $C$ is a minimal cycle in $\cg{\schedule}$, it immediately follows that $\trans[i] = \trans[2]$.
    But then $\trans[1]$ commits before $\trans[2]$ in $\schedule$, leading to the desired contradiction.
    
    \noindent
    (Condition \ref{def:mvsplitschedule}.\ref{c:2})
    The last condition to verify is that $b_1 <_{\trans[1]} a_1$ or {$b_m$ and $a_1$ are rw-conflicting}. Towards a contradiction, assume that $a_1 \leq_{\trans[1]} b_1$ \emph{and} {$b_m$ and $a_1$ are not rw-conflicting}.
    We argued above that $\prefix{\trans[1]}{b_1}$ cannot contain a write operation {ww-}conflicting with a write operation in $\trans[m]$. {Therefore, $b_m$ and $a_1$ must be wr-conflicting, and $\dependson{b_m}{a_1}$ is a wr-dependency.}
    Since $\schedule$ is allowed under \mvrc, it follows that $b_m <_\schedule \CT[m] <_\schedule a_1 \leq_\schedule b_1 <_\schedule \CT[2]$, contradicting our assumption that $\trans[2]$ commits first in $\schedule$.
    
    \medskip
    \emph{(\ref{theo:char:split3} $\to$ \ref{theo:char:split1})}
    {
    Let $\schedule$ be a multiversion split schedule for $\transset$ based on $C = (T_1, b_1, a_2, T_2), (T_2, b_2, a_3, T_3), \ldots, (T_m,b_m, a_1, T_1)$ consisting of conflicting quadruples. We can assume that $\schedule$ is read-last-committed. Otherwise,  choosing an appropriate {version order $\ll_\schedule$ and} version function $v_\schedule$. 
    Notice that {$\ll_\schedule$ and} $v_\schedule$ have no influence on the conflict quadruples in $C$. 
    
    First, we show that schedule $\schedule$ is  allowed under \mvrc (c.f.~Definition~\ref{def:isolationlevels}).  We only need to show that $\schedule$ exhibits no dirty writes.  For this, let $b_i$ and $a_j$ be two arbitrary {ww-conflicting} operations in $\schedule$ in two
    different transactions $T_i$ and $T_j$, with $b_i <_\schedule a_j$. If $i > 1$ or $j >m$, it follows from the definition of multiversion split schedule that $b_i <_\schedule \CT[i] <_\schedule a_j$. For $i=1$ and $j \le m$, dirty writes are forbidden by condition (\ref{c:1}) of Definition~\ref{def:mvsplitschedule}.
    }
    
    {It remains} to show that $\schedule$ is not conflict serializable.
    To this end, we argue that for each conflicting quadruple $(\trans[i], b_i, a_j, \trans[j])$ in $C$, the operation $a_j$ depends on $b_i$ in $\schedule$, that is, $b_i\to_s a_j$, thereby showing that the transactions in $C$ represent a cycle in $\cg{\schedule}$.
    If both $\trans[i]$ and $\trans[j]$ are different from $\trans[1]$, then $b_i <_\schedule \CT[i] <_\schedule a_j$ by construction of $\schedule$.
    Since $\schedule$ is read-last-committed, it immediately follows that $\dependson{b_i}{a_j}$, independent of whether $a_j$ and $b_i$ are {rw-, wr- or ww-conflicting}.
    If $\trans[i] = \trans[1]$, then {$b_i = b_1$ and $a_j = a_2$ are rw-conflicting} by Definition~\ref{def:mvsplitschedule}.
    Since $b_1 <_\schedule a_2$ implies that $op_0=v_\schedule(b_1) {\ll_\schedule} a_2$, we obtain an rw-antidependency from $b_1$ to $a_2$.
    
    Lastly, if $\trans[j] = \trans[1]$, then $a_j = a_1$ and $b_i = b_m$. 
    By Definition~\ref{def:mvsplitschedule}, $b_1 <_\schedule a_1$ or {$b_m$ and $a_1$ are rw-conflicting}.
    In the former case, we have that $b_m <_\schedule \CT[m] <_\schedule a_1$, again implying that $\dependson{b_m}{a_1}$.
    {In the latter case, $b_m$ is a read operation on a tuple $\x$ where $v_\schedule(b_m)$ is either $\sstart$ or the write operation on $\x$ that committed last before $b_m$. In both cases, $v_\schedule(b_m) \ll_\schedule a_1$, since $b_m <_\schedule \CT[1]$ and $\ll_\schedule$ coincides with the commit order in $<_\schedule$. The rw-antidependency from $b_m$ to $a_1$ now follows immediately.}

\subsection{Proof for Theorem~\ref{theo:robustness-transactions:complexity}}

Intuitively, Algorithm~\ref{alg:transaction_robustness} applies Theorem~\ref{theo:characterization:split-shedules} and checks whether a multiversion split schedule over $\transset$ exists.
We first argue that Algorithm~\ref{alg:transaction_robustness} is correct, followed by the complexity analysis.
 
\paragraph{Correctness} Assume $\transset$ is not robust against \MVRC. By Theorem~\ref{theo:characterization:split-shedules}, a multiversion split schedule $\schedule$ for $\transset$ based on some $$C = (T_1', b_1', a_2', T_2'), (T_2', b_2', a_3', T_3'), \ldots, (T_m',b_m', a_1', T_1')$$ exists. We argue that Algorithm~\ref{alg:transaction_robustness} returns False.
To this end, assume $\trans[1]$ and $b_1$ in Algorithm~\ref{alg:transaction_robustness} are instantiated by $\trans[1]'$ and $b_1'$, respectively.
Then, there is a path from $\trans[2]'$ to $\trans[m]'$ in $\prefixfreegraph(b_1,T_1, \transset \setminus \{\trans[1]\})$, witnessed by the conflicts in $C$.
Indeed, by Definition~\ref{def:mvsplitschedule}, transactions $\trans[2]', \trans[3]',\ldots,\trans[m]'$ are not ww-conflicting with $\prefix{T_1}{b_1}$.
As a result, $(\trans[2], \trans[m])$ is in $TC$ if we instantiate $\trans[2]$ and $\trans[m]$ by $\trans[2]'$ and $\trans[m]'$, respectively.
If we take $a_1 = a_1'$, $a_2 = a_2'$ and $b_m = b_m'$, the condition in the if-test of Algorithm~\ref{alg:transaction_robustness} is immediate by Definition~\ref{def:mvsplitschedule}, implying that the algorithm correctly returns False.

It remains to argue that Algorithm~\ref{alg:transaction_robustness} returns True when $\transset$ is robust against $\mvrc$.
Towards a contradiction, assume Algorithm~\ref{alg:transaction_robustness} returns False instead, witnessed by transactions $\trans[1], \trans[2], \trans[m] \in \transset$ and operations $b_1, a_1 \in \trans[1]$, $a_2 \in \trans[2]$ and $b_m \in \trans[m]$.
Let $$C = (T_2, b_2, a_3, T_3), (T_3, b_3, a_4, T_4), \ldots, (\trans[m-1],b_{m-1}, a_m, T_m)$$ be the sequence of conflict quadruples witnessing the path from $\trans[2]$ to $\trans[m]$ in $\prefixfreegraph(b_1,T_1, \transset \setminus \{\trans[1]\})$ (notice that $C$ can be the empty sequence in the special case that $\trans[2] = \trans[m]$).
Then, the multiversion split schedule $\schedule$ for $\transset$ based on $$C' = (T_1, b_1, a_2, T_2), C, (\trans[m],b_{m}, a_1, T_1)$$ is a valid multiversion split schedule.
Indeed, the transactions $\trans[2], \ldots, \trans[m]$ do not contain a ww-conflict with an operation in $\prefix{T_1}{b_1}$ by definition of $\prefixfreegraph(b_1,T_1, \transset)$, and the remaining conditions of Definition~\ref{def:mvsplitschedule} are immediate by the if-test in Algorithm~\ref{alg:transaction_robustness}.
According to Theorem~\ref{theo:characterization:split-shedules}, this schedule $\schedule$ contradicts our assumption that $\transset$ is robust against \MVRC.
 
\paragraph{Complexity} Let $k$ be the total number of operations in $\transset$ and $\ell$ the maximum number of operations in a transaction in $\transset$. The two outer for-loops in Algorithm~\ref{alg:transaction_robustness} iterate over all read operations in $\transset$, so there are at most $k$ iterations. Each such iteration consists of three steps: constructing the prefix-conflict-free-graph $G$, computing the reflexive-transitive-closure $TC$ over $G$, and checking a specific condition over the pairs of transactions in $TC$.

The construction of $G$ requires us to verify for each transaction in $\transset \setminus \{\trans[1]\}$ whether it has an operation that is ww-conflicting with an operation in $\prefix{\trans[1]}{b_1}$. We add each such transaction as a node to $G$, and add edges to other transactions in $G$ if they have conflicting operations. Both parts can be done in time ${O}(\ell^2)$. The computation of $TC$ over $G$ can be achieved in time ${O}(|\transset|^3)$ by an application of the Floyd-Warshall algorithm.

The third step checks a specific condition over pairs of transactions in $TC$. Worst case, $TC$ is the complete graph, and the condition will iterate over all triples of operations $(a_1, a_2, b_m)$, with $a_1$ an operation in $\trans[1]$, and $a_2$ and $b_m$ operations in two other transactions occurring in $G$. Therefore, this third step can be done in time ${O}(\ell . k^2)$.

By combining the results above, and since $l \leq k$, we get that
Algorithm~\ref{alg:transaction_robustness} decides whether $\transset$ is robust against \MVRC in
time ${O}(\text{max}\{k.|\transset|^3, k^3.\ell\})$.


\section{Proofs for Section~\ref{sec:robustness-templates}}

\subsection{Proof for Theorem~\ref{theo:ptime:noconstraints}}

{First, we show in Lemma~\ref{lemma:noconstraints:limitedtuples} that for each $D$, we only need to consider one mapping $\bar \mu$ of a canonical form
that partitions the mapped variables into three or four disjoint sets: all variables connected to $o_1$ (in a way to be made precise next), all variables connected to $p_1$ (when $o_1$ and $p_1$ are themselves connected, the two sets coincide), all variables in $\tau_1$ not in the previous two sets, and all remaining variables in all other templates. 
Furthermore, at most four different tuples for each variable type are needed.}
We need the following notion: a variable $\vx$ in $\tau_i$ is \emph{connected} to an operation $o$ in $\tau_j$ in $D$ if either $i=j$ and $\vx$ is the variable of operation $o$; there is a potentially conflicting quadruple $(\tau_i, o', p', \tau_j)$ with $o'$ having the same variable as $o$ and $p'$ having variable $\vx$; or $\vx$ is connected to an operation whose variable is connected to $o$.

We  encode the choice of tuples for variables through (total) functions $c:\schRel\to\objects$ that we call \emph{type mappings} and which map a relation onto a particular tuple of that relation's type.
The canonical mapping $\canmu$ for $D = (\tau_1, o_1, \allowbreak p_2,\allowbreak \tau_2), \ldots, (\tau_m, o_m, p_1, \tau_1)$ is defined relative to four type mappings $c_1$, $c_2$, $c_3$, and $c_4$, whose ranges do not matter as long as they are all different. Then $\canmu$ consists of the following set of $m$ variable mappings $\mu_i$ for occurrences $\tau_i$ of \ptrans{} in $D$. For $\mu_1$,
\[
    \mu_1(\vx) = \left\{
        \begin{array}{ll}
            c_1(\type{\vx}) & \text{if $\vx$ is the variable of $o_1$,} \\
            c_2(\type{\vx}) & \text{if $\vx$ is connected to $p_1$ and not to $o_1$,} \\ 
            c_4(\type{\vx}) & \text{otherwise.} 
        \end{array}\right.
\]
For every $1<i\leq m$, 
\[
    \mu_i(\vx) = \left\{
        \begin{array}{ll}
            c_1(\type{\vx}) & \text{if $\vx$ is connected to $o_1$}, \\
            c_2(\type{\vx}) & \text{if $\vx$ is connected to $p_1$ and not to $o_1$}, \\
            c_3(\type{\vx}) & \text{otherwise.} 
        \end{array}\right.
\]

\begin{lemma}\label{lemma:noconstraints:limitedtuples}
    Let $\workload$ be a set of \ptranss. The following are equivalent:
    \begin{itemize}
        \item $\workload$ is not robust against \mvrc;
        \item there is a multiversion split schedule $\schedule$ for some $C$ over a set of transactions $\transset$ consistent with $\workload$ and a database $\db$, where $C$ is induced by a sequence of potentially conflicting quadruples $D$ over $\workload$ and its canonical variable mapping.
    \end{itemize}
    Furthermore, for every sequence of potentially conflicting quadruples $D$ over $\workload$ and every variable mapping $\bar{\mu}$ for $D$, there is a database $\db$ where the transactions in the induced sequence of conflicting quadruples are consistent with.
\end{lemma}

\begin{proof}
    First, we observe that for a sequence of potentially conflicting quadruples $D$ for a set $\workload$ and a variable mapping $\bar{\mu}$ for $D$, there always exists a database $\db$ such that the transactions in the sequences of conflicting quadruples $C$ induced by $D$ and $\bar{\mu}$ are consistent with $\workload$ and $\db$. Consistency with $\workload$ is immediate. As $\db$ we can take the database that contains (in its respective relations) all tuples $\mu_i(\vx)$ for every variable $\vx$ in a \ptrans{} $\tau_i$ in $D$ with $\mu_i$ the variable mapping that $\bar{\mu}$ has assigned to $\tau_i$.

    Second, observe that if a variable $\vx$ in some \ptrans{} occurrence $\tau_i$ in a sequence of potentially conflicting quadruples $D$ is connected to an operation $o$ in some (not necessarily different) \ptrans{} occurrence $\tau_j$ in $D$, then every mapping $\bar{\mu}$ for $D$ assigns the same tuple to $\vx$ in $\tau_i$ and the variable of operation $o$ in $\tau_j$.

    \smallskip
    \noindent
    (2 $\Rightarrow$ 1) Is a direct result of Definition~\ref{def:template_robustness} and Theorem~\ref{theo:characterization:split-shedules}.

    \smallskip
    \noindent
    (1 $\Rightarrow$ 2) There is a multiversion split schedule $\schedule$ for $C$ over a set $\transset$ consistent with $\workload$, due to Definition~\ref{def:template_robustness} and Theorem~\ref{theo:characterization:split-shedules}. From $C=(\trans[1], b_1, a_2, \trans[2]), \ldots, (\trans[m], b_m, a_1, \tau_1)$ we can derive a sequence of potentially conflicting quadruples $D = (\tau_1, o_1, p_2, \tau_2), \ldots, (\tau_m, o_m, p_1, \tau_1)$ and a variable mapping $\bar{\mu}$ for $D$ with variable mappings $\mu_i$ for every $\tau_i$, such that $\mu_i(\tau_i) = \trans[i]$, $\mu_i(o_i) = b_i$, and $\mu_i(p_i) = a_i$.

    We claim that the canonical variable mapping $\canmu$ for $D$ induces a sequence of conflicting quadruples $C^c$ and thus a schedule $\schedule'$ for $C^c$ as in Definition~\ref{def:mvsplitschedule}. Since the transactions in $C^c$ are consistent with $\workload$, and we already showed that for every variable mapping for $D$ (including $\canmu$) there exists a database $\db$ where these transactions are also consistent with, we only have to show that schedule $\schedule$ has properties $(1-3)$ of Definition~\ref{def:mvsplitschedule}. In the below argument, we write $\mu'_i$ to denote the variable mapping for \ptrans{} occurrence $\tau_i$ in $D$ implied by $\canmu$. 

    Condition~(1) requires most explanation. Therefore, towards a contradiction, let us assume that Condition~(1) is not true for $\schedule'$. Then there is a write operation in the prefix of $\mu'_1(\tau_1)$ (say with variable $\vx$) that is ww-conflicting with a write operation in another transaction $\mu'_j(\tau_j)$ in $\schedule'$, say with variable $\vy$ in the respective operation. The definition of $\canmu$ for $D$ implies that $\mu_1'(\vx) \in \{c_1,c_2\}$ and $\mu_j'(\vy) \in \{c_1,c_2,c_3\}$. More precisely, by the assumption $\mu_1'(\vx) = \mu'_j(\tau_j)$, we have that $\mu_j'(\vy) \in \{c_1,c_2\}$ implying (again by definition of $\canmu$) that $\vy$ is connected to either $o_1$ or $p_1$ in $\tau_1$. That latter means that also $\mu_1(\vx) = \mu_j(\vy)$ and thus that there is a write operation in the prefix of $\mu_1(\tau)$ that is ww-conflicting with a write operation in transaction $\mu_j(\tau_j)$ in $\schedule$, which contradicts that $\schedule$ is a multiversion split schedule.

    Condition~(2) and Condition~(3) are based on the type of operations, which are fixed in $D$ and thus shared between $C$ and $C^c$. Particularly, for Condition~(2) we have that $\mu_1(o_1) <_{\mu_1(\tau_1)}\mu_1(p_1)$ or $\mu_m(o_m)$ is rw-conflicting with $\mu_1(p_1)$, due to $\schedule$ being a multiversion split schedule, from which follows that $o_1 <_{\tau_1} p_1$ or $o_m$ is potentially rw-conflicting with $p_1$. Since the variable of $o_m$ in $\tau_m$ is connected to $p_1$ in $\tau_1$ it follows that $\mu_1'(o_1) <_{\mu_1'(\tau_1)}\mu_1'(p_1)$ or $\mu_m'(o_m)$ is rw-conflicting with $\mu_1'(p_1)$, thus that Condition~(2) is indeed true for $\schedule'$ as well. Condition~(3) follows similarly, as $\mu_1(o_1)$ is rw-conflicting with $\mu_2(p_2)$ due to $\schedule$ being a multiversion split schedule, implying that $o_1$ is potentially rw-conflicting with $p_2$. Since the variable of $p_2$ is also connected to $o_1$ in $D$, we have that $\mu_1'(p_1)$ is rw-conflicting with $\mu_2'(p_2)$ and thus that Condition~(3) is true in $\schedule'$, which concludes the proof.
\end{proof}

\begin{example}
    We provide a more elaborate example justifying the need for exactly four tuples of the same type in a counterexample. Consider the set of \ptranss{} $\workload = \{\tau_1, \tau_2\}$ with
    \begin{align*}
        \tau_1:& \ \W[1]{Y: S\ListAttr{B}} \, \W[1]{Z: S\ListAttr{A}} \, \W[1]{X: S\ListAttr{A,B}} \, \R[1]{Y: S\ListAttr{A}} \, \W[1]{Z: S\ListAttr{B}},\\
        \tau_2:& \ \W[2]{X: S\ListAttr{A,B}} \, \W[2]{Y: S\ListAttr{A}} \, \W[2]{Z: S\ListAttr{B}},
    \end{align*}
    and let $D$ be the sequence of potentially conflicting quadruples
    $$(\tau_1, \R[1]{Y: S\ListAttr{A}}, \W[2]{Y: S\ListAttr{A}}, \tau_2), (\tau_2, \W[2]{Z: S\ListAttr{B}}, \W[1]{Z: S\ListAttr{B}}, \tau_1).$$
    Then, the multiversion split schedule $\schedule_2$ based on the sequence of conflict quadruples $C$ induced by $D$ and its canonical variable mapping is as follows (we assume $c_i(S) = \x_i$ for $i \in \{1,2,3,4\}$):
    \begin{center}
    \begin{tabular}{l @{\hspace{.2em}} l @{\hspace{.2em}} l}
        $\trans[1]: \, \W[1]{\x_1} \, \W[1]{\x_2} \, \W[1]{\x_4} \, \R[1]{\x_1}$ & & $\W[1]{\x_2} \, \CT[1]$\\
        $\trans[2]:$ & $\W[2]{\x_3} \, \W[2]{\x_1} \, \W[2]{\x_2} \, \CT[2]$ & 
    \end{tabular}
    \end{center}
    There are no dirty writes in $\schedule_2$, as the write operations on $\x_1$ and $\x_2$ in $\prefix{\trans[1]}{\R[1]{\x_1}}$ write to attributes disjoint from the write operations on $\x_1$ and $\x_2$ in $\trans[2]$. It is not possible to construct this schedule with less than four tuples, as trying to replace any two tuples $t_i$ and $t_j$ with one tuple leads to a dirty write invalidating the schedule under RC. \hfill $\Box$
\end{example}

To cycle through all possible sequences $D$, Algorithm~\ref{alg:ptime:template} iterates over the possible split \ptranss{} $\tau_1\in \workload$ and its possible operations $o_1,p_1\in\tau_1$, and relies on a graph referred to as 
$\prefixfreegraphB({o_1},\allowbreak {p_1}, h, \tau_1, \workload)$.
Here, $h\in\{1,2\}$ signals that the prefix and suffix of the split of $\tau_1$ use the same type mapping $c_1$ when $h=1$ and that the suffix uses type mapping $c_2$ when $h=2$.
The graph has as nodes the quadruples $(\tau, o, i, j)$ with $\tau \in \workload$, $o \in \tau$, $i\in\{1,2,3\}$ and $j \in \{\text{in},\text{out}\}$.
Here, $i\in\{1,2,3\}$ encodes that $o$ is assigned the type mapping $c_i$ in $\tau$ (the type mapping $c_4$ is not used). There will be two types of edges: (1) inner edges $(\tau,o,i,\text{in}) \to (\tau, p,i',\text{out})$ that stay within the same transaction $\tau$ and indicate how the type mapping changes (or stays the same) from $c_i$ for $o$ to $c_{i'}$ for $p$; and (2) outer edges $(\tau,o,i,\text{out}) \to (\tau', p,i',\text{in})$ between different occurrences of \ptranss{} encoding a potentially conflicting quadruple $(\tau,o,p,\tau')$ and maintaining information on type mappings as well.

More formally, a quadruple node $(\tau, o, i, j)$ in the graph satisfies the following properties:
\begin{itemize}
\item[(a)] $i=1$ implies that there is no operation $o_1' \in \prefix{\tau_1}{o_1}$ over the same variable as ${o}_1$ in
$\tau_1$ s.t.\ $o_1'$ is potentially ww-conflicting with an operation over the same variable as $o$ in $\tau$.
\item[(b)] $i=h$ implies that there is no operation $o_1' \in \prefix{\tau_1}{o_1}$ over the same variable as ${p}_1$ in
    $\tau_1$ s.t.\ $o_1'$ is potentially ww-conflicting with an operation over the same variable as $o$ in $\tau$. 
\end{itemize}
Conditions (a) and (b) on the nodes, ensure that condition (1) of
Definition~\ref{def:mvsplitschedule} is always guaranteed for all possible variable mappings that are consistent with the particular choice of type mapping.
Furthermore, two nodes $(\tau, o, i, j)$ and $(\tau', o', i', j')$ are connected by a directed edge if either
\begin{itemize}
    \item[$(\dagger)$] $\tau = \tau'$, $j = \text{in}$, $j' = \text{out}$, and
    if $o$ and $o'$ are over the same variable in $\tau$, then $i = i'$
        (i.e., remain within the same transaction and
        change the  {type mapping} only
        when $o$ and $o'$ are not over the same variable); or,
    \item[$(\ddagger)$] $j = \text{out}$, $j' = \text{in}$, $i = i'$ and $o$ and $o'$ are
        potentially conflicting (i.e., the analogy of
        $b$ and $a$ for consecutive transactions in a
        split schedule, but here defined for transaction templates).
\end{itemize}


The correctness of Algorithm~\ref{alg:ptime:template} now follows immediately from the following lemma:

\begin{lemma}\label{lemma:bmt:graph:characterization}
    Let $\workload$ be a set of \ptranss{}. Then, $\workload$ is not robust against \mvrc iff
    for some \ptrans{} $\tau_1 \in \workload$, $o_1, p_1 \in \tau_1$ and $i \in \{1,2\}$, a path in $\prefixfreegraphB(o_1, p_1, i, \tau_1, \workload)$ from a node $(\tau_2, p_2, 1, in)$ to a node $(\tau_m, o_m, i, out)$ exists with the following properties:
    \begin{itemize}
        \item $p_1$ is potentially conflicting with $o_m$;
        \item $o_1$ is potentially rw-conflicting with $p_2$; and
        \item $o_1 <_{\tau_1} p_1$ or $o_m$ is potentially rw-conflicting with $p_1$.
    \end{itemize}
\end{lemma}

\begin{proof}
    \emph{(if)} Let $P = (\tau_2, p_2, \ell_2, in), (\tau_2, o_2, k_2, out), (\tau_3, p_3, \ell_3, in), \ldots, (\tau_m, o_m, k_m, out)$ be the path in $\prefixfreegraphB(o_1, p_1,\allowbreak i, \tau_1, \workload)$, with $\ell_2 = 1$ and $k_m = i$. From this path $P$, we derive the sequence of potentially conflicting quadruples $D = (\tau_1, o_1, p_2, \tau_2), \ldots,\allowbreak (\tau_m, o_m, p_1, \tau_1)$.
    Note that for each such quadruple $(\tau_j, o_j, p_k, \tau_k)$ in $D$, the operations $o_j$ and $p_k$ are indeed potentially conflicting: if $j = 1$ or $j = m$, this is immediate by the additional conditions stated in Lemma~\ref{lemma:bmt:graph:characterization}. Otherwise, this follows from the fact that there can only be an edge from $(\tau_j, o_j, f_{o_j}, out)$ to $(\tau_k, p_k, f_{p_k}, in)$ if $o_j$ and $p_k$ are potentially conflicting.
    
    For each \shortptrans{} $\tau_j$ in $D$, we next define a variable assignment $\mu_j$ using four disjoint tuple mappings $c_1$, $c_2$, $c_3$ and $c_4$, thereby creating a variable mapping $\bar{\mu}$ for $D$. By construction of $\prefixfreegraphB(o_1, p_1, i, \tau_1, \workload)$, this $\bar{\mu}$ will actually coincide with the canonical mapping for $D$. We first define $\mu_1$:
    \begin{align*}
        \mu_1'(\vx) &= c_1(\type{\vx}) && \text{if $\vx$ is the variable occurring in $o_1$}, \\
        \mu_1'(\vx) &= c_2(\type{\vx}) && \text{if the variables occurring in $o_1$ and $p_1$ are different and not connected, and $\vx$ is the variable occurring in $p_1$}, \\
        \mu_1'(\vx) &= c_4(\type{\vx}) && \text{otherwise.} 
    \end{align*}
    For each $\tau_j$ different from $\tau_1$, the variable assignment $\mu_j$ is constructed as follows:
    \begin{align*}
        \mu_j(\vx) &= c_k(\type{\vx}) && \text{if $\vx$ is the variable occurring in $o_j$ and $(\tau_j, o_j, k, out)$ is a node in $P$}, \\
        \mu_j(\vx) &= c_\ell(\type{\vx}) && \text{if $\vx$ is the variable occurring in $p_j$ and $(\tau_j, p_j, \ell, in)$ is a node in $P$}, \\
        \mu_j(\vx) &= c_3(\type{\vx}) && \text{otherwise.}
    \end{align*}
    This variable assignment $\mu_j$ is well defined for each variable $\vx$, even if $\vx$ is the variable occurring in both $o_j$ and $p_j$. In this case, there can only be an edge from $(\tau_j, p_j, \ell, in)$ to $(\tau_j, o_j, k, out)$ if $k = \ell$. Notice furthermore that $c_3$ is never used in $\mu_1$, and $c_4$ is never used in a $\mu_j$ different from $\mu_1$, as $k, \ell \in \{1,2,3\}$ by construction of $P$.
    
    Let $\schedule$ be the multiversion split schedule based on $C = \bar{\mu}(D) = (\trans[1], b_1, a_2, \trans[2]), \ldots, (\trans[m], b_m, a_1, \trans[1])$. We argue that $\schedule$ satisfies all properties of Definition~\ref{def:mvsplitschedule}, thereby proving that $\workload$ is not robust against \mvrc. Towards a contradiction, assume Condition~(\ref{c:1}) is not true.
    That is, there is an operation $b_1' \in \prefix{\trans[1]}{b_1}$ ww-conflicting with an operation $b_j' \in \trans[j]$.
    By construction of $C$, the operations $o_1' \in \tau_1$ over a variable $\vx$ and $o_j' \in \tau_j$ over a variable $\vy$ corresponding to $b_1'$ and $b_j'$ are potentially ww-conflicting. The variable assignments $\mu_1$ and $\mu_j$ applied type mapping $c_k$ with $k \in \{1,2\}$ on both $\vx$ and $\vy$, as all four type mappings are disjoint and these are the only two type mappings occurring in both $\mu_1$ and $\mu_j$.
    Since we applied $c_1$ or $c_2$, the variable $\vy$ is occurring in $\tau_j$ in either $o_j$ or $p_j$ (or both). But then the corresponding node $(\tau_j, o_j, k, out)$ or $(\tau_j, p_j, k, in)$ cannot occur by construction of $\prefixfreegraphB(o_1, p_1, i, \tau_1, \workload)$, leading to the desired contradiction. Condition~(\ref{c:2}) and Condition~(\ref{c:3}) are immediate by the properties on $P$ specified in Lemma~\ref{lemma:bmt:graph:characterization}.
    
    \emph{(only if)}
    Assume $\workload$ is not robust against \mvrc.
    According to Lemma~\ref{lemma:noconstraints:limitedtuples}, there exists a multiversion split schedule $\schedule$ for some $C$ over a set of transactions $\transset$ consistent with $\workload$ and a database $\db$,
    where $C$ is induced by a sequence of potentially conflicting quadruples $D = (\tau_1, o_1, p_2, \tau_2)\ldots (\tau_m, o_m, p_1, \tau_m)$ over $\workload$ and it canonical variable mapping $\canmu$.
    We introduce a function $f$ mapping each operation in $D$ onto the corresponding type mapping used in the construction of $\schedule$.
    More formally, for each operation $o_j \in \tau_j$ over a variable $X$ appearing in $D$, we have $f(o_j) = i$ such that $\mu_j(X) = c_i(\type{X})$, with $\mu_j$ the corresponding variable mapping in $\canmu$.
    We now argue that the sequence of nodes $P = (\tau_2, p_2, f(p_2), in), (\tau_2, o_2, f(o_2), out), \ldots, (\tau_m, p_m, f(p_m), in), (\tau_m, o_m, f(o_m), out)$ is a valid path in $\prefixfreegraphB(o_1, p_1, i, \tau_1, \workload)$, where $i = 1$ if $p_1$ and $o_1$ are over the same variable in $\tau_1$, and $i = 2$ if not.
    
    We first argue that each node $(\tau_j, o_j, f(o_j), k)$ with $k \in \{in, out\}$ on this path $P$ is indeed a node in $\prefixfreegraphB(o_1, p_1, i, \tau_1, \workload)$.
    To this end, notice that $f(o_j) \in \{1,2,3\}$, as only $c_1$, $c_2$ and $c_3$ are used for operations occurring in $D$.
    If $f(o_j) = 1$, the node appears in the graph as long as there is no $o_j' \in \tau_j$ over the same variable as $o_j$ and potentially ww-conflicting with an operation in $\prefix{\tau_1}{o_1}$ over the same variable as $o_1$. Analogously, if $f(o_j) = 2$ and the operations $o_1$ and $p_1$ are not over the same variable in $\tau_1$, then the node appears in the graph if there is no $o_j' \in \tau_j$ over the same variable as $o_j$ and potentially ww-conflicting with an operation in $\prefix{\tau_1}{o_1}$ over the same variable as $p_1$.
    In both cases, the node not appearing in the graph would imply that the schedule $\schedule$ is not a valid multiversion split schedule, as Condition~(\ref{c:1}) in Definition~\ref{def:mvsplitschedule} would be violated.
    If $f(o_j) = 3$, then the node always occurs in the graph.
    
    We now argue that there is indeed an edge between each consecutive pair of nodes in $P$. For a pair $(\tau_j, p_j, f(p_j), in), (\tau_j, o_j, f(o_j), out)$, this follows trivially. For a pair $(\tau_j, o_j, f(o_j), out), (\tau_{j+1}, p_{j+1}, f(p_{j+1}), in)$, notice that $f(o_j) = f(p_{j+1})$, as otherwise $(\mu_j(\tau_j), o_j, p_{j+1}, \mu_{j+1}(\tau_{j+1}))$ would not be a conflict quadruple in $C$, where $\mu_j$ and $\mu_{j+1}$ are the corresponding variable mappings in $\canmu$.
    
    To conclude, we show that this path $P$ satisfies all required conditions. Since $\schedule$ is a multiversion split schedule, these conditions are immediate by Definition~\ref{def:mvsplitschedule}.
\end{proof}

It remains to argue that Algorithm~\ref{alg:ptime:template} indeed runs in time $\mathcal{O}(k^4.\ell)$, with $k$ the total number of operations in $\workload$ and $\ell$ the maximum number of operations in a \ptrans{} in $\workload$.
The two outer loops of Algorithm~\ref{alg:ptime:template} iterate over each pair of operations in the same template $\tau_1$ implying that the total number of iterations is ${O}(k.\ell)$. During each such iteration, the graph $G$ is constructed, containing at most $6k$ nodes. The transitive closure $TC$ over $G$ can therefore be computed in time ${O}(k^3)$ by an immediate application of the Floyd-Warshall algorithm. The last step of each iteration of the outer loops is to verify for each pair of operations $p_2$ and $o_m$ in $\workload$ whether a specific condition holds. As a result, this check can be verified in time ${O}(k^2)$.
By combining these results, we conclude that Algorithm~\ref{alg:ptime:template} indeed decides whether $\workload$ is robust against \mvrc in time $\mathcal{O}(k^4.\ell)$.

}
{
}

\end{document}